\documentclass[acmsmall,screen,authorversion]{acmart}
\pdfoutput=1

\setcopyright{none}

\usepackage{amsmath,mathtools,stmaryrd,color,alltt}
\usepackage{mathrsfs} \usepackage{bbm}
\usepackage{etex}

\usepackage{amsthm}
\usepackage[nameinlink]{cleveref}

\usepackage[all,2cell]{xy}
\UseAllTwocells
\SilentMatrices

\usepackage{bcprules}
\usepackage{paralist}

\newif\ifblind\blindfalse

\newif\iflong\longtrue

\newif\ifdraft\draftfalse

\newif\ifwithappendix\withappendixtrue

\usepackage{local,symbols}

\theoremstyle{plain}
\newtheorem{theorem}{Theorem}
\newtheorem{lemma}[theorem]{Lemma}
\newtheorem{corollary}[theorem]{Corollary}
\newtheorem{proposition}[theorem]{Proposition}
\newtheorem*{proposition*}{Proposition}
\newtheorem*{claim*}{Claim}

\theoremstyle{definition}
\newtheorem{definition}[theorem]{Definition}
\newtheorem{example}[theorem]{Example}
\newtheorem*{problem*}{Problem}

\theoremstyle{remark}
\newtheorem{remark}[theorem]{Remark}
\newtheorem{notation}[theorem]{Notation}

\crefname{theorem}{Theorem}{Theorems}
\crefname{lemma}{Lemma}{Lemmas}
\crefname{corollary}{Corollary}{Corollaries}
\crefname{proposition}{Proposition}{Propositions}
\crefname{definition}{Definition}{Definitions}
\crefname{remark}{Remark}{Remarks}
\crefname{notation}{Notation}{Notations}

\crefname{section}{Section}{Sections}

\begin{document}

\title{Enriched Presheaf Model of Quantum FPC}

\author{Takeshi Tsukada}

\orcid{0000-0002-2824-8708}             \affiliation{
  \institution{Chiba University}            \country{Japan}                    }
\email{tsukada@math.s.chiba-u.ac.jp}          

\author{Kazuyuki Asada}

\orcid{0000-0001-8782-2119}             \affiliation{
  \institution{Tohoku University}           \country{Japan}
}
\email{asada@riec.tohoku.ac.jp}         

\begin{abstract}
  Selinger gave a superoperator model of a first-order quantum programming language and proved that it is fully definable and hence fully abstract.
  This paper proposes an extension of the superoperator model to higher-order programs based on modules over superoperators or, equivalently, enriched presheaves over the category of superoperators.
  The enriched presheaf category can be easily proved to be a model of intuitionistic linear logic with cofree exponential, from which one can cave out a model of classical linear logic by a kind of bi-orthogonality construction.
  Although the structures of an enriched presheaf category are usually rather complex,
  a morphism in the classical model can be expressed simply as a matrix of completely positive maps.

  The model inherits many desirable properties from the superoperator model.
  A conceptually interesting property is that our model has only a state whose ``total probability'' is bounded by \( 1 \), \ie~does not have a state where true and false each occur with probability \(2/3\).
  Another convenient property inherited from the superoperator model is a \( \omega \)CPO-enrichment.
  Remarkably, our model has a sufficient structure to interpret arbitrary recursive types by the standard domain theoretic technique.
  We introduce \emph{Quantum FPC}, a quantum \( \lambda \)-calculus with recursive types, and prove that our model is a fully abstract model of Quantum FPC.
\end{abstract}



\maketitle

\section{Introduction}
Quantum computation is a paradigm that makes use of the principles of quantum mechanics to perform computations.
Quantum mechanics can only predict outcomes probabilistically, so quantum computation is \emph{stochastic}.
It is not only stochastic;
\emph{quantum entanglement} is expected to allow us to perform efficient computations that could not be achieved by stochastic computations.

A first-order programming language for quantum computation was studied by \citet{Selinger2004}.
He gave a fully definable model of the language using a category of \emph{superoperators}.
Subsequent work by \citet{Selinger2008} provided a model of a higher-order linear quantum programming language based on the category \( \CPM \) of \emph{completely positive maps}.

These models have an important qualitative difference.
Recall that quantum computation is stochastic.
An important invariant of stochastic computation is the ``total probability'', which must be of course bounded by \( 1 \).
That means, we never reach a state where true and false each occur with probability \( 2/3 \).
A superoperator is defined as a completely positive map that preserves the ``total probability'' of a state, so a superoperator preserves the ``total probability'', but a completely positive map does not in general.

The superoperator model captures the invariant ``total probability \(\le 1\)'' but it is applicable only to first-order computations.
Here is a long-standing open problem:
\begin{description}
    \item[Problem.]
        How to capture the invariant ``total probability \(\le 1 \)'' for higher-order types
        \\
        \hphantom{Problem}(without relying on intentional structures of programs).
\end{description}

\citet{Selinger2004a} addressed this problem, introducing \emph{normed cones}.
The ``total probability'' in the superoperator model is a norm known as the \emph{trace norm}, and Selinger tried to capture the ``total probability'' in higher types by using norms different from the trace norm.
This approach seems fairly natural, but unfortunately, it did not succeed.
The tensor product in his model does not appropriately handle the quantum entanglement, which is a stronger correlation than classical correlations (\ie~probabilistic correlations).

The subsequent development of the semantics of quantum programs was not oriented toward capturing the ``total probability'' in higher types, but rather the opposite.
The previously mentioned model by completely positive maps~\cite{Selinger2008}, which was published \( 4 \) years after Selinger's work~\cite{Selinger2004a} on normed corns, succeeded to model a linear higher-order quantum computation by ignoring the ``total probability'', or equivalently, by allowing the ``total probability'' to be any finite non-negative real.
\citet{Pagani2014} studied a further extension based on the completion of completely positive maps by \( \infty \)'s, in order to give a model of a quantum calculus with the linear exponential modality \( ! \) and term-level recursion.

Interestingly \citet[Proposition~43 and Section~7]{Pagani2014} observed that the denotation of a program cannot be \( \infty \) and that adding a constant with ``probability \( 2 \)'' breaks this property.
Hence ``total probability \(\le 1\)'' is an important property even though their model has \( \infty \)'s.

The primary aim of this paper is to solve the problem of the ``total probability \( \le 1 \)'' invariant, using a technique of qualitative model construction for linear logic.
A pleasant surprise is that the resulting model is quite powerful, as it provides a fully abstract model of \emph{Quantum FPC}, a quantum \( \lambda \)-calculus with arbitrary recursive types.
See \cref{table:comparison-of-models} for a brief comparison of the features of quantum calculi and their models.
\begin{table}[b]
    \caption{Models of quantum programming languages, their properties and their target languages.  Here ``classical'' means ``non-quantum''; do not confuse it with classical (linear) logic.  ``Entangled tensor'' means the tensor product \( A \otimes B \) whose values are pair of possibly \emph{entangled} values of \( A \) and \( B \).  ``Quantum higher-order'' means the right adjoint to an entangled tensor.  ``Quantum recursive type'' is a recursive type for quantum data types.  The blank of a language feature means that the feature is not dealt with the paper, not necessarily mean that the model is unable to handle the feature.}
    \label{table:comparison-of-models}
    \vspace{-2.5ex}
    \small
\begin{tabular}{r|ccccccc|c|}
        & & & & \hphantom{$\mathbf{PER}_{\mathcal{Q}}$${}^{(4)}$} & \hphantom{$\overline{\mathbf{CPMs}}^{\oplus}$${}^{(5)}$} & \hphantom{${\sim}\textbf{-QA}$${}^{(6)}$} & \hphantom{\texttt{VQPL}${}^{(7)}$\quad} & ours \\
        & $\mathbf{Q}$\rlap{${}^{(1)}$} & $\mathbf{Q'}$\rlap{${}^{(2)}$} & $\mathbf{CPM}$\rlap{${}^{(3)}$} & $\mathbf{PER}_{\mathcal{Q}}$\rlap{${}^{(4)}$} & $\overline{\mathbf{CPMs}}^{\oplus}$\rlap{${}^{(5)}$} & ${\sim}\textbf{-QA}$\rlap{${}^{(6)}$} & \texttt{VQPL}\rlap{${}^{(7)}$} & $\classicalCQ$
        \\ \hline \hline
        ``Total probability'' \( \le 1 \) & \checkmark & \checkmark & & & & \checkmark & \checkmark & \checkmark
        \\
        ``Total probability'' \( < \infty \) & \checkmark & \checkmark & \checkmark & & & \checkmark & \checkmark & \checkmark
        \\
        Full abstraction & \checkmark & & \checkmark & & \checkmark\rlap{${}^{(8)}$} & \checkmark\rlap{${}^{(8)}$} & & \checkmark
        \\ \hline
        First-order quantum & \checkmark & \checkmark & \checkmark & \checkmark & \checkmark & \checkmark & \checkmark & \checkmark
        \\
        Entangled tensor & \checkmark & & \checkmark & & \checkmark & \checkmark & \checkmark & \checkmark
        \\
        Classical higher-order & & \checkmark & \checkmark & \checkmark & \checkmark & \checkmark & \checkmark & \checkmark
        \\
        Quantum higher-order & & & \checkmark & & \checkmark & \checkmark & & \checkmark
        \\
        Linear exponential (\(!\)) & & & & \checkmark & \checkmark & \checkmark & \checkmark & \checkmark
        \\
        Term-level recursion & (\checkmark)\rlap{${}^{*1}$} & \checkmark & & \checkmark & \checkmark & \checkmark & \checkmark & \checkmark
        \\
        Algebraic data type & & & & & (\checkmark)\rlap{${}^{*2}$} & (\checkmark)\rlap{${}^{*2}$} & \checkmark & \checkmark
        \\
        Classical recursive type & & & & & & & \checkmark & \checkmark
        \\
        Quantum recursive type & & & & & & & & \checkmark
    \end{tabular}
\\[4pt]
        (1) \cite{Selinger2004},
        (2) \cite{Selinger2004a},
        (3) \cite{Selinger2008},
        (4) \cite{Hasuo2011},
    \\[2pt]
        (5) \cite{Pagani2014},
        (6) \cite{Clairambault2019a},
        (7) \cite{Jia2022},
        (8) \cite{Clairambault2020}.
    \\[2pt]
        \begin{description}
\item[*1] Only first-order recursions, namely while-loops, are supported.
            \item[*2] Only list types are discussed (although we expect that initial algebras of polynomial functors exist).
\end{description}
\end{table}

\subsection{Approach: Modules over Superoperators}
Our approach to the ``total probability \(\le 1\)'' invariant is based on \emph{modules over superoperators}.
The construction is motivated by the \emph{probabilistic coherence space} model \( \PCoh \)~\cite{Girard2004,Danos2011,Ehrhard2014,Ehrhard2018}, which is a fully abstract model for higher-order probabilistic programs.
The total probability is an invariant of this model in the sense that \( \PCoh(I,I) = [0,1] \) (\cf~\cite[Remark in Section~2.1.2]{Danos2011}).
It is natural to expect that a quantum extension of the probabilistic coherence space model has desired properties.

However, it is not straightforward to give an appropriate quantum extension of (probabilistic) coherence spaces.
Actually, a quantum extension has been mentioned by \citet{Girard2004}, named \emph{quantum coherence spaces}, but \citet[Section~5]{Selinger2004a} observed that it does not precisely capture the first-order quantum computation.
Here again, the issue is on the quantum entanglement.

Our development is based on a recent algebraic approach to models of linear logic, proposed by \citet{Tsukada2022}.
They characterise the probabilistic coherence space model as a category of \emph{modules} over \( [0,1] \).
A module \( M \) over a ring \( R \) in the standard sense is an abelian group together with an action \( x \cdot r \in M \) of \( r \in R \) to \( x \in M \).
Tsukada and Asada studied a variant of the module theory, obtained by replacing abelian groups with \emph{\(\Sigma\)-monoids}~\cite{Haghverdi2000,Hoshino2011,Tsukada2018,Tsukada2022}, an algebra with partially-defined countable sum.

We regard the superoperator model \( \CQ \) as the coefficient ``ring''.
It has a hom-set-wise sum, which indeed satisfies the axioms for \( \Sigma \)-monoids.
The composition of morphisms is regarded as the multiplication.
These structures satisfy basic laws, such as the associativity of multiplication and the distributive law.
In this sense, the category \( \CQ \) of superoperators has a ring-like structure.

The concept of a module over superoperators can then be defined as a natural extension of the standard definition of a module over a ring.
We believe the definition should be acceptable, but there is another justification from the viewpoint of \emph{enriched category theory}: A module in the standard sense is just an \emph{abelian-enriched presheaf} (see, \eg, \cite[Section~1.3]{BorceuxVolume2}), whereas a module in our sense is a \emph{\(\Sigma\)-monoid-enriched presheaf}.
Hence, writing the category of \( \Sigma \)-monoids as \( \SMon \), the category of modules over superoperators is defined as \( \pshCQ \defe [\CQ^{\op}, \SMon]_{\SMon} \), where \( [\mathcal{C},\mathcal{D}]_{\SMon} \) means the category of \( \SMon \)-enriched functors from \( \mathcal{C} \) to \( \mathcal{D} \).
This alternative perspective enables us to utilise a number of constructions and results for (enriched) presheaf categories.

Although the enriched presheaf category \( \pshCQ \) has many desired properties, such as (co)completeness and local presentability, it contains many ``wild'' objects that are hard to deal with.
This fact may make sense if one considers the (standard, \( \Set \)-enriched) presheaf category \( \widehat{\mathcal{C}} \) of a non-trivial category \( \mathcal{C} \), but it can also be understood by an analogy from the theory of modules over a ring: unlike vector spaces, a module over a ring may not have a basis, for example.
So it would be desirable to find an easy-to-handle subcategory of \( \pshCQ \), consisting of ``tame'' objects.

Here we appeal to a way to extract a model of classical linear logic given by \citet{Tsukada2022}, an intrinsic variant of the \emph{bi-orthogonality} construction (see, e.g., \cite{Hyland2003}).
For a given \emph{\( \Sigma \)-semiring} \( \mathcal{R} \) (i.e.~a one-object \( \SMon \)-enriched category), Tsukada and Asada considered a full subcategory of the enriched presheaf category \( \widehat{\mathcal{R}} \) consisting of modules \( M \) that has a \emph{basis} and is canonically isomorphic to its second dual \( \neg\neg M \).
This paper introduces a multi-object version of this construction and applies it to \( \pshCQ \), yielding a model \( \classicalCQ \) of classical linear logic.

Every module \( M \) in \( \classicalCQ \) has many desirable properties: (1) a basis allows us to represent an element in \( M \) as a vector with entries from completely positive maps, and correspondingly, to represent a morphism \( f \colon M \longrightarrow N \) in \( \classicalCQ \) as a matrix; (2) \( M \) is canonically an \( \omega \)CPO, \( \classicalCQ \) is \( \omega \)CPO-enriched, and all linear logic constructs including \( ! \) are \( \omega \)CPO-enriched functors.
The first point significantly eases the calculation, and the second point allows us to use a standard technique to solve a type-level recursive equation based on embedding-projection pairs in an $\omega$CPO-enriched category.

\subsection{Quantum FPC}
We demonstrate the usefulness of \( \classicalCQ \) as a model of programming languages, proving the full abstraction for \emph{Quantum FPC}, a quantum \( \lambda \)-calculus~\cite{Pagani2014,Selinger2009} with arbitrary recursive types.
The language Quantum FPC is obtained as a simple combination of linearity of the linear \( \lambda \)-calculus, recursive type from FPC~\cite{Plotkin1985,Fiore1994a,Fiore1994}, and the quantum feature from quantum \( \lambda \)-calculi~\cite{Selinger2009,Pagani2014}, each of which is quite well-understood.

Although Quantum FPC is a simple combination and there is no novelty in the language design, a denotational model for the simple combination is one of the major contributions of this paper.
\citet{Jia2022} studied a quantum calculus with recursive types, but their language has a strong restriction at least from the viewpoint of the quantum \( \lambda \)-calculus~\cite{Pagani2014,Selinger2009}: only first-order quantum functions are allowed and quantum functions cannot be entangled with each other.
The restriction separates the first-order quantum language from the higher-order classical (non-quantum) language, and this separation is explained as the key to extending the classical part by advanced features such as recursive types~\cite[Section~8.2]{Jia2022}.
We show that the separation is not necessary, at least for supporting recursive types, by providing the denotational model of Quantum FPC, a language without the separation.

\subsection{Contributions}
The contributions of this paper can be summarised as follows.
\begin{itemize}
    \item We introduce the category \( \pshCQ \) of modules over superoperators, which is a model of intuitionistic linear logic, and cave out a model \( \classicalCQ \) of classical linear logic.
    \item We show that \( \classicalCQ \) inherits many desirable properties of \( \CQ \) such as the ``total probability \(\le 1\)'' invariant and \( \omega \)CPO enrichment.
    The \( \omega \)CPO enrichment allows us to apply the standard domain-theoretic approach to recursive types.
    \item We define a quantum \( \lambda \)-calculus with recursive types, named Quantum FPC, and prove that \( \classicalCQ \) is an adequate and fully abstract model of Quantum FPC.
    \item We revisit the norm-based approach~\cite{Selinger2004a} to quantum programs from the viewpoint of superoperator modules and propose an alternative based on \emph{families of norms}.
\end{itemize}
Our model is similar to the model by \citet{Pagani2014} in many respects.  The biggest difference is that we do not apply the completion to the hom-set-wise sum of \( \CQ \).
We do not add \( \infty \)'s, but at the cost of having to be careful about the convergence of sums.
The absence of \( \infty \)'s plays an important role in the second and third points above.
See \cref{sec:related-work} for more details.

\paragraph{Organisation of the paper.}
\cref{sec:qfpc} introduces our target calculus, \emph{Quanutm FPC}.
\cref{sec:model} discusses the enriched presheaf category \( \pshCQ \) over the \( \SMon \)-enriched category \( \CQ \) of superoperators.
\cref{sec:matrix} introduces the concept of basis, defines the classical model \( \classicalCQ \), and studies its properties extensively using bases.
\cref{sec:recursive} solves type-level recursive equations and \cref{sec:interpretation} defines the interpretation and proves soundness and adequacy.
\cref{sec:full-abstraction} proves the full abstraction (\cref{thm:fullabstraction}).
In \cref{sec:norm} we revisit the norm-based approach to quantum programs.

 \section{Quantum FPC}
\label{sec:qfpc}
This section defines the target language of this paper, \emph{Quantum FPC}.
It is a simple combination of two well-known languages, namely FPC~\cite{Plotkin1985,Fiore1994a,Fiore1994} and the quantum \( \lambda \)-calculus~\cite{Pagani2014,Selinger2009}, and we believe there is nothing hard to understand.
The operational semantics of Quantum FPC is based on that of the quantum \( \lambda \)-calculus by \citet{Pagani2014}.
Our formalisation uses a model of first-order quantum computation, namely the category \( \CQ \) of superoperators, which we briefly review in \cref{sec:defining-completely-positive-map}.

\subsection{Preliminaries: Completely Positive Maps and Superoperators}
\label{sec:defining-completely-positive-map}
Here we quite briefly review notions related to quantum computation, some of which are not even defined here.
See, \eg, \cite{Selinger2004,Selinger2004a} for a gentle introduction and concrete definitions.

The state space of a quantum system is usually regarded as a Hilbert space.
This paper mainly focuses on finite-dimensional Hilbert spaces \( \Complex^n \), \( n \in \Nat \).
Let \( \Mat_n(\Complex) \) be the set of \( (n \times n) \)-matrices over \( \Complex \).
A matrix \( x \in \Mat_n(\Complex) \) is \emph{self-adjoint} or \emph{Hermitian} if \( x^* = x \), where \( x^* \) is the conjugate transpose of \( x \).
A self-adjoint \( x \in \Mat_n(\Complex) \) is \emph{positive} if \( v^* x v \in \Real_{\ge 0} \) for every vector \( v \in \Complex^n \).
We write \( 0 \le x \) to mean that \( x \) is a positive self-adjoint matrix.
The \emph{trace} \( (\trace x) \) of \( x \in \Mat_n(\Complex) \) is the sum of the diagonal elements of \( x \).
For positive \( x \), \( \trace x \in \Real_{\ge 0} \).

A \emph{(mixed) state} of a quantum system is represented as a positive self-adjoint \( (n \times n) \)-matrix \( x \in \Mat_n(\Complex) \) of trace \( \le 1 \).
The trace of a positive self-adjoint operator intuitively means the ``total probability'', so a state is usually defined as a positive self-adjoint matrix of trace \( 1 \).  However, in the context of quantum programming, the ``total probability'' of the result of a computation may be strictly smaller than \( 1 \) when the computation may diverge~\cite[Section~5.1]{Selinger2004}.
For this reason, the trace \( (\trace\,x) \) of a state \( x \) is in \( [0,1] \) in this paper.

A quantum operation is represented by a \emph{superoperator}, a \( \Complex \)-linear function \( \varphi \colon \Mat_n(\Complex) \longrightarrow \Mat_m(\Complex) \) that maps a mixed state to a mixed state and satisfies an additional condition.
The preservation of mixed states is equivalent to that \( \varphi \) is \emph{positive} (\ie~\( 0 \le x \) implies \( 0 \le \varphi(x) \)) and \emph{trace-non-increasing} (\ie~\( \trace x \ge \trace \varphi(x) \) for every \( x \ge 0 \)).
The additional condition strengthens the former, the positivity, and is called the \emph{complete positivity}.
We omit its definition.
A linear map satisfying the complete positivity is called a \emph{completely positive map}.
In this terminology, a superoperator is a completely positive map that is trace-non-increasing.

The \emph{operator norm} \( \opnorm{\varphi} \) (with respect to the trace norm) of a completely positive map \( \varphi \) is \( \sup \{ \trace \varphi(x) \mid x \ge 0, \trace x = 1 \} \).
Then \( \varphi \) is trace-non-increasing if and only if \( \opnorm{\varphi} \le 1 \).

The categories \( \CPM \) of completely positive maps and \( \CQ \) of superoperators are defined as follows.
Their object is a natural number \( n \in \Nat \).
A morphism \( \varphi \in \CPM(n,m) \) in \( \CPM \) is a completely positive map \( \varphi \colon \Mat_n(\Complex) \longrightarrow \Mat_m(\Complex) \), and a morphism \( \psi \in \CQ(n,m) \) in \( \CQ \) is a superoperator \( \psi \colon \Mat_n(\Complex) \longrightarrow \Mat_m(\Complex) \).
Note that \( \CPM(1,1) = \Real_{\ge 0} \) and \( \CQ(1,1) = [0,1] \).
Morphisms are composed as functions, and the identity is the identity function.
Note that \( \CQ \) is a \emph{wide} subcategory of \( \CPM \) (\ie~the sets of objects of \( \CQ \) and \( \CPM \) coincide).
We shall identify a mixed state \( x \in \Mat_n(\Complex) \) with the morphism \( r \mapsto r\,x \) in \( \CQ(1,n) \).

The algebraic property of completely positive maps is something like positive reals.
The set of completely positive maps is closed under the sum (as linear functions) and the scalar multiplication of positive reals,
and a rearrangement of a convergent series \( \sum_{i = 0}^\infty f_i \) of completely positive maps does not change the value.
Similarly the set of superoperators behaves like the unit interval \( [0,1] \).

These categories are symmetric monoidal.
The tensor product is \( n \otimes m \defe nm \) on objects, with the unit object \( 1 \).
As \( \Complex \)-vector spaces, \( \Mat_n(\Complex) \otimes \Mat_m(\Complex) \cong \Mat_{n \otimes m}(\Complex) \) and this isomorphism induces the action of the tensor product on morphisms.
By appropriately choosing the above isomorphism, the monoidal structures on \( \CPM \) and \( \CQ \) can be made strict.

Furthermore \( \CPM \) is compact closed.
The dual \( n^* \) of \( n \) is \( n \), and the compact closed structure can be chosen to be strict (in the sense of \cite{Kelly1980}).
In particular, the canonical morphism \( n \cong n^{**} \) is the identity.

A state \( x \in \CQ(1, n \otimes m) \) is \emph{entangled} if \( x \) cannot be written as \( \sum_{i = 1}^k p_i (y_i \otimes z_i) \) for \( k \in \Nat \), \( p_i \in [0,1] \) with \( \sum_{i=1}^k p_i \le 1 \), \( y_i \in \CQ(1,n) \) and \( z_i \in \CQ(1,m) \).
Entangled states are critically important in quantum computation but, at the same time, the biggest obstacle in constructing a denotational model.

\subsection{Types and Terms of Quantum FPC}
\Cref{fig:syntax-of-types-and-terms} gives the syntax of types and terms of Quantum FPC.
It is a simple combination of the linear \( \lambda \)-calculus (\( \TyUnit, A \multimap B, A \otimes B, A + B \) and \( {!}A \) types), recursive type constructs (\( X \) and \( \mu X. A \) types) and quantum constructs (\( \TyQubit \) type), so the meaning of types and terms should be self-explanatory.
Here \( U \) is a unitary operator on the Hilbert space \( \Complex^{2^n} \) for some \( n \) and regarded as a constant \( U \colon \TyQubit^n \multimap \TyQubit^n \).
We often omit the type annotations such as \( A+B \) in \( \TInl_{A+B} \).
The notion of free (type or term) variables are defined as usual.
We allow tacit renaming of bound variables.
We write \( A[B/X] \) and \( t[u/x] \) for the capture-avoiding substitution.
\begin{figure}[t]
    \begin{align*}
        & \mathit{Types} & A,B &{}::={} X \mid \TyQubit \mid \TyUnit \mid A \multimap B \mid A \otimes B \mid A + B \mid {!}A \mid \mu X. A \\
        & \mathit{Terms} & s,t &{}::={} x \mid \lambda x^A. t \mid s\,t & (\mbox{linear function}) \\
        & & & {}\:\:\mid\:\, \TUnitVal \mid (t;s) & (\mbox{unit}) \\
        & & & {}\:\:\mid\:\, t \otimes u \mid \TLet{x \otimes y}{t}{u} & (\mbox{tensor}) \\
        & & & {}\:\:\mid\:\, \TInl_{A+B} \mid \TInr_{A+B} \mid \TMatch{t}{x}{s}{x'}{s'} & (\mbox{coproduct}) \\
        & & & {}\:\:\mid\:\, \TBang{t} \mid \TDer_{!A} & (\mbox{linear exponential}) \\
        & & & {}\:\:\mid\:\, \TFold_{\mu X.A} \mid \TUnfold_{\mu X.A} & (\mbox{recursive type}) \\
        & & & {}\:\:\mid\:\, \qket{0} \mid U \mid \QMeas & (\mbox{quantum}) \\
        & \mathit{Values} & v,w &{}::={} x \mid \lambda x^A.t \mid \TUnitVal \mid v \otimes w \mid \TInl_{A+B}\,v \mid \TInr_{A+B}\,v \mid \TBang{t} \mid \TFold_{\mu X.A}\,v \hspace{-120pt} \\
        & & & {}\:\:\mid\:\, \TInl_{A+B} \mid \TInr_{A+B} \mid \TDer_{!A} \mid \TFold_{\mu X.A} \mid \TUnfold_{\mu X.A} \mid U \mid \QMeas \hspace{-120pt}
    \end{align*}
    \vspace{-20pt}
    \caption{The syntax of Quantum FPC}
    \label{fig:syntax-of-types-and-terms}
\end{figure}

The following constructs, which are not given as primitives of Quantum FPC, are definable.
\begin{itemize}
    \item \textbf{Term-level recursion}:
        Term-level recursion is definable by using type-level recursion by a well-known technique~\cite{Abadi1996,Fiore1994a} (see~\cite{Lindenhovius2021} for the linear setting).
        So
\( (\mathtt{fix}\,f^{!(A \multimap B)}. \TBang{\lambda x^{A}. t}) \colon !(A \multimap B) \) is obtained as a derived form.
The typing rule is in \citet{Pagani2014}, agreeing with \cite{Lindenhovius2021}.
    \item \textbf{Natural number types}:
        The natural number type can be defined as \( \mathbb{N} \defe \mu X. (\TyUnit + X) \), which our model actually interprets as the initial algebra of the functor \( F(X) \defe \TyUnit + X \) as expected.
        Initial algebra types for polynomial functors (composed of \( + \), \( \otimes \) and \( ! \)),
such as the list type \( (\texttt{list}\,A) \defe \mu X. (\TyUnit + A \otimes X) \), can be obtained in a similar way.
\end{itemize}

The type system is that of the standard linear \( \lambda \)-calculus.
A \emph{type environment} \( \Gamma \) is a list of type bindings of the form \( x \colon A \).
We write \( !\Gamma \) for a type environment consisting only of \( ! \)-types \( !A \).
A \emph{type judgement} is of the form \( !\Gamma, \Delta \vdash t \colon A \)
where two different bindings in the environment never share the same variable,
and type variables must not freely occur.
The typing rules are standard (see, \eg, \cite{Lindenhovius2021}):
for example,
\begin{equation*}
    \dfrac{
        !\Gamma \vdash t \colon A
    }{
        !\Gamma \vdash \TBang{t} \colon {!}A
    },
    \quad
    \dfrac{
        !\Gamma, \Delta \vdash t \colon \TyUnit
        \quad
        !\Gamma, \Delta' \vdash s \colon A
    }{
        !\Gamma, \Delta, \Delta' \vdash (t;s) \colon A
    }
    \:\:\mbox{and}\:\:
    \dfrac{
        !\Gamma, \Delta \vdash t \colon A_1 \otimes A_2
        \quad
        !\Gamma, \Delta', x_1 \colon A_1, x_2 \colon A_2 \vdash s \colon B
    }{
        !\Gamma, \Delta, \Delta' \vdash \TLet{x\otimes y}{t}{s} \colon B
    }.
\end{equation*}
Note that all term constructors for recursive type and for quantum computation are introduced as constants.
The types of constants are given by
\begin{gather*}
    \TUnitVal \colon \TyUnit
    \qquad
    \TInl_{A+B} \colon A \multimap (A+B)
    \qquad
    \TInr_{A+B} \colon B \multimap (A+B)
    \qquad
    \TDer_{!A} \colon {!}A \multimap A
    \\
    \TFold_{\mu X.A} \colon (A[\mu X.A/X]) \multimap \mu X.A
    \qquad
    \TUnfold_{\mu X.A} \colon \mu X.A \multimap (A[\mu X.A/X])
    \\
    \ket{0} \colon \TyQubit
    \qquad
    U \colon \TyQubit^n \multimap \TyQubit^n
    \qquad
    \QMeas \colon \TyQubit \multimap (\TyUnit + \TyUnit)
\end{gather*}
where \( U \) is a unitary operator on \( \Complex^{2^n} \) and \( \TyQubit^n \) is the tensor product of \( n \) \( \TyQubit \)'s.
\begin{remark}
    The syntax of terms slightly differs from \citet{Pagani2014} on the primitives for the linear exponential.
    In their calculus, the promotion is implicit (\ie~their calculus has no term constructor for the promotion) and allowed only for values of function types:
    \begin{equation*}
        \dfrac{!\Gamma \vdash v \colon (A \multimap B)}{!\Gamma \vdash v \colon !(A \multimap B)}.
    \end{equation*}
    In contrast, our calculus has a term constructor \( \TBang{} \) for the promotion, which can take an arbitrary term of any type; the computation of \( t \) in \( \TBang{t} \) is suspended until the \( \TDer \) constructor invokes it (intuitively \( \TDer\,\TBang{t} \) is reduced to \( t \)).
    Our constructor \( \TBang{t} \) corresponds to \( \lambda z^\TyUnit. (z;t) \) in \citet{Pagani2014}, where \( z \) is a variable not free in \( t \), and \( \TDer \) corresponds to the application of the unit value.
    So, this syntactic difference is just a matter of preference.
    The syntax in our style can be found in \citet{Lindenhovius2021}, for example.
\eqed
\end{remark}

\subsection{Operational Semantics}
The operational semantics is similar to that by \citet{Pagani2014}.
It is given as a rewriting system of \emph{quantum closures}.
\begin{definition}
    A \emph{quantum closure} is a tuple \( \QCloE{\vec{q}}{\varphi}{(c_i^{{!A}_i} = \TBang{s_i})_{i=1}^m}{t} \) where
    \begin{itemize}
        \item \( \vec{q} = (q_1\dots q_k) \) is a list of variables of \( \TyQubit \) type,
        \item \( \varphi \in \CQ(1, 2^k) \) is a mixed state of \( k \)-qubits,
        \item \( (c_i^{{!A}_i} = \TBang{s_i})_{i=1}^m \) is a list of assignments to a variable \( c_i \) of a closed \(!\)-term \( {} \vdash \TBang{s} \colon {!}A_i \), and
        \item \( q_1 \colon \TyQubit, \dots, q_k \colon \TyQubit, c_1 \colon {!}A_1, \dots, c_m \colon {!}A_m \vdash t \colon \TyUnit \) is a term.
    \end{itemize}
    We use the metavariable \( \Xi \) for the \( ! \)-terms component,
and \( \QCloS \) for quantum closures.
    We identify closures obtained by rearrangement of \( \vec{q} \), \ie~\( \QCloE{q_1\dots q_n}{\varphi}{\Xi}{t} = \QCloE{q_{\sigma^{-1}(1)}\dots q_{\sigma^{-1}(n)}}{\sigma \circ \varphi}{\Xi}{t} \) for every permutation \( \sigma \in \mathfrak{S}_n \).
    When \( \vec{q} \) or \( \Xi \) is the empty list, we write \( \langle\rangle \).
A \emph{quantum value closure} (ranged over by \( \QCloValS \)) is a quantum closure such that \( t \) is a value.
    So a quantum value closure must be of the form \( \QCloE{\langle\rangle}{\varphi}{(c_i^{{!A}_i} = \TBang{s_i})_{i=1}^m}{\TUnitVal} \) where \( \varphi \in \CQ(1,1) = [0,1] \), and only meaningful information is \( \varphi \) denoting the probability to reach this state.
\eqed
\end{definition}

\begin{remark}
    The definition above differs slightly from the original definition \cite{Selinger2009,Pagani2014} in the following respects.
    First the original definition represents the state of qubits as a vector \( v \in \Complex^{2^k} \) instead of a mixed state \( \varphi \in \CQ(1, 2^k) \).
    We use a mixed state because it allows us to describe a measurement process as a composition.
    Second our definition has an additional component \( (c_i^{{!A}_i} = \TBang{s_i})_{i=1}^m \) collecting all \( ! \)-expressions evaluated during the computation, whereas the original definition uses a substitution \( t[\TBang{s_1}/c_1,\dots,\TBang{s_m}/c_m] \).
    This is nothing but \emph{explicit substitution} and because of a technical convenience for our adequacy proof, in which we need to finely control duplications.
\eqed
\end{remark}

The \emph{evaluation context} is given by the following grammar:
\begin{align*}
    E&\quad::=\quad [] \mid E\,t \mid v\,E \mid (E;t) \mid E \otimes t \mid v \otimes E \mid (\TLet{x \otimes y}{E}{t})
    \\ &
    \quad\quad|\quad \TMatch{E}{x}{s}{x'}{s'}.
\end{align*}
Let \( \mathit{qCl} \) be the set of quantum closures and \( \mathcal{M}_{\mathrm{fin}}(X) \) be the set of finite multisets over \( X \).
Then the \emph{one-step reduction relation} is a subset \( (\rightsquigarrow) \subseteq (\mathit{qCl}) \times \mathcal{M}_{\mathrm{fin}}(\mathit{qCl}) \), relating a quantum closure with a formal sum of quantum closures.
It is defined by the rules in \cref{fig:quantum-one-step-reduction}, where we write \( x \rightsquigarrow y_1 + y_2 + \dots + y_k \) to mean \( (x, [y_1,\dots,y_k]) \in (\rightsquigarrow) \).
In most cases, the number of elements after the reduction is \( 1 \); the exception is the rule for the measurement \( \QMeas \), which has two possibilities (i.e., two elements).
\begin{figure}
    \begin{gather*}
        \QCloE{\vec{q}}{\varphi}{\Xi}{E[(\lambda x.t)\,v]} \rightsquigarrow \QCloE{\vec{q}}{\varphi}{\Xi}{E[t[v/x]]}
        \qquad
        \QCloE{\vec{q}}{\varphi}{\Xi}{E[\TUnitVal; t]} \rightsquigarrow \QCloE{\vec{q}}{\varphi}{\Xi}{E[t]}
        \\
        \QCloE{\vec{q}}{\varphi}{\Xi}{E[\TLet{x \otimes y}{v \otimes w}{t}]} \rightsquigarrow \QCloE{\vec{q}}{\varphi}{\Xi}{E[t[v/x, w/y]]}
        \\
        \QCloE{\vec{q}}{\varphi}{\Xi}{E[\TMatch{\TInl\,v}{x}{t}{y}{s}]} \rightsquigarrow \QCloE{\vec{q}}{\varphi}{\Xi}{E[t[v/x]]}
        \\
        \QCloE{\vec{q}}{\varphi}{\Xi}{E[\TMatch{\TInr\,v}{x}{t}{y}{s}]} \rightsquigarrow \QCloE{\vec{q}}{\varphi}{\Xi}{E[s[v/y]]}
        \\
        \QCloE{\vec{q}}{\varphi}{\Xi}{E[\TBang{t}]} \rightsquigarrow \QCloE{\vec{q}}{\varphi}{(\Xi, c=\TBang{t})}{E[c]}
        \\
        \QCloE{\vec{q}}{\varphi}{(\Xi, c=\TBang{t}, \Xi')}{E[\TDer\,c]} \rightsquigarrow \QCloE{\vec{q}}{\varphi}{(\Xi, c=\TBang{t}, \Xi')}{E[t]}
        \\
        \QCloE{\vec{q}}{\varphi}{\Xi}{E[\TUnfold\,(\TFold\,v)]} \rightsquigarrow \QCloE{\vec{q}}{\varphi}{\Xi}{E[v]}
        \\
        \QCloE{\vec{q}}{\varphi}{\Xi}{E[\,\qket{0}]} \rightsquigarrow \QCloE{\vec{q}q'}{\varphi \otimes \varphi_0}{\Xi}{E[q']}
        \\
        \QCloE{q_1 q_2 \dots q_k \vec{q}}{\varphi}{\Xi}{E[U\,(q_1 \otimes \dots \otimes q_k)]} \rightsquigarrow \QCloE{q_1 q_2 \dots q_k \vec{q}}{(\psi_U \otimes \ident) \circ \varphi}{\Xi}{E[q_1 \otimes \dots \otimes q_k]}
        \\
        \QCloE{q' \vec{q}}{\varphi}{\Xi}{E[\QMeas\,q']} \rightsquigarrow \QCloE{\vec{q}}{(\theta_0 \otimes \ident) \circ \varphi}{\Xi}{E[\TInl\,\TUnitVal]} + \QCloE{\vec{q}}{(\theta_1 \otimes \ident) \circ \varphi}{\Xi}{E[\TInr\,\TUnitVal]}
        \\[-10pt]
    \end{gather*}
    where superoperators \( \varphi_0 \in \CQ(1,2) \), \( \theta_0,\theta_1 \in \CQ(2,1) \) and \( \psi_U \in \CQ(2^n,2^n) \) are given by
    \begin{gather*}
        \varphi_0(r) \defe r \left( \begin{array}{cc} 1 & 0 \\ 0 & 0 \end{array} \right),
        \qquad
        \psi_U(X) \defe UXU^\dagger,
        \qquad
        \theta_0\left( \left( \begin{array}{cc} a & b \\ c & d \end{array} \right) \right) \defe a
        \quad\mbox{and}\quad
        \theta_1\left( \left( \begin{array}{cc} a & b \\ c & d \end{array} \right) \right) \defe d.
    \end{gather*}
    \vspace{-10pt}
    \caption{The one-step reduction relation \( (\rightsquigarrow) \).}
    \label{fig:quantum-one-step-reduction}
\end{figure}

The \emph{multi-step reduction relation} \( (\stackrel{*}{\rightsquigarrow}) \) is defined by the following rules:
\begin{gather*}
    \dfrac{
        \mathstrut
    }{
        \QCloS \stackrel{*}{\rightsquigarrow} \QCloS
    }
    \qquad
    \dfrac{
        \QCloS \stackrel{*}{\rightsquigarrow} \QCloS_0 + \sum_{i=1}^k \QCloS_{\ac{i}}
        \qquad
        \QCloS_0 \rightsquigarrow \sum_{i=1}^m \QCloS'_{\ac{i}}
    }{
        \QCloS \stackrel{*}{\rightsquigarrow} \sum_{i=1}^m \QCloS'_{\ac{i}} + \sum_{i=1}^k \QCloS_{\ac{i}}
    }.
\end{gather*}

A closed term \( {}\vdash t \colon \TyUnit \) of the unit type is called a \emph{program}.
We write \( [t] \) for \( \QCloE{\langle\rangle}{\ident}{\langle\rangle}{t} \).
If \( [t] \stackrel{*}{\rightsquigarrow} \sum_i \QCloS_i + \sum_j \QCloValS_j \) with \( \QCloValS_j = \QCloE{\langle\rangle}{\varphi_j}{\Xi_j}{\TUnitVal} \), then \( t \) terminates at least probability \( \sum_j \varphi_j \).
The \emph{termination probability of \( t \)} is defined as the supremum of these lower bounds, \ie,
\begin{equation*}
    \textstyle
    \mathrm{Pr}(t \Downarrow \TUnitVal) \defe \sup \{\, \sum_j \mathrm{Pr}(\QCloValS_j \Downarrow \TUnitVal) \mid [t] \stackrel{*}{\rightsquigarrow} \sum_i \QCloS_i + \sum_j \QCloValS_j \,\},
\end{equation*}
where \( \mathrm{Pr}(\QCloValS \Downarrow \TUnitVal) = \varphi \) for \( \QCloValS = \QCloE{\langle\rangle}{\varphi}{\Xi}{\TUnitVal} \). 
It is not difficult to see that \( 0 \le \mathrm{Pr}(t \Downarrow \TUnitVal) \le 1 \).

 \section{Superoperator-Module Model}\label{sec:model}
This section aims to give concrete definitions of our models, the category \( \pshCQ \) of modules over superoperators, and to state the full abstraction result.
This section also provides a necessary background to understand the definition of \( \pshCQ \) and \( \classicalCQ \), together with their justifications.

\subsection{Preliminaries: $\Sigma$-monoids and Enriched Categories}
\label{sec:pre-sigma-monoids-and-enriched-categories}
This subsection briefly reviews two key notions employed in this paper, namely \emph{\(\Sigma\)-monoids} and \emph{enriched categories}.
We use them as a guide to define modules over superoperators properly.

\subsubsection{$\Sigma$-monoids}
\label{sec:defining-sigma-monoids}
A \emph{\( \Sigma \)-monoid}~\cite{Haghverdi2000,Hoshino2011,Tsukada2018,Tsukada2022} is a partial algebra \( M \) having countable partial sum \( \sum_{i \in I} x_i \) for \( I \subseteq \Nat \) and \( x_i \in M \).

The motivation for considering \( \Sigma \)-monoid in this paper comes from a linear algebraic approach to linear logic.
Given a field \( k \) in the standard sense, the category of \( k \)-linear spaces and \( k \)-linear maps is a prototypical model\footnote{However, the notion of spaces that gave rise to linear logic is different, known as coherence spaces.}
of intuitionistic linear logic with the cofree linear exponential comonad.
However, the linear algebra model is not completely satisfactory, as the category of \( k \)-linear spaces is not a model of classical linear logic, \ie~the negation is not involutive on infinite-dimesnional spaces.
The subcategory of finite-dimensional \( k \)-linear spaces is a model of classical linear logic, but the cofree linear exponential comonad is not closed in this subcategory as \( !M \) is often infinite-dimensional (even for finite-dimensional \( M \)).
We want a variant of linear algebras that provides a model with involutive negation and a linear exponential comonad.

In the linear algebra models, the \emph{finite}-dimensional spaces behave well because a \( k \)-linear space is an algebra with \emph{finite} sums.
So, it is natural to expect an algebra with \emph{infinite} sums to provide a basis for linear logic with \emph{infinite}-dimensional spaces such as \( !M \).
\citet{Laird2013} studied linear algebras with totally-defined infinite sums, and \citet{Tsukada2022} discussed linear algebras with partially-defined infinite sums.
We follow the latter approach since the sum of superoperators is not always defined.\footnote{The sum \( \varphi + \psi \) of superoperators \( \varphi,\psi \) may not be a superoperator as the trace of \( \varphi+\psi \) may exceeds \( 1 \).}
\(\Sigma\)-monoid is the algebra that the latter approach is based on.

\begin{notation}
  This paper deals with algebras with partially-defined operations.
  So the value of an expression \( e \) may not be defined (as in the case of \( 1/0 \)).
  For expressions \( e \) and \( e' \), \( e = e' \) means that both \( e \) and \( e' \) are defined and their values coincide.
  The \emph{Kleene inequality} \( e \Kle e' \) means that, if \( e \) is defined, then so is \( e' \) and their values coincide.
  Note that \( e = e \) is not always true; this means that \( e \) is defined.
  We write \( \IsDefined{e} \) for \( e = e \) and \( e \Keq e' \) for \( (e \Kle e') \wedge (e' \Kle e) \).
\eqed
\end{notation}

\begin{definition}[$\Sigma$-monoid]\label{def:sigma-monoid}
  A \emph{$\Sigma$-monoid} is a set \( M \) equipped with a partially-defined sum \( \sum (x_i)_{i \in I} \) for a countable family \( (x_i)_{i \in I} \) of elements in \( M \) subject to the following axioms:
\begin{itemize}
    \item \textbf{Empty Sum}: \( \sum \emptyset \) is defined.  Its value is written as \( 0 \).
    \item \textbf{Singleton Sum}: \( \sum (x_i)_{i \in \{\star\}} = x_\star \).
    \item \textbf{Zero}: If \( (x_i)_{i \in I} \) is a countable family of elements in \( M \) and \( I' \subseteq I \) is a subset such that \( i \in I \setminus I' \) implies \( x_i = 0 \), then \( \sum (x_i)_{i \in I} \Keq \sum (x_i)_{i \in I'} \).
    \item \textbf{Commutative}: If \( (x_i)_{i \in I} \) and \( (y_j)_{j \in J} \) are countable families of elements in \( M \) and \( \sigma \colon I \longrightarrow J \) is a bijection such that \( x_i = y_{\sigma(i)} \), then \( \sum (x_i)_{i \in I} \Keq \sum (y_j)_{j \in J} \).
    \item \textbf{Associativity}: \( \sum (\sum_{j \in J_i} x_{i,j})_{i \in I} \Keq \sum (x_{i,j})_{i \in I, j \in J_i} \).
  \end{itemize}
  We write \( \sum_{i \in I} x_i \) or \( \sum_i x_i \) for \( \sum (x_i)_{i \in I} \).
  A \( \Sigma \)-monoid is equipped with the canonical pre-order defined by \( (x \le y) \defp (\exists z.\ x+z=y) \).
  A \( \Sigma \)-monoid \( M \) is \emph{\( \omega \)-complete} if \( \IsDef{(\sum_{i \in J} x_i)} \) for every finite subset \( J \subseteq I \) implies \( \IsDef{(\sum_{i \in I} x_i)} \) and \emph{cancellable} if \( x + y = x + z \) implies \( y = z \).
  The canonical pre-order of a cancellable \( \Sigma \)-monoid is a partial order.
\eqed
\end{definition}

\begin{example}
  Let \( \Real \) (resp.~\( \Real_{\ge 0} \)) be the set of real numbers (resp.~non-negative real numbers).
  \begin{itemize}
    \item \( \Real_{\ge 0} \) is a \( \Sigma \)-monoid.
    The finite sum is the standard one, and an infinite sum is defined by \( \sum_{i \in \Nat} x_i \Keq \lim_{n \to \infty} \sum_{i=0}^n x_i \) (and it is undefined if \( \lim_{n \to \infty} \sum_{i=0}^n x_i = \infty \)).
    The infinite sum is stable under a rearrangement since it absolutely converges whenever it converges.
    \item The unit interval \( [0,1] \subseteq \Real_{\ge 0} \) has a \( \Sigma \)-monoid structure inherited from \( \Real_{\ge 0} \).
    The sum \( \sum_{i} x_i \) is defined if it is defined in \( \Real_{\ge 0} \) and its value is in \( [0,1] \).
    In this \( \Sigma \)-monoid, even a finite sum may be undefined.
    This \( \Sigma \)-monoid is \( \omega \)-complete.
    \item For every \(n,m \in \Nat \), \( \CPM(n,m) \) and \( \CQ(n,m) \) are \( \Sigma \)-monoids.
    See \cref{sec:enrichment-of-q} for details.
  \end{itemize}
  We also give some non-examples.
  \begin{itemize}
    \item \( \Real \) with the standard sum (\ie~\( \sum_i x_i \) is the limit of \( \sum_{i = 0}^n x_i \)) is not a \( \Sigma \)-monoid because the sum of a conditionally convergent series \( (x_i)_{i \in \Nat} \) is changed by a rearrangement.
    \item \( \Real \) with the absolutely converge sum is not a \( \Sigma \)-monoid because the sum is not associative.
    For example, for \( x_i = (-1)^i \), the sum \( \sum_{i} (x_{2i} + x_{2i+1}) \) absolutely converges but \( \sum_i x_i \) does not. 
\qed
  \end{itemize}
\end{example}

\begin{definition}
  Let \( M \) and \( N \) be \( \Sigma \)-monoids.
  A \emph{\(\Sigma\)-monoid homomorphism} is a function \( f \colon M \longrightarrow N \) that preserves the sum in the sense that
  \(
    f( \sum_{i \in I} x_i )
    \Kle
    \sum_{i \in I} f(x_i)
  \).
  So \( f(0)=0 \) by applying this condition to the empty sum.
  The \emph{category of \(\Sigma\)-monoids} \( \SMon \) has \( \Sigma \)-monoids as objects and \( \Sigma \)-monoid homomorphisms as morphisms, composed as functions.
\eqed
\end{definition}

The category \( \SMon \) satisfies many good properties (see \cite{Tsukada2022} for basic properties of \( \SMon \)).
The category \( \SMon \) is a symmetric monoidal closed category, where the tensor product \( M \otimes N \) is defined as the representing object of bilinear maps.
For \( \Sigma \)-monoids \( M \), \( N \) and \( L \), a \emph{bilinear map} from \( M \) and \( N \) to \( L \) is a function \( f \colon M \times N \longrightarrow L \) such that \( f(\sum_i x_i, \sum_j y_j) \Kle \sum_{i,j} f(x_i, y_j) \).
Let \( \Bilin(M, N; L) \) be the set of bilinear maps from \( M \) and \( N \) to \( L \).
Then \( \Bilin(M, N; {-}) \colon \SMon \longrightarrow \Set \) is a functor in a natural way, and \( M \otimes N \) is a \( \Sigma \)-monoid such that \( \Bilin(M, N; {-}) \cong \SMon(M \otimes N, {-}) \).
The right adjoint \( M \multimap N \) is the set of \( \Sigma \)-monoid homomorphisms \( f \colon M \longrightarrow N \) whose sum is defined if \( \IsDef{(\sum_i f_i)} \defp (\forall x \in M. \IsDef{(\sum_i f_i(x))}) \) and in this case, \( (\sum_i f_i)(x) = \sum_i f_i(x) \).

\subsubsection{Enriched Categories}
\label{sec:defining-enriched-category}
For a usual category \( \cat \), the collection \( \cat(A,B) \) of all morphisms from \( A \) to \( B \) is a mere set.
The enriched category theory considers the situation where \( \cat(A,B) \) has an additional structure.
For example, in an order-enriched category \( \cat \), each hom-set \( \cat(A,B) \) is an ordered set.
See Kelly's textbook~\cite{Kelly1982} for the formal definition and basic results.

This paper employs a characterisation of the modules over a ring in terms of the enriched category theory.  Let \( \Ab \) be the category of abelian groups.   Intuitively, an \( \Ab \)-enriched category \( \cat \) is a category in which \( \cat(A,B) \) is an abelian group for each object \( A,B \in \cat \).
Then a commutative ring \( R \) induces an \( \Ab \)-enriched single-object symmetric monoidal category and an \( R \)-module is just a presheaf over \( R \) in the \( \Ab \)-enriched sense (see below for the precise definition).
This paper aims to develop a variant of the module theory by replacing \( \Ab \) with \( \SMon \).

No knowledge of deep results of the enriched category theory is needed to understand this paper; we just use some definitions of basic notions such as enriched categories and enriched functors.
Since this paper considers mainly \( \SMon \)-enriched categories, we give a definition specialised to this setting, instead of the general definition of the enriched category. 

\begin{definition}
\label{def:enriched-category}
  A category \( \mathcal{C} \) is \emph{$\SMon$-enriched} if each hom-set \( \mathcal{C}(A,B) \) is equipped with a \( \Sigma \)-monoid structure preserved by composition, \ie~\( (\sum_{i \in I} g_i) \circ (\sum_{j \in J} f_j) \Kle \sum_{(i,j) \in I \times J} (g_i \circ f_j) \) for every family \( (g_i)_{i \in I} \) of \( \mathcal{C}(B,C) \) and \( (f_j)_{j \in J} \) of \( \mathcal{C}(A,B) \).
  A functor \( F \colon \mathcal{C} \longrightarrow \mathcal{D} \) between \( \SMon \)-enriched categories is \emph{$\SMon$-enriched} if \( F (\sum_{i \in I} f_i) \Kle \sum_{i \in I} F f_i \).
A \emph{$\SMon$-enriched natural transformation} between \( \SMon \)-enriched functors is a natural transformation in the standard sense.
\eqed
\end{definition}

\begin{example}
  \( \SMon \) itself is \( \SMon \)-enriched by \( \SMon(M,N) = (M \multimap N) \) as \( \Sigma \)-monoids.
\eqed
\end{example}

\begin{remark}
  For a \( \SMon \)-enriched category \( \mathcal{C} \), its underlying (\( \Set \)-enriched) category is also written as \( \mathcal{C} \) by abuse of notation.
  Technically this convention does not cause a problem since the forgetful functor \( \SMon(I, {-}) \colon \SMon \longrightarrow \Set \) is faithful.
  Because of this convention, a functor on \( \mathcal{C} \) is not necessarily \( \SMon \)-enriched, unless explicitly stated to be so.
  An example of non-\(\SMon\)-enriched functor in this paper is the cofree linear exponential comonad.
\eqed
\end{remark}

\subsection{Modules over Superoperators}
\label{sec:enrichment-of-q}
Now we define the (\(\SMon\)-enriched) category \( \pshCQ \) of \emph{modules over superoperators}.
We provide two equivalent definitions.
The first one is elementary; it is a direct modification of the standard notion of module over a ring.
The second one is more sophisticated, using the notion of enriched presheaves, as we mentioned.
The equivalence of the two definitions is a variant of a well-known result.

We first introduce the \( \Sigma \)-monoid structure to \( \CPM(n,m) \) and \( \CQ(n,m) \).
A finite sum in \( \CPM(n,m) \) is always defined and an infinite sum \( \sum_{i \in \Nat} \varphi_i \) is defined as the limit \( \lim_{n \to \infty} \sum_{i=0}^n \varphi_i \) with respect to the standard topology on \( \Mat_n(\Complex) \cong \Complex^{n \times n} \).
The subset \( \CQ(n,m) \subseteq \CPM(n,m) \) has the \( \Sigma \)-monoid structure inherited from \( \CPM(n,m) \):
given a family \( (\varphi_i)_i \) on \( \CQ(n,m) \), the sum \( \sum_i \varphi_i \) is defined in \( \CQ(n,m) \) if \( \sum_i \varphi_i \) is defined in \( \CPM(n,m) \) and its value belongs to the subset \( \CQ(n,m) \subseteq \CPM(n,m) \).
Note that even a finite sum \( \varphi + \varphi' \) may not be defined in \( \CQ(n,m) \).

\( \CPM \) and \( \CQ \) are \( \SMon \)-enriched categories, \ie~their compositions are bilinear \( (\sum_i \varphi_i) \circ (\sum_j \psi_j) \Keq \sum_{i,j} \varphi_i \circ \psi_j \).
\( \CPM \) is a \( \SMon \)-enriched symmetric monoidal closed category in the sense that \( \otimes \) and \( \multimap \) are \( \SMon \)-enriched, \ie~\( (\sum_i \varphi_i) \otimes (\sum_j \psi_j) \Kle \sum_{i,j} (\varphi_i \otimes \psi_j) \) and \( (\sum_i \varphi_i) \multimap (\sum_j \psi_j) \Kle \sum_{i,j} (\varphi_i \multimap \psi_j) \).
\( \CQ \) is a \( \SMon \)-enriched symmetric monoidal category.
Furthermore \( \CPM(n,m) \) has an action of non-negative real number \( r \in \Real_{\ge 0} \).
For \( \varphi \in \CPM(n,m) \) and \( r \in \Real_{\ge 0} \), let \( r\varphi \) be the function defined by \( (r\varphi)(x) \defe r(\varphi(x)) \) as usual.
Similarly \( \CQ \) has the action of \( [0,1] \).
The composition preserves this action, \ie~\( (r\varphi) \circ \psi = r(\varphi \circ \psi) = \varphi \circ (r\psi) \).

The first definition of modules over superoperators is a direct modification of modules over a ring.
Let us recall the standard notion of (right-)module over a ring \( \mathcal{R} \): it is an Abelian group \( M \) with a bilinear action \( x \cdot r \in M \) for \( x \in M \) and \( r \in \mathcal{R} \) satisfying \( x \cdot 1 = x \) and \( (x \cdot r) \cdot r' = x \cdot (rr') \).
The bilinearity means that \( (\sum_{i \in I} x_i) \cdot (\sum_{j \in J} r_j) = \sum_{i \in I, j \in J} (x_i \cdot r_j) \) (\( I \) and \( J \) are finite).
\begin{definition}\label{def:elementary-module}
  A \emph{\(\CQ\)-module} \( M \) is a family \( M = (M_n)_{n \in \Nat} \) of \( \Sigma \)-monoids together with a family of bilinear actions \( ({-}) \cdot ({-}) \colon M_n \otimes \CQ(m,n) \longrightarrow M_m \) satisfying \( x \cdot \ident_n = x \) and \( (x \cdot \varphi) \cdot \psi = x \cdot (\varphi \circ \psi) \) for every \( x \in M_n \), \( \varphi \in \CQ(m,n) \) and \( \psi \in \CQ(k,m) \).
  The bilinearity of actions means that
  \begin{equation*}
    \textstyle
    \big(\sum_i x_i \big) \cdot \big(\sum_j \varphi_j \big)
    \quad\Kle\quad
    \sum_{i,j} (x_i \cdot \varphi_j).
  \end{equation*}
  A \emph{\(\CQ\)-module morphism} from \( M \) to \( N \) is a family \( f = (f_n)_{n \in \Nat} \) of \( \Sigma \)-monoid homomorphisms \( f_n \colon M_n \longrightarrow N_n \) that preserves actions: \( f_m(x \cdot \varphi) = f_n(x) \cdot \varphi \) for every \( x \in M_n \) and \( \varphi \in \CQ(m,n) \).
  Let \( \pshCQ \) be the category of \( \CQ \)-modules and \( \CQ \)-module morphisms.
\eqed
\end{definition}

Another definition of \( \CQ \)-module is in terms of the enriched category theory.
\begin{definition}\label{def:enriched-module}
  A \emph{\(\CQ\)-module} \( M \) is a \(\SMon\)-enriched presheaf over \( \CQ \), \ie~a \( \SMon \)-enriched functor \( M \colon \CQ^{\op} \longrightarrow \SMon \).
  Let \( \pshCQ \) be the category of \( \SMon \)-enriched presheaves over \( \CQ \) and \(\SMon\)-enriched natural transformations.
  (Recall that a \( \SMon \)-enriched natural transformation is a natural transformation in the standard sense, \cf~\cref{def:enriched-category}.)
\eqed
\end{definition}

The coincidence of \cref{def:elementary-module,def:enriched-module} is a variant of a well-known result in the abelian-enriched case (see, \eg, \cite[Section~1.3]{BorceuxVolume2}).
A \( \SMon \)-enriched presheaf \( M \colon \CQ^{\op} \longrightarrow \SMon \) gives a
  \( \CQ \)-module \( (M(n))_{n \in \Nat} \) with the action \( x \cdot \varphi \defe M(\varphi)(x) \).
  Conversely, given a \( \CQ \)-module in the sense of \cref{def:elementary-module}, the associated \( \SMon \)-enriched functor is \( n \mapsto M_n \) on objects and \( \varphi \mapsto ({-}) \cdot \varphi \) on morphisms.

The category \( \pshCQ \) itself is \( \SMon \)-enriched.
Given a family \( (f_i)_i \) over \( \pshCQ(M,N) \), their sum \( \sum_i f_i \) is defined if \( \sum_i (f_i)_n \) is defined for every \( n \) and then \( (\sum_i f_i)_n \) is defined as \( \sum_i (f_i)_n \).

\begin{remark}\tk{This remark is added}
  There is another justification of \( \CQ \)-modules in terms of the programming language theory.
  Usually, the interpretation \( \sem{\tau} \) of a type \( \tau \) in a categorical model is the object expressing the collection of all values of type \( \tau \).
  So, consideration of such a categorical model can be, to some extent, substituted by consideration of values.

  Whereas the interpretation of a type is a set \( \mathcal{V} \) of values in the standard setting, we model a type by a \emph{family of values} \( (\mathcal{V}_n)_{n \in \CQ} \) parametrised by \( n \in \CQ \).
  To intuitively understand the relevance of the parameterisation, consider a value \( \vdash f \colon \TyQubit \multimap \TyQubit \) of a function type \( \TyQubit \multimap \TyQubit \).
  As in many other programming languages, a value of a function type is a \emph{closure}, \ie~it secretly captures some variables in the context in which \( f \) is defined.
  This implicit capturing is harmless in most programming languages since a value at runtime is always closed.
  However, quantum programming languages have non-closed values, namely values of the qubit type, which may be entangled with other qubits.
  Due to this non-closeness, the information about what type of data a closure is hiding has a semantic significance.
  The parameter \( n \in \CQ \) in the family \( (\mathcal{V}_n)_{n \in \CQ} \) intuitively corresponds to the type of data that a closure hides (\eg~a value in \( \mathcal{V}_8 \) captures \( 3 \) qubits).

The different components \( \mathcal{V}_n \) and \( \mathcal{V}_m \) of the family are not independent.
  Given \( v(x) \in \mathcal{V}_4 \) (where \( x \) is the captured variable) and a closed term \( \vdash g \colon \TyQubit \multimap (\TyQubit \otimes \TyQubit) \), the expression \( \mathtt{let}\,x=g\,y\,\mathtt{in}\,V(x) \) determines a value \( w(y) \) in \( \mathcal{V}_{2} \).
  This is the \emph{action} of a morphism \( g \in \CQ(2,4) \) to \( v \in \mathcal{V}_4 \).
  This family of open values accompanied by action is almost the definition of the module or enriched presheaf over superoperators.

  The parameterisation is relevant to a technical issue, namely the decomposability of a value of type \( A \otimes B \) into values of type \( A \) and of type \( B \).
  Consider the following program of type \( \mathtt{qubit} \otimes (\mathtt{unit} \multimap \mathtt{bool}) \):
  \begin{equation*}
      \mathtt{let}\:(x \otimes y) = (|0\rangle+|1\rangle)(|0\rangle+|1\rangle)/2\,\mathtt{in}\:(x \otimes (\lambda \_. \mathtt{meas}\,y)).
  \end{equation*}
Its evaluation (in a slightly different syntax) results in
  \begin{equation*}
      \big[\: |xy\rangle = (|0\rangle+|1\rangle)(|0\rangle+|1\rangle)/2, \quad x \otimes (\lambda \_. \mathtt{meas}\,y) \:\big],
  \end{equation*}
which can be decomposed into values of type \( \mathtt{qubit} \) and of \( \mathtt{unit} \multimap \mathtt{bool} \) as
  \begin{equation*}
      \big[\: |x\rangle \!=\! (|0\rangle\!+\!|1\rangle)/\sqrt{2}, ~~ x \:\big]
      \otimes
      \big[\: |y\rangle \!=\! (|0\rangle\!+\!|1\rangle)/\sqrt{2}, ~~ (\lambda \_. \mathtt{meas}\,y) \:\big].
  \end{equation*}
This syntactic decomposition is not always possible.
  Consider another program of the same type,
  \begin{equation*}
      \mathtt{let}\:(x \otimes y) = (|00\rangle + |11\rangle)/\sqrt{2}\,\mathtt{in}\:(x \otimes (\lambda \_. \mathtt{meas}\,y)),
  \end{equation*}
  evaluated to a result almost identical to the expression:
  \begin{equation}
      \big[\: |xy\rangle = (|00\rangle + |11\rangle)/\sqrt{2}, \quad x \otimes (\lambda \_. \mathtt{meas}\,y) \:\big].
      \label{eq:indecomposable-value}
  \end{equation}
  This result cannot be decomposed into \emph{closed} values of \( \mathtt{qubit} \) and of \( \mathtt{unit} \multimap \mathtt{bool} \) because the two components are entangled.
  In our understanding, the failure of Selinger's norm-based model~\cite{Selinger2004a} and Girard's quantum coherence space model~\cite{Girard2004} stems from this indecomposability problem. 

  We address this issue by introducing \emph{open} values.
  A value \( v \in \mathcal{V}_n \) in the \( n \)-component of the family \( (\mathcal{V}_n)_{n \in \CQ} \) implicitly captures a variable, say \( x \), of type \( n \).
  We regard it as a value \( v(x) \) with a free variable \( x \).
  So the parameterisation allows values to be open.
  Using open values, the value in \cref{eq:indecomposable-value} can be decomposed as
  \begin{equation*}
    \big[\: x \:\big]
    \:\otimes\:
    \big[\: (\lambda \_. \mathtt{meas}\,y) \:\big]
    \qquad
    \mbox{where}
    \quad
    \ket{xy} = (|00\rangle + |11\rangle),
  \end{equation*}
  a pair of an open value \( x \) of type \( \TyQubit \), an open value \( (\lambda \_. \mathtt{meas}\,y) \) of type \( (\mathtt{unit} \multimap \mathtt{bool}) \) together with a glue \( \ket{xy} = (|00\rangle + |11\rangle) \) that entangles the two values.
  Actually, this decomposition can be found in \cref{lem:linear:lqt-function-and-tensor}.
  This is why our model precisely handles values of type \( A \otimes B \) in the presence of entanglement between two components.
\qed
\end{remark}

\subsection{Basic Definitions and Properties}

\paragraph{Yoneda Embedding}
As an enriched presheaf category, \( \pshCQ \) enjoys the Yoneda Lemma.
For each object \( n \in \CQ \), the \emph{representable functor} \( \yoneda(n) \defe \CQ({-}, n) \) is a \( \CQ \)-module.
Its \( m \)-th component is given by \( \yoneda(n)_m \defe \CQ(m,n) \).
The action of \( \varphi \in \CQ(k,m) \) to \( x \in \yoneda(n)_m = \CQ(m,n) \) is the composition \( x \cdot \varphi \defe x \circ \varphi \).
Among others, \( \yoneda(1) \) is of particular importance since it is the unit object of the monoidal structure.
The map \( \yoneda \) is extended to a \( \SMon \)-enriched functor \( \CQ \longrightarrow \pshCQ \).
Its action to a morphism \( \varphi \in \CQ(m,n) \) is 
\( \big(\yoneda(m)_k \ni x \mapsto \varphi \circ x \in \yoneda(n)_k\big)_{k \in \CQ^{\op}} \).
By the (enriched) Yoneda Lemma, \( M_n \cong \pshCQ(\yoneda(n), M) \) as \( \Sigma \)-monoids for every \( \CQ \)-module \( M \).

A closely related \( \CQ \)-module is \( \CPM({-}, n) \) for each \( n \).
Although \( \CPM({-}, n) \) is not representable, we have a canonical \( \Sigma \)-monoid isomorphism
\(
    \pshCQ(\CPM({-},n), \CPM({-},m))
    \cong
    \CPM(n,m)
\)
for each \( n \).
This is because every element \( x \in \CPM(k,n) \) is a finite sum \( x_1 + \dots + x_j \) of elements coming from \( \CQ(k,n) \).
A variant of this fact (\cref{thm:cpm-representation}) shall be heavily used.

\paragraph{Submodule}
A \emph{\(\CQ\)-submodule} of \( M \) is a monomorphism \( L \hookrightarrow M \).
Since a morphism \( \iota \colon L \longrightarrow M \) is a monomorphism if and only if \( \iota_n \) is an injection for every \( n \), we can assume w.l.o.g.~that \( L_n \subseteq M_n \) for every \( n \).
Note that a \( \CQ \)-submodule is not necessarily closed by the sum, \ie~\( x,y \in L_n \) and \( x+y=z \in M_n \) does not imply \( x + y = z \) holds in \( L_n \); \( z \) may not belong to \( L_n \), or \( x+y=z \) may not hold in \( L_n \) even if \( z \in L_n \) (in each case, \( x+y \) must be undefined in \( L_n \)).

A \( \CQ \)-submodule \( L \hookrightarrow M \) is \emph{sum-reflecting} if \( \sum_i x_i = y \) in \( M_n \) and \( x_i, y \in L_n \) implies \( \sum_i x_i = y \) in \( L_n \).
It is \emph{downward-closed} if \( x \in L_n \) and \( y \le x \) holds in \( M_n \) implies \( y \in L_n \).
It is \emph{hereditary} if it is downward-closed and sum-reflecting.

\paragraph{Actions}
Every \( \CQ \)-module \( M \) has the action of \( [0,1] \) defined by \( r\,x \defe x \cdot (r\,\ident_k) \) for \( x \in M_k \) and \( r \in [0,1] \) (where \( \ident_k \in \CQ(k,k) \)).
So \( M_k \) is not just a \( \Sigma \)-monoid but a \( [0,1] \)-module in this sense.
This action can be extended to \( \pshCQ(M,N) \) and others.
These actions of \( r \in [0,1] \), as actions of \( \CQ \), are preserved by many operations including the \( \CQ \)-action, the composition and the tensor product (below).
The action of non-zero \( r \in (0,1] \) is cancellable: \( r\,x = r\,y \) implies \( x = y \).

\paragraph{Local Presentability}
The category \( \pshCQ \) of \( \CQ \)-modules is locally presentable by its algebraic nature.
\emph{Local presentable category} is a class of well-behaved categories, which can be characterised as the class of models of a certain kind of algebras.
For the theory of locally presentable categories, see a textbook~\cite{Adamek1994}.
As a consequence, \( \pshCQ \) is (co)complete.
\begin{proposition}\label{prop:local-presentability}
    \( \pshCQ \) is locally \( \aleph_1 \)-presentable.
\qed
\end{proposition}

\subsection{Structures of Intuitionistic Linear Logic}
\label{sec:intuitionistic}
This subsection proves that the category \( \pshCQ \) of \( \CQ \)-modules is a model of intuitionistic linear logic. 
\begin{theorem}\label{thm:presheaf-model}
    The underlying category of \( \pshCQ \) equipped with the Day tensor is a Lafont model of intuitionistic linear logic, \ie~it is a symmetric monoidal closed category with products, coproducts and cofree cocommutative comonoids.
\qed
\end{theorem}

\begin{notation}\label{notation:ll-connective}
  In order to avoid confusion, we shall reserve linear logical symbols such as \( \otimes \), \( \multimap \) and \( ! \) for types of Quantum FPC and structures of its model \( \classicalCQ \).
  The category \( \classicalCQ \) is a subcategory of \( \pshCQ \) but the embedding does not preserve tensor product, coproduct nor exponential modality.
  The tensor product, coproduct and exponential modality for \( \pshCQ \) shall be written as \( \ptensor \), \( \amalg \) and \( \cofreeexp \).
\eqed
\end{notation}

\paragraph{Additives}
\( \pshCQ \) has the (co)products as \( \pshCQ \) is (co)complete.
Here is a concrete description.

Since the constructions rely on the additive structure of \( \SMon \), let us first recall the (co)product of \( \Sigma \)-monoids.
Let \( X \) and \( Y \) be \( \Sigma \)-monoids with the underlying sets \( |X| \) and \( |Y| \).
Then \( |X\times Y| \defe |X|\times|Y| \) and \( |X \amalg Y| \defe \{\,(x,y) \in |X|\times|Y| \mid x = 0 \vee y = 0 \,\} \).
The sum in \( X \times Y \) is component-wise: \( \sum_i (x_i, y_i) \cong (\sum_i x_i, \sum_i y_i) \).
The sum \( \sum_i (x_i, y_i) \) in \( X+Y \) is similar, but it is defined only if \( \sum_i x_i = 0 \) or \( \sum_i y_i = 0 \).
Note that \( (X+Y) \subseteq (X \times Y) \); actually \( X \amalg Y \) is a \( \Sigma \)-submonoid.
For \( Z = X_1 \times X_2 \) or \( X_1 \amalg X_2 \), we have both projections \( \Proj_i \colon Z \longrightarrow X_i \) and injections \( \Inj_i \colon X_i \longrightarrow Z \) given by, for example, \( \Proj_1(x_1,x_2) = x_1 \) and \( \Inj_1(x_1) = (x_1, 0) \).
The projections and injections satisfy \( \ident_Z = \sum_{i=1,2} \Inj_i \circ \Proj_i \), \( \Proj_i \circ \Inj_i = \ident_{X_i} \) and \( \Proj_i \circ \Inj_{3-i} = 0 \).

The product \( M \times N \) and coproduct \( M \amalg N \) of \( \CQ \)-modules \( M \) and \( N \) is given by the component-wise product \( (M \times N)_n \defe M_n \times N_n \) and component-wise coproduct \( (M \amalg N)_n \defe M_n \amalg N_n \).
The \( \CQ \)-action is defined by \( (x,y) \cdot \varphi \defe (x \cdot \varphi, \, y \cdot \varphi) \) for \( x \in M_n \), \( y \in N_n \) and \( \varphi \in \CQ(m,n) \).
Both the product \( M \times N \) and coproduct \( M \amalg N \) have projections and injections, similar to the case of the \( \Sigma \)-monoid (co)product.
For \( \oast \in \{ \times, \amalg \} \), we write \( \Proj_i \colon M_1 \oast M_2 \longrightarrow M_i \) and \( \Inj_i \colon M_i \longrightarrow M_1 \oast M_2 \) for the projection and injection, respectively.

\paragraph{The Day Tensor}
\label{sec:intuitionistic:day-tensor}
The presheaf category \( \pshCQ \) has the tensor product inherited from \( \CQ \), known as the \emph{Day tensor}~\cite{10.1007/BFb0060438}.
We give an elementary description based on multilinear maps.

Let \( M \), \( N \) and \( L \) be \( \CQ \)-modules.
A \emph{bilinear map} \( f \in \Bilin(M,N; L) \) from \( M \) and \( N \) to \( L \) is a family \( (f_{m,n} \colon M_m \times N_n \longrightarrow L_{m \otimes n})_{m,n \in \Nat} \) of bilinear maps in \( \SMon \) that preserves \( \CQ \)-actions.
Concretely it is a family of maps satisfying
\begin{gather*}
    \textstyle
    f_{m,n}(\sum_i x_i, \sum_j y_j) \:\:\Kle\:\: \sum_{i,j} f_{m,n}(x_i,y_j)
    \quad\mbox{and}\quad
f_{m',n'}(x \cdot \varphi,\: y \cdot \psi) \:\:=\:\: f_{m,n}(x,y) \cdot (\varphi \otimes \psi)
\end{gather*}
for every \( x, x_i \in M_m \), \( y,y_j \in N_n \), \( \varphi \in \CQ(m',m) \) and \( \psi \in \CQ(n',n) \).
The set \( \Bilin(M,N; L) \) of bilinear maps is equipped with the natural sum (\ie~\( \IsDefined{(\sum_i f_i)} \) if and only if \( \IsDefined{(\sum_i (f_i)_{n,m})} \) for every \( n \) and \( m \))
and \( \Bilin(M,N; {-}) \colon \pshCQ \to \SMon \) is a \( \SMon \)-enriched functor.
This \( \SMon \)-enriched functor is representable, \ie~there exists a \( \CQ \)-module \( R_{M,N} \) such that \( \Bilin(M,N; {-}) \cong \pshCQ(R_{M,N}, {-}) \) (natural \( \SMon \)-isomorphism).
The representing object \( R_{M,N} \) is written as \( (M \ptensor N) \) and called the \emph{Day tensor product} of \( M \) and \( N \).
By the general result by Day, \( M \ptensor N \) is defined for every pair of \( \CQ \)-modules.
It can be extended to a functor that is \( \SMon \)-enriched in the sense that \( (\sum_i f_i) \ptensor (\sum_j g_j) \Kle \sum_{i,j} (f_i \ptensor g_j) \).
The Yoneda embedding \( \yoneda \) is strong monoidal: \( \yoneda(n \otimes m) \cong \yoneda(n) \ptensor \yoneda(m) \).
So the tensor unit is \( I = \yoneda(1) \).

\paragraph{Linear Function Space}
\label{sec:intuitionistic:function-space}
As shown generally by \citet{10.1007/BFb0060438},
\( \pshCQ \) has the \( \SMon \)-enriched right-adjoint \( M \multimap ({-}) \) of the Day tensor product \( ({-}) \ptensor M \) for each \( \CQ \)-module \( M \).
Since \(\pshCQ(M, N) \cong \pshCQ(\yoneda(1) \ptensor M, N) \cong \pshCQ(\yoneda(1), M \multimap N) \), by the Yoneda Lemma,
the \( 1 \)st-component \( (M \multimap N)_1 \) is isomorphic to \( \pshCQ(M,N) \) as \( \Sigma \)-monoids.
For general \( n \in \Nat \),
the \( n \)th-component \( (M \multimap N)_n \) of \( M \multimap N \) is given by
\( (M \multimap N)_n \cong \pshCQ(\yoneda(n) \otimes M, N) \cong \Bilin(\yoneda(n), M; N) \)
for the same reason as above.

This can be extended to a \( \SMon \)-enriched functor \( ({-}) \multimap ({-}) \), which is contravariant on the first argument.
The \( \SMon \)-enrichment means \( (\sum_i f_i) \multimap (\sum_j g_j) \Kle \sum_{i,j} f_i \multimap g_j \).

\paragraph{Exponential}
The linear exponential modality is the most interesting structure of linear logic, of which the existence is non-trivial in many cases.
An advantage of the algebraic approach of \citet{Tsukada2022}
is that the existence of the cofree linear exponential comonad trivially follows from the local presentability.
In general, every locally-presentable symmetric monoidal-closed category has cofree cocommutative comonoids~\cite[Remarks~1 in Section~2.7]{Porst2008}.
The proof uses an adjoint functor theorem, so the concrete description is not clear.

\begin{theorem}\label{thm:free-exponential}
    \( \pshCQ \) has the cofree exponential comonad \(\cofreeexp\).
\qed
\end{theorem}

 \section{Matrix Representation and Classical Model}\label{sec:matrix}
This section introduces a matrix calculus in order to ease the computation in \( \pshCQ \).
We introduce a notion of the basis of a \( \CQ \)-module and prove that a morphism between \( \CQ \)-modules with bases can naturally be expressed as a matrix.
Unlike vector spaces, and like modules over rings, a \( \CQ \)-module does not necessarily have a basis.
However many \( \CQ \)-modules of interest actually have bases.

We then define a model \( \classicalCQ \) of classical linear logic, which shall be used to interpret Quantum FPC.
The classical model \( \classicalCQ \) is a full subcategory of \( \pshCQ \) consisting of \( \CQ \)-modules \( M \) with bases such that the canonical morphisms \( M \longrightarrow ((M \multimap I) \multimap I) \) is an isomorphism.
Even though Quantum FPC does not have the duality of classical linear logic nor control operators,
the law \( M \cong \neg\neg M \) of classical linear logic is useful since negated \( \CQ \)-modules \( \neg N \)
have nice properties (\cf~\cref{sec:omega-cpo-enrichment}).

\subsection{Basis, Vector and Matrix}
Recall that our model \( \pshCQ \) is an analogy for the category of (right-)modules over a ring.
Unlike vector spaces, there are several candidates for the notion of a basis
of modules over a ring, possibly in a wider sense than that of free modules.
\citet{Tsukada2022} found that free \(\SMon\)-modules do not have closure properties for obtaining a model of linear logic, and also that the notion of a \emph{dual basis}, which is used to characterise \emph{finitely generated projective modules}, is useful to explain classical models of linear logic.
We follow this approach.

Let us first explain the notion of a dual basis in the standard module theory.
Let \( R \) be a ring.
A \emph{dual basis} of a right \( R \)-module \( M \) is a finite list \( (e_i, f_i)_{i = 1}^k \) of pairs \( e_i \in M \) and \( f_i \colon M \longrightarrow R \) such that
\begin{equation*}
    x
    \quad=\quad
    e_1 \cdot f_1(x) + \dots + e_k \cdot f_k(x)
\end{equation*}
for every \( x \in M \).
Given that a ring can be regarded as a single-object abelian-enriched category, what we need is a multiple-object version of this notion.
\begin{definition}[Representable Basis]\label{def:yoneda-basis}
  Let \( M \) be a \( \CQ \)-module.
  A \emph{countable dual representable basis} (or simply \emph{representable basis}) is a countable family of tuples \( (n_i, e_i, f_i)_{i \in I} \) where \( n_i \) is an object of \( \CQ \), \( e_i \in M_{n_i} \) is an element of \( M \) at \( n_i \), and \( f_i \colon M \longrightarrow \yoneda(n_i) \) is a morphism in \( \pshCQ \) such that
  \begin{equation*}
    \textstyle
    x
    \quad=\quad
    \sum_{i \in I} e_i \cdot (f_{i})_{k}(x)
  \end{equation*}
  for every \( k \in \Nat \) and \( x \in M_k \).
Note that \( (f_{i})_{k}(x) \in \CQ(k,n_i) \) and thus \( e_i \cdot (f_{i})_{k}(x) \in M_k \).
We write \( \basedCQ \) for the full subcategory of \( \pshCQ \) consisting of \( \CQ \)-modules with representable bases.
\eqed
\end{definition}

In the above definition, a coefficient \( (f_{i})_{k}(x) \) is an element \( (f_{i})_{k}(x) \in \CQ(k, n_i) = \yoneda(n_i)_k \) of the representable \( \CQ \)-module \( \yoneda(n_i) \).
It is theoretically possible to consider the case of having
coefficients in general \( \CQ \)-modules,
and it is actually useful to have a slightly wider class of coefficients available than the representable ones, as we will see later. The key to obtaining a general definition is that the element 
\( e_i \in M_{n_i} \)
 can be identified with the morphism \( 
\hat{e_i} \colon \yoneda(n_i)
 \longrightarrow M \).
\begin{definition}[General Basis]\label{def:general-basis}
    Let \( M \) be a \( \CQ \)-module.
    A \emph{countable dual basis} (or simply \emph{basis})  is a countable family of tuples \( (\BObj{b}, \ket{b}, \bra{b})_{b \in B} \) where \( \BObj{b} \) is a \( \CQ \)-module, \( \ket{b} \colon \BObj{b} \longrightarrow M \) and \( \bra{b} \colon M \longrightarrow \BObj{b} \) such that \( x = \sum_{b \in B} (\ket{b} \circ \bra{b})_k(x) \) for every \( k \in \Nat \) and \( x \in M_k \), or equivalently,
    \begin{equation*}
        \textstyle
        \ident_M \quad=\quad \sum_{b \in B} \ket{b} \bra{b}.
    \end{equation*}
    For \( \mathcal{O} \subseteq \mathrm{Obj}(\pshCQ) \), a basis \( (\BObj{b}, \ket{b}, \bra{b})_{b \in B} \) is an \emph{\( \mathcal{O} \)-basis} if \( \BObj{b} \in \mathcal{O} \) for every \( b \in B \).
\eqed
\end{definition}

Let \( M \), \( N \) and \( L \) be \( \CQ \)-modules and \( (\BObj{a}, \ket{a}, \bra{a})_{a \in \Base{M}} \), \( (\BObj{b}, \ket{b}, \bra{b})_{b \in \Base{N}} \) and \( (\BObj{c}, \ket{c}, \bra{c})_{c \in \Base{L}} \) be bases of \( M \), \( N \) and \( L \), respectively.
Then a morphism \( f \colon M \longrightarrow N \) in \( \pshCQ \) 
gives rise to
a matrix \( f = (f_{a,b})_{a \in \Base{M}, b \in \Base{N}} \) given by
\begin{equation*}
    f_{a,b}
    \quad\defe\quad
    \bra{b}f\ket{a}
    \qquad\in\quad \pshCQ(\BObj{a}, \BObj{b}).
\end{equation*}
The composition \( g \circ f \) of morphisms \( f \colon M \longrightarrow N \) and \( g \colon N \longrightarrow L \) 
then gives
the matrix composition \( (g \circ f)_{a,c} = \sum_{b \in \Base{N}} g_{b,c} f_{a,b} \) since
\(
    \bra{c} g f \ket{a}
    =
    \bra{c} g \big(\sum_{b \in \Base{N}} \ket{b}\bra{b}\big) f \ket{a}
    \Kle
    \sum_{b \in \Base{N}} \bra{c} g \ket{b} \bra{b} f \ket{a}
\).
As we shall see later, the actions of functors \( \otimes \), \( \multimap \) and \( ! \) can also be represented as matrix manipulations if one appropriately chooses the bases for \( M \otimes N \), \( M \multimap N \) and \( {!}M \).

Conversely, a matrix \( (f_{a,b} \in \pshCQ(\BObj{a}, \BObj{b}))_{a \in \Base{M}, b \in \Base{N}} \) defines a \( \CQ \)-module homomorphism \( \sum_{a \in \Base{M}, b \in \Base{N}} \ket{b} f_{a,b} \bra{a} \colon M \longrightarrow N \) provided that \( \sum_{a \in \Base{M}, b \in \Base{N}} \ket{b} f_{a,b} \bra{a} \) is defined.
Obviously the matrix representation \( (\bra{b}f\ket{a})_{a,b} \) of \( f \colon M \longrightarrow N \) defines \( f \) as expected.

It is important to note that the matrix representation of a \( \CQ \)-module homomorphism differs in some respects from the matrix representation of a linear map between vector spaces.
\begin{itemize}
    \item A matrix \( (f_{a,b} \in \pshCQ(\BObj{a}, \BObj{b}))_{a \in \Base{M}, b \in \Base{N}} \) does not necessarily define a morphism.
    \item Different matrices \( (f_{a,b} \in \pshCQ(\BObj{a}, \BObj{b}))_{a \in \Base{M}, b \in \Base{N}} \) and \( (g_{a,b} \in \pshCQ(\BObj{a}, \BObj{b}))_{a \in \Base{M}, b \in \Base{N}} \) may define the same morphism.
\end{itemize}
Because of the former, we must always be careful about the convergence of \( \sum_{a \in \Base{M}, b \in \Base{N}} \ket{b} f_{a,b} \bra{a} \).

\begin{remark}
    It is also possible to make a one-to-one correspondence between matrices and morphisms by taking the \emph{canonical matrix representation}: a matrix \( (f_{a,b})_{a,b} \) is \emph{canonical} if \( f_{a,b} = \sum_{a',b'}\braket{b}{b'}f_{a',b'}\braket{a'}{a} \) for every \( a \) and \( b \).
    The matrix \( (\bra{b}f\ket{a})_{a,b} \) induced from a morphism \( f \colon M \longrightarrow N \) is always canonical.
    For every matrix \( (f_{b,a} \in \pshCQ(\BObj{a},\BObj{b}))_{a,b} \), the matrix \( (\sum_{a',b'}\braket{b}{b'}f_{a',b'}\braket{a'}{a})_{a,b} \) is canonical, provided that it is defined.
\eqed
\end{remark}

\subsection{Pseudo-Representable Modules}
\label{sec:pseudo-representable}
We have seen that a morphism \( f \colon M \longrightarrow N \) between \( \CQ \)-modules \(M\) and \( N \) with \( \mathcal{O} \)-bases can be represented as a matrix whose elements are morphisms between \( \CQ \)-modules in \( \mathcal{O} \).
For this matrix representation to be useful, morphisms between \( \CQ \)-modules in \( \mathcal{O} \) must be tractable.

An obvious candidate for \( \mathcal{O} \) is the set of representable \( \CQ \)-modules \( \{\, \yoneda(n) \mid n \in \CQ \,\} \).
By the Yoneda Lemma, \( \pshCQ(\yoneda(n), \yoneda(m)) \cong \CQ(n,m) \) so the morphisms between representable \( \CQ \)-modules are just superoperators, which are easy to reason about.
However the class of representable \( \CQ \)-modules has an inconvenience: it is not closed under the linear function space.

This section introduces a more general class of \emph{pseudo-representable \( \CQ \)-module} and studies its properties including the closure properties under \( \otimes \) and \( \multimap \).

\begin{definition}\label{def:pseudo-representable}
    A \( \CQ \)-module \( \LQT \) is \emph{pseudo-representable} if it is a hereditary submodule \( \LQT \hookrightarrow \CPM({-}, \ell) \) that
satisfies the following conditions:
    \begin{itemize}
\item \textbf{Bounded}: There exists \( B > 0 \) such that \( \opnorm{x} \le B \) for every \( n \) and \( x \in \LQT_n \subseteq \CPM(n,\ell) \).
\item \textbf{Pseudo-universal element}: \( r\,\ident_\ell \in \LQT_\ell \) for some \( r > 0 \).
\end{itemize}
The object \( \ell \) is called the \emph{underlying object} of \( \LQT \) and written as \( \#\LQT \).
\qed
\end{definition}

Every morphism \( f \colon \LQT \longrightarrow \LQT' \) in \( \widehat{\CQ} \) between pseudo-representable \( \CQ \)-modules can be represented by a completely positive map \( \varphi \in \CPM(\#\LQT, \#\LQT') \), although not all completely positive maps represent \( \CQ \)-module morphisms.
\begin{theorem}\label{thm:cpm-representation}
    Let \( \LQT \) and \( \LQT' \) be pseudo-representable \( \CQ \)-modules.
    Then
\begin{equation*}
\widehat{\CQ}(\LQT, \LQT')
        \quad\cong\quad
        \{ \varphi \in \CPM(\#\LQT,\#\LQT') \mid \forall n. \forall x \in \mathcal{L}_n. \varphi \circ x \in \mathcal{L}'_n \}
        \qquad\mbox{(as \( [0,1] \)-modules).}
    \end{equation*}
Here \( \varphi \in \CPM(\#\LQT,\#\LQT') \) satisfying the above condition corresponds to a morphism \( f_\varphi = (f_{\varphi,n})_n \colon \LQT \longrightarrow \LQT' \) given by \( f_{\varphi,n}(x) \defe \varphi \circ x \) for every \( n \) and \( x \in \LQT_n \).
\end{theorem}
\begin{proof}[Proof sketch]
    Clearly \( f_{\varphi} \) is a morphism from \( \LQT \) to \( \LQT' \).
    Let \( g \in \pshCQ(\LQT,\LQT') \) and assume \( r\,\ident_\ell \in \LQT_\ell \) for some
\( r > 0 \), where \( \ell = \#\LQT \).
We can assume \(r \le 1\), for if \(r>1\), 
we have \(1\,\ident_\ell = (1/r)(r\,\ident_\ell) \in \LQT_\ell\).
    Let \( \psi_0 \defe g_\ell(r\,\ident_\ell) \in \LQT'_\ell \subseteq \CPM(\ell,\#\LQT') \) and \( \psi \defe (1/r) \psi_0 \in \CPM(\ell,\#\LQT') \).
    Then \( g = f_{\psi} \).
\end{proof}
In the sequel, we shall identify a \( \CQ \)-module morphism \( f \colon \LQT \longrightarrow \LQT' \) between pseudo-representable modules with a completely positive map \( \varphi \in \CPM(\#\LQT, \#\LQT') \) representing \( f \).

The tensor product and linear function space of pseudo-representable \( \CQ \)-modules are again pseudo-representable.
The following lemma gives a concrete description. 
\begin{lemma}\label{lem:linear:lqt-function-and-tensor}
    Let \( \LQT \) and \( \LQT' \) be pseudo-representable \( \CQ \)-modules.
Let \( (\LQT \rightarrowtriangle \LQT') \hookrightarrow \CPM({-}, \#\LQT \multimap \#\LQT') \) and \( (\LQT \boxtimes \LQT') \hookrightarrow \CPM({-}, \#\LQT \otimes \#\LQT') \) be the hereditary \(\CQ\)-submodules given by
    \begin{align*}
        (\LQT \rightarrowtriangle \LQT')_n &:= \{\, \varphi \in \CPM(n, \#\LQT \multimap \#\LQT') \mid
        \forall m. \forall x \in \mathcal{L}_m. \mathbf{ev} \circ (\varphi \otimes x) \in \mathcal{L}'_{n \otimes m} \,\} \\
        (\LQT \otensor \LQT')_n &:= \{\, (x \otimes x') \circ \varphi \in \CPM(n, \#\LQT \otimes \#\LQT') \mid
        x \in \mathcal{L}_m, x' \in \mathcal{L}'_{m'}, \varphi \in \CQ(n, m \otimes m') \,\}.
    \end{align*}
    Then
\( \LQT \rightarrowtriangle \LQT' \) and \( \LQT \boxtimes \LQT' \) are pseudo-representable.
Furthermore
        \( (\LQT \rightarrowtriangle \LQT') \cong (\LQT \multimap \LQT') \) and \( (\LQT \otensor \LQT') \cong \LQT \ptensor \LQT' \).
\eqed
\end{lemma}
\begin{remark}
    A pseudo-representable \( \CQ \)-module as a subobject \( \LQT \subseteq \CPM({-}, \ell) \) can be seen as a \emph{parametrised predicate} in the sense of Hasegawa~\cite{Hasegawa1999a}.
    The above definitions of \( \LQT \rightarrowtriangle \LQT' \) and \( \LQT \otensor \LQT' \) coincide with the definition of \emph{linear logical predicate}.
\eqed
\end{remark}

We shall hereafter assume that \( (\LQT \multimap \LQT') = (\LQT \rightarrowtriangle \LQT') \) and \( (\LQT \ptensor \LQT') = (\LQT \otensor \LQT') \).

\label{sec:orthogonal-basis}
\begin{definition}[Orthogonal Pseudo-Representable basis]
    A \emph{(countable) pseudo-representable basis} of a \( \CQ \)-module \( M \) is a basis with coefficients taken from pseudo-representable \( \CQ \)-modules.
It is \emph{orthogonal} if \( \braket{a}{a'} = 0 \) for every \( a \neq a' \).
\eqed
\end{definition}

\begin{proposition}\label{prop:pseudo-basis-implies-basis}
    A \(\CQ\)-module \( M \) has a pseudo-representable basis if and only if \( M \in \basedCQ \).
\eqed
\end{proposition}

\subsection{Analysis of Linear Logic Connectives}

We give a concrete description of the linear logic connectives on \( \basedCQ \) using bases.
\( \basedCQ \) is a symmetric monoidal closed category with (co)products, of which the structures are inherited from \( \pshCQ \).
\begin{lemma}\label{lem:based-additive-multiplicative}
  Let \( M,N \in \basedCQ \) and \( (\BObj{a}, \ket{a}, \bra{a})_{a \in \Base{M}} \) and \( (\BObj{b}, \ket{b}, \bra{b})_{b \in \Base{N}} \) be pseudo-representable bases.
  Then \( M \times N \), \( M \amalg N \), \( M \ptensor N \) and \( M \multimap N \) have the following pseudo-representable bases.
    \begin{gather*}
        \Base{M \times N} \defe \{\, (a \times \bullet) \mid a \in \Base{M} \,\} \cup \{\, (\bullet \times b) \mid b \in \Base{N} \,\}
        \\
        \Base{M \amalg N} \defe \{\, \BInl(a) \mid a \in \Base{M} \,\} \cup \{\, \BInr(b) \mid b \in \Base{N} \,\}
        \\
        \Base{M \ptensor N} \defe \{\, a \ptensor b \mid a \in \Base{M}, b \in \Base{N} \,\}
        \\
        \Base{M \multimap N} \defe \{\, a \multimap b \mid a \in \Base{M}, b \in \Base{N} \,\}
        \\
        \BObj{(a \times \bullet)} \defe \BObj{a}
        \qquad
        \ket{a \times \bullet} \defe \Inj_1 \circ \ket{a}
        \qquad
        \bra{a \times \bullet} \defe \bra{a} \circ \Proj_1
        \\
        \BObj{(\bullet \times b)} \defe \BObj{b}
        \qquad
        \ket{\bullet \times b} \defe \Inj_2 \circ \ket{b}
        \qquad
        \bra{\bullet \times b} \defe \bra{b} \circ \Proj_2
        \\
        \BObj{(\BInl(a))} \defe \BObj{a}
        \qquad
        \ket{\BInl(a)} \defe \Inj_1 \circ \ket{a}
        \qquad
        \bra{\BInl(a)} \defe \bra{a} \circ \Proj_1
        \\
        \BObj{(\BInr(b))} \defe \BObj{b}
        \qquad
        \ket{\BInr(b)} \defe \Inj_2 \circ \ket{b}
        \qquad
        \bra{\BInr(b)} \defe \bra{b} \circ \Proj_2
        \\
        \BObj{(a \ptensor b)} \defe (\BObj{a} \ptensor \BObj{b})
        \qquad
        \ket{a \ptensor b} \defe (\ket{a} \ptensor \ket{b})
        \qquad
        \bra{a \ptensor b} \defe (\bra{a} \ptensor \bra{b})
        \\
        \BObj{(a \multimap b)} \defe (\BObj{a} \multimap \BObj{b})
        \qquad
        \ket{a \multimap b} \defe (\bra{a} \multimap \ket{b})
        \qquad
        \bra{a \multimap b} \defe (\ket{a} \multimap \bra{b}).
    \end{gather*}
    The above bases are orthogonal if so are \( \Base{M} \) and \( \Base{N} \).
\eqed
\end{lemma}

The matrix representation of the action of the functors can be obtained by the above lemma.
For example, the matrix representation of \( (f \ptensor f') \) for \( f \colon M \longrightarrow N \) and \( f' \colon M' \longrightarrow N' \) is given as follows: for \( a \in \Base{M} \), \( a' \in \Base{M'} \), \( b \in \Base{N} \) and \( b' \in \Base{N'} \),
\begin{equation*}
    (f \ptensor f')_{a \ptensor a', b \ptensor b'}
    \ =\ 
    \bra{b \ptensor b'} (f \ptensor f') \ket{a \ptensor a'}
    \ =\ 
    \bra{b}f\ket{a} \ptensor \bra{b'}f'\ket{a'}
    \ =\ 
    f_{a,b} \ptensor f'_{a',b'}.
\end{equation*}

The cofree linear exponential comonad \( \cofreeexp M \) is exceptionally hard to describe.
We have shown its existence by the theory of locally presentable category (\cref{thm:free-exponential}), but in fact we do not even know whether \( \cofreeexp M \) has a basis for \( M \in \basedCQ \).
The structure of \(\cofreeexp M\) will therefore be analysed under the assumption that \( M \cong \neg N \)\footnote{The really necessary assumption is that the \( 1 \)-component \( M_1 \) of \( M \) is convex in a weak sense.}
for some \( N \).

Let us recall the basic definitions.
A \emph{comonoid} in \( \pshCQ \) is a \( \CQ \)-module \( C \) together with \( \CQ \)-module morphisms \( u \colon C \longrightarrow \yoneda(1) \) and \( \delta \colon C \longrightarrow C \ptensor C \), called the \emph{counit} and \emph{comultiplication}, respectively, that satisfies
\(
    (u \ptensor \ident) \circ \delta = \ident = (\ident \ptensor u) \circ \delta
\)
and
\(
    (\delta \ptensor \ident) \circ \delta = (\ident \ptensor \delta) \circ \delta
\) (hereafter we often omit the associator).
A comonoid in \( \pshCQ \) is \emph{cocommutative} if \( \mathbf{symm} \circ \delta = \delta \), where \( \mathbf{symm} \colon C \ptensor C \longrightarrow C \ptensor C \) is the symmetry.
A \emph{comonoid morphism} between comonoids \( (C, u_C, \delta_C) \) and \( (D, u_D, \delta_D) \) is a \( \CQ \)-module morphism \( \alpha \colon C \longrightarrow D \) that preserves \( u \) and \( \delta \) in a certain sense.
We write \( \Comon(\pshCQ) \) for the category of cocommutative comonoids in \( \pshCQ \).
The \emph{cofree comonoid over \( M \)}, written \( \cofreeexp M \), is a comonoid \( \cofreeexp M \) with a morphism \( \Der \colon \cofreeexp M \longrightarrow M \) in \(\pshCQ\) (that is not necessarily a comonoid morphism) that satisfies the following universal property: For every comonoid \( C \) with a morphism \( 
h
 \colon C \longrightarrow M \) in \(\pshCQ\), there exists a unique comonoid morphism \(
h^{\cofreeexp}
 \colon C \longrightarrow \cofreeexp M \) such that \( 
h
 = \Der \circ 
h^{\cofreeexp}
 \).
\cref{thm:free-exponential} says that, for every \( \CQ \)-module \( M \), the cofree comonoid \( \cofreeexp M \) over \( M \) indeed exists.

We aim to realise \( \cofreeexp M \) as (a submodule of) the module of formal power series.
Here the \emph{module of formal power series} means \( \prod_{n \in \Nat} M^{\ptensor n} \), where
\( M^{\ptensor n} = \overbrace{M \ptensor \dots \ptensor M}^n \) for \( n \ge 1 \) and \( M^{\ptensor 0} = \yoneda(1) \).

\begin{theorem}\label{thm:exponential-based-setting-formal-power-series}
    Let \( M, M' \in \basedCQ \) with pseudo-representable bases \( \Base{M} \) and \( \Base{M'} \).
    Assume that \( M \cong \neg N \) and \( M' \cong \neg N' \) for some \( \CQ \)-modules \( N \) and \( N' \).
    We assume arbitrary linear orders on \( \Base{M} \) and \( \Base{M'} \) (in order to canonically choose a representative of a finite multiset).
Then:
    \begin{enumerate}
     \item \( \cofreeexp M \in \basedCQ \) and it is a 
submodule of the module of formal power series: \( \cofreeexp M \hookrightarrow \prod_n M^{\ptensor n} \).
\end{enumerate}
    Let \( \Proj^{(n)} \) be the morphism \( \cofreeexp M \hookrightarrow \big(\prod_n M^{\ptensor n}) \stackrel{\textrm{projection}}{\longrightarrow} M^{\ptensor n} \).
    Assume the orthogonality of \( \Base{M} \) and \( \Base{M'} \).  Then:
    \begin{enumerate}\setcounter{enumi}{1}
        \item The following diagram commutes for every \( n,m \ge 0 \):
\begin{equation*}
            \begin{aligned}
                \xymatrix@C=18pt@R=18pt{
                    \cofreeexp M \ar@{^{(}->}[d]_{\Proj^{(n+m)}} \ar[rr]^{\delta} & & \cofreeexp M \ptensor \cofreeexp M \ar@{^{(}->}[d]_{\Proj^{(n)} \ptensor \Proj^{(m)}} \\
                  M^{\ptensor (n+m)} \ar[rr]_{\cong} & & M^{\ptensor n} \ptensor M^{\ptensor m}
                }
\end{aligned}
          \end{equation*}
\item \( \cofreeexp M \) has an orthogonal pseudo-representable basis that satisfies the following properties:
        \begin{gather*}
            \Base{{!}M} = (\mbox{sorted finite sequences over \( \Base{M} \)}),
            \\
\text{
\( \BObj{(a_1 \dots a_k)} \) is a pseudo-representable module over \( (\#a_1 \otimes \dots \otimes \#a_k) \):
}
\\
\BObj{(a_1 \dots a_k)}
            \quad\hookrightarrow\quad
            \BObj{(a_1)} \ptensor \dots \ptensor \BObj{(a_k)}
            \quad\hookrightarrow\quad
            \CPM({-}, \#a_1 \otimes \dots \otimes \#a_k),
\\
            \textstyle
            \Proj^{(n)} \circ \ket{a_1 a_2 \dots a_k} = \sum_{\sigma \in \mathfrak{S}_k} \sigma \circ (\ket{a_1} \ptensor \dots \ptensor \ket{a_k}) \ (\mbox{if } n=k), \qquad 0 \ (\mbox{if } n\neq k),
            \\
            \bra{a_1 \dots a_k} = \dfrac{1}{\#\mathsf{fix}(a_1 \dots a_k)} (\bra{a_1} \ptensor \dots \ptensor \bra{a_k}) \circ \Proj^{(k)}.
        \end{gather*}
        Here \( \mathfrak{S}_k \) is the set of permutations and \( \mathsf{fix}(a_1 \dots a_k) \defe \{\, \sigma \in \mathfrak{S}_k \mid \forall i.\ a_i = a_{\sigma(i)} \,\} \).
\item 
Let \( g \in \basedCQ(M,M') \).
        Then \( \cofreeexp g \) is represented as a matrix \( h = (h_{\vec{a},\vec{a}'})_{\vec{a},\vec{a}'} \) given by \( h_{a_1 \dots a_k, a'_1 \dots a'_{k'}} = g_{a_1,a'_1} \otimes \dots \otimes g_{a_k,a'_k} \) if \( k = k' \) and \( 0 \) if \( k \neq k' \).
        (This matrix representation is not canonical.)
        \item \( x^{\cofreeexp} \in \basedCQ(\yoneda(1), \cofreeexp{M}) \) for \( x \in \basedCQ(\yoneda(1), M) \) is given by \( \bra{a_1 \dots a_k}x^{\cofreeexp} = (\bra{a_1} x) \otimes \dots \otimes (\bra{a_k} x) \).
        (Here, \(x^{\cofreeexp}\) is what is given by the universal property of the cofree exponential.)
\eqed
    \end{enumerate}
\end{theorem}

\subsection{Classical Structures}\label{sec:classic}

Let \( \classicalCQ \hookrightarrow \basedCQ \) be the full subcategory of \( \basedCQ \) consisting of \( M \in \basedCQ \) such that the canonical morphism \( M \longrightarrow \neg\neg M \) is an isomorphism.
Another characterisation of \( \classicalCQ \) is as the Eilenberg-Moore category of the continuation monad \( \neg\neg \) on \( \basedCQ \).
It has the following structures.

\begin{theorem}\label{thm:classical-structure}
    \( \classicalCQ \) is a model of classical linear logic with the tensor product \( (M \otimes N) \defe \neg \neg (M \ptensor N) \)and the linear exponential comonad \( {!}M \defe \neg\neg\cofreeexp M \).
    Its product is \( M \times N \) and coproduct is \( (M + N) \defe \neg\neg(M \amalg N) \).
    The linear function space is \( M \multimap N \). 
\qed
\end{theorem}

 \section{Recursive Types}
\label{sec:recursive}
This section describes the interpretation of recursive types in the model \( \classicalCQ \).
A standard approach to recursive type is based on the CPO enrichment and the \( \omega \)-colimit in the wide subcategory of embedding-projection pairs, and this section follows this standard approach.
We first study the CPO-enrichment of \( \classicalCQ \) and then the wide subcategory of \( \classicalCQ \) consisting of embedding-projection pairs.
Interestingly an \( \omega \)-chain of embedding-projection pairs induces an increasing sequence of bases, and the colimit has a basis consisting of their union.

\subsection{$\omega$CPO Enrichment and Other Properties Inherited from $\CQ$}
\label{sec:omega-cpo-enrichment}
Recall that a \( \CQ \)-module \( M \) has a natural pre-order defined by, for every \( x,y \in M_n \),
\(
(x \le y)
    \defp
    \exists z \in M_n.\ x+z = y
\).
This pre-order is defined in terms of sum, and is preserved by any operation that preserves sums.
So the actions of \( \times \), \( + \), \( \multimap \) to morphisms, as well as composition \( ({-} \circ {-}) \colon \classicalCQ(M,N) \otimes \classicalCQ(N,L) \longrightarrow \classicalCQ(M,L) \) of morphisms, are all monotone.
But this pre-order is not always useful.
The pre-order on a \( \CQ \)-module is not even a partial order in general.

The situation is drastically changed when we focus on \( M \in \classicalCQ \).
Since \( M \cong ((M \multimap I) \multimap I) \) for every \( M \in \classicalCQ \), the module \( M \) inherits many good properties of \( I = \yoneda(1) \).
\begin{lemma}\label{lem:omega-cpo-enrichment}
    For every \( M \in \pshCQ \), the \( \Sigma \)-monoid \( \pshCQ(M, I) \) is cancellable and \( \omega \)-complete.
    The poset \( (\pshCQ(M,I), \le) \) is an \( \omega \)CPO and the infinite sum is the least upper bound of finite partial sums:
    \begin{equation*}
        \textstyle
        \sum_{i \in I} f_i \quad\Keq\quad (\mbox{least upper bound of } \{\, \sum_{i \in I'} f_i \mid I' \subseteq_{\mathrm{fin}} I \,\}).
    \end{equation*}
    In particular \( \classicalCQ(M,N) \) satisfies the above properties for every \( M,N \in \classicalCQ \).
\end{lemma}
\begin{proof}
    These properties trivially hold for the \( \Sigma \)-monoid \( [0,1] \).
    Then the same properties hold for \( \CQ(n,1) \) because \( \CQ(n,1) \) consists of functions to \( [0,1] \) and its sum is point-wise.
    So the same properties hold for \( \pshCQ(M,I) \) because \( \pshCQ(M,I) \) consists of natural transformations, which are families of functions to \( \CQ(n,1) \) (\( n \in \Nat \)), and the sum is component-wise.
    The last claim follows from \( \classicalCQ(M,N) \cong \classicalCQ(M,\neg\neg N) \cong \pshCQ(M,\neg\neg N) \cong \pshCQ(M \ptensor \neg N, I) \) as \( \Sigma \)-monoids.
\end{proof}

Hence \( \classicalCQ \) is an \( \omega \)CPO-enriched category and \( \SMon \)-enriched functors on \( \classicalCQ \) such as \( \times, +, \otimes, \multimap \) are \( \omega \)CPO-enriched.
The exponential \(!\) is \( \omega \)CPO-enriched despite that it is not \( \SMon \)-enriched.
\begin{lemma}\label{lem:classical-exponential-cpo-enriched}
    \( {!} \colon \classicalCQ \longrightarrow \classicalCQ \) is an \( \omega \)CPO-enriched functor.
\eqed
\end{lemma}

\subsection{Embedding-Projection Pairs}
An \emph{embedding-projection pair} in \( \classicalCQ \) is a pair \( (e,p) \) of morphisms \( e \colon M \longrightarrow N \) and \( p \colon N \longrightarrow M \) such that \( p \circ e = \ident_M \) and \( e \circ p \le \ident_N \).
Embedding-projection pairs \( (e,p) \) and \( (e',p') \) compose as \( (e' \circ e, p \circ p') \), provided that \( e' \) and \( e \) are composable.
Let \( \classicalCQep \) be the category whose objects are objects in \( \classicalCQ \) and morphisms from \( M \) to \( N \) are embedding-projection pairs \( (e,p) \) such that \( e \colon M \longrightarrow N \).
Since the embedding \( e \) uniquely determines a projection (if it exists), \( \classicalCQep \) can be seen as a wide subcategory of \( \classicalCQ \).

What is notable is that an embedding-projection pair in \( \classicalCQ \) induces a biproduct-like situation.
Let \( e \colon M \longrightarrow N \) be an embedding with the projection \( p \).
Then \( \iota \defe e \circ p \colon N \longrightarrow N \) is an idempotent.
By the assumption, \( \iota \le \ident_N \), that means, \( \iota + \iota' = \ident_N \) for some \( \iota' \), which is unique by cancellability (\cref{lem:omega-cpo-enrichment}).
Then
\(
    \iota + 0
    =
    \iota 
    =
    (\iota + \iota') \circ \iota
    =
    (\iota \circ \iota) + (\iota' \circ \iota)
    =
    \iota + (\iota' \circ \iota)
\),
so \( \iota' \circ \iota = 0 \) by cancellability.
Similarly \( \iota \circ \iota' = 0 \).
So
\(
    \iota + \iota'
    =
    \ident
    =
    (\iota + \iota') \circ (\iota + \iota')
    =
    (\iota \circ \iota) + (\iota \circ \iota') + (\iota' \circ \iota) + (\iota' \circ \iota')
    =
    \iota + (\iota' \circ \iota')
\)
and thus \( \iota' = \iota' \circ \iota' \).
Since \( \classicalCQ \) has a splitting object of an idempotent,
there exist \( e' \colon M' \longrightarrow N \) and \( p' \colon N \longrightarrow M' \) for some \( M' \in \classicalCQ \) such that \( p' \circ e' = \ident_{M'} \) and \( e' \circ p' = \iota' \).
Then, \(p \circ e' = (p \circ e) \circ (p \circ e') \circ (p' \circ e') =\allowbreak
p \circ \iota \circ \iota' \circ e' =\allowbreak 0\), and similarly, \(p' \circ e = 0\).
So we have a situation like a biproduct:\footnote{These pairs do not define a biproduct since \( (e \circ f) + (e' \circ f') \) is not always defined for \( f \colon X \longrightarrow M \) and \( f' \colon X \longrightarrow M' \).}
\begin{equation*}
    \xymatrix{
        M \ar@/^4pt/[r]^e & N \ar@/^4pt/[l]^p \ar@/^4pt/[r]^{p'} & M' \ar@/^4pt/[l]^{e'}
    },
    \qquad
    \begin{aligned}
        (e \circ p) + (e' \circ p') = \ident_N,
        \qquad
        p \circ e' = 0,
        \qquad
        p' \circ e = 0,
        \qquad\quad
        \\
        p \circ e = \ident_M,
        \quad\mbox{and}\quad
        p' \circ e' = \ident_{M'}.
    \end{aligned}
\end{equation*}

\begin{lemma}
    \( \classicalCQep \) has all \( \omega \)-colimits.\footnote{By the limit-colimit coincidence, it suffices to prove that \( \classicalCQ \) has \( \omega^{\op} \)-limits,
but this claim is not trivial.
        Construction of an \( \omega^\op \)-limit in \( \classicalCQ \), as well as construction of an \( \omega \)-colimit in \( \classicalCQep \), requires to provide a basis of the (co)limit, which is the key to the proof.
}
    The \( \omega \)-colimits are preserved by the embedding \( \classicalCQep \hookrightarrow \classicalCQ \).
\end{lemma}
\begin{proof}
    Assume an \( \omega \)-chain \( (M_i)_{i \in \omega} \) in \( \classicalCQep \).
    By the above argument, we have
    \begin{equation*}
        \xymatrix{
            M_0 \ar@/^4pt/[r]^{e_0} & M_1 \ar@/^4pt/[l]^{p_0} \ar@/^4pt/[d]^{p'_1} \ar@/^4pt/[r]^{e_1} & M_2 \ar@/^4pt/[l]^{p_1} \ar@/^4pt/[d]^{p'_2} \ar@/^4pt/[r]^{e_2} & \cdots \ar@/^4pt/[l]^{p_2} \\
            M'_0 \ar@{=}[u] & M'_1 \ar@/^4pt/[u]^{e'_1} & M'_2 \ar@/^4pt/[u]^{e'_2}
        },
        \qquad
        \begin{aligned}[t]
            &
            (e_i \circ p_i) + (e'_{i+1} \circ p'_{i+1}) = \ident_{M_{i+1}},
            \\ &
            p_i \circ e'_{i+1} = 0,
            \qquad
            p'_{i+1} \circ e_i = 0,
            \qquad\quad
            \\ &
            p_i \circ e_i = \ident_{M_{i}},
            \quad\mbox{and}\quad
            p'_i \circ e'_i = \ident_{M'_i}.
        \end{aligned}
    \end{equation*}
    Each \( x \in (M_k)_n \) is canonically written as \( x = x'_0 + x'_1 + \dots + x'_k \) with \( x'_i \in (M'_i)_n \) (\( i = 0,\dots,k \)), where we omit the embeddings for simplicity.
    Hence \( M_k \) has a basis given as the disjoint union of bases for \( M'_0, M'_1,\dots,M'_k \).
    Let \( \mathcal{F}_m \subseteq \prod_{k \in \Nat} \classicalCQ(\yoneda(m) \ptensor M'_k, \yoneda(1)) \) be the subset consisting of \( (f_k)_k \) such that \( \IsDef{(\sum_{i=0}^k (f_i)_n(\ident \ptensor x'_i))} \) in \( \CQ(m \otimes n,1) \) for every \( k \) and \( x = x'_0 + \dots + x'_k \in (M_k)_n \), and \( M' \) be the hereditary submodule of \( \prod_i M'_i \) consisting of elements whose every \(n\)-th component \( x' = (x'_0,x'_1,\dots)\in (\prod_i M'_i)_n \) satisfies \( \IsDef{(\sum_{i=0}^\infty (f_i)_n(\ident \ptensor x'_i))} \) in \( \CQ(m \otimes n,1) \) for every \( (f_k)_k \in \mathcal{F}_m \).
    Then \( M' \) has a basis given by the disjoint union of bases for \( M'_0, M'_1, \dots \).
    We have an embedding \( (M_k)_n \ni (x'_0 + \dots + x'_k) \mapsto (x'_0, \dots, x'_k, 0, 0, \dots) \in (M')_n \) and a projection \( (M')_n \ni (x'_0,x'_1,\dots) \mapsto x'_0 + \dots + x'_k \in (M_k)_n \) (here the convergence of the sum follows from \( \neg\neg M_k \cong M_k \)). 

    This is a colimiting cocone both in \( \classicalCQep \) and \( \classicalCQ \).
    Given a cocone \( (L, (f_i \colon M_i \longrightarrow L)_i) \) in \( \classicalCQ \), we have \( (M')_n \ni (x'_i)_{i \in \omega} \mapsto \sum_i (f_i)_n((e'_i)_n(x'_i)) \in L_n \).
    The well-definedness of the sum comes from the directed completeness of \( L \), which follows from \( L \cong \neg\neg L \).
    Its uniqueness is easy.
    When the cocone comes from \( \classicalCQep \), \ie~\( f_i \colon M_i \longrightarrow L \) is an embedding with projection \( \hat{p}_i \colon L \longrightarrow M_i \) for every \( i \), the projection associated to the map \( M' \longrightarrow L \) is given by \( L_n \ni y \mapsto \big((p'_i)_n\big((\hat{p}_i)_n(y)\big)\big)_{i \in \omega} \in (M')_n \).
\end{proof}

\subsection{Interpretation of Recursive Types}
\label{sec:canonical-basis}
The results in the previous subsections are sufficient to give an interpretation of recursive types in \( \classicalCQ \) following the standard approach (see, \eg, \cite{Fiore1994,Lindenhovius2021}).
We interpret a type \( A = A(X_1,\dots,X_k) \) with free type variables \( X_1,\dots,X_k \) as a functor \( \sem{A} \colon \classicalCQep \times \dots \times \classicalCQep \longrightarrow \classicalCQep \).
For example, \( \sem{A \multimap B}(\overrightarrow{M}) \defe (\sem{A}(\overrightarrow{M}) \multimap \sem{B}(\overrightarrow{M})) \) on objects and \( \sem{A \multimap B}(\overrightarrow{(e, p)}) \defe (\sem{A}^p(\overrightarrow{(e,p)}) \multimap \sem{B}^e(\overrightarrow{(e,p)}),\: \sem{A}^e(\overrightarrow{(e,p)}) \multimap \sem{B}^p(\overrightarrow{(e,p)})) \), where \( (\sem{A}^e(\overrightarrow{(e,p)}), \sem{A}^p(\overrightarrow{(e,p)})) \defe \sem{A}(\overrightarrow{(e,p)}) \) are the embedding/projection components of \( \sem{A}(\overrightarrow{(e,p)}) \).
The base type \(\TyQubit\) is interpreted by \(\sem{\TyQubit}(\overrightarrow{M}) \defe \yoneda(2)\) and \( \sem{\TyQubit}(\overrightarrow{(e,p)}) \defe (\ident_{\yoneda(2)}, \ident_{\yoneda(2)}) \).
The interpretation of the recursive type \( \sem{\mu X.A(X,\overrightarrow{Y})} \) is defined as the \( \omega \)-colimit of the diagram \( 0 \longrightarrow \sem{A}(0, \overrightarrow{Y}) \longrightarrow \sem{A}(\sem{A}(0, \overrightarrow{Y}), \overrightarrow{Y}) \longrightarrow \cdots \).
Since the functor \( \sem{A} \) preserves \( \omega \)-colimits, \( \sem{\mu X.A}(\overrightarrow{M}) \) is isomorphic to \( \sem{A}(\sem{\mu X.A}(\overrightarrow{M}), \overrightarrow{M}) \).

We systematically assign bases to the interpretations of types.
Given a type \( A = A(X_1,\dots,X_k) \) with \( k \)-free type variables and \( M_1,\dots,M_k \in \classicalCQ \) with bases \( \mathcal{B}_1,\dots,\mathcal{B}_k \), the \emph{canonical basis} of \( A \), written as \( \Base{A}[\mathcal{B}_1,\dots,\mathcal{B}_k] \), is defined by the following rules:
\begin{gather*}
    \dfrac{
        a \in \mathcal{B}_i
    }{
        a \in \Base{X_i}[\vec{\mathcal{B}}]
    }
    \qquad
    \dfrac{
        \mathstrut
    }{
        \ast_2 \in \Base{\TyQubit}[\vec{\mathcal{B}}]
    }
    \qquad
    \dfrac{
        \mathstrut
    }{
        \ast_1 \in \Base{\TyUnit}[\vec{\mathcal{B}}]
    }
    \qquad
    \dfrac{
        a \in \Base{A}[\vec{\mathcal{B}}]
        \qquad
        b \in \Base{B}[\vec{\mathcal{B}}]
    }{
        (a \multimap b) \in \Base{A \multimap B}[\vec{\mathcal{B}}]
    }
    \\
    \dfrac{
        a \in \Base{A}[\vec{\mathcal{B}}]
        \qquad
        b \in \Base{B}[\vec{\mathcal{B}}]
    }{
        (a \otimes b) \in \Base{A \otimes B}[\vec{\mathcal{B}}]
    }
    \qquad
    \dfrac{
        a \in \Base{A}[\vec{\mathcal{B}}]
    }{
        \TInl(a) \in \Base{A+B}[\vec{\mathcal{B}}]
    }
    \qquad
    \dfrac{
        b \in \Base{B}[\vec{\mathcal{B}}]
    }{
        \TInr(b) \in \Base{A+B}[\vec{\mathcal{B}}]
    }
    \\
    \dfrac{
        a_1,\dots,a_k \in \Base{A}[\vec{\mathcal{B}}]
        \qquad
        \mathit{Sorted}(a_1\dots a_k)
    }{
        (a_1 \dots a_k) \in \Base{{!}A}[\vec{\mathcal{B}}]
    }
    \qquad
    \dfrac{
        a \in \Base{A[\mu X.A/X]}[\vec{\mathcal{B}}]
    }{
        \TFold(a) \in \Base{\mu X.A}[\vec{\mathcal{B}}]
    }
\end{gather*}
The meaning of the bases of most types should be clear from \cref{lem:based-additive-multiplicative} and \cref{thm:exponential-based-setting-formal-power-series} (but \( (a_1 \dots a_k) \in \Base{!A}[\vec{\mathcal{B}}] \) means \( \neg\neg(a_1 \dots a_k) \) in the notation of \cref{thm:exponential-based-setting-formal-power-series}).\tk{todo: give a proper definition}
For \( \ast_2 \in \Base{\TyQubit} \), we define \((\BObj{\ast_2}, \bra{\ast_2}, \ket{\ast_2}) \defe (\yoneda(2), \ident_{\yoneda(2)}, \ident_{\yoneda(2)}) \).
For \( a \in \Base{\mu X.A}(\vec{\mathcal{B}}) \), the \( \CQ \)-module \( \BObj{a} \) and morphisms \( \ket{a} \colon \BObj{a} \longrightarrow \sem{\mu X.A} \) and \( \bra{a} \colon \sem{\mu X.A} \longrightarrow \BObj{a} \) are
\begin{equation*}
    \BObj{(\TFold(a))} \defe \BObj{a}
    \qquad
    \ket{\TFold(a)} \defe \mathit{fold} \circ \ket{a}
    \quad\mbox{and}\quad
    \bra{\TFold(a)} \defe \bra{a} \circ \mathit{fold}^{-1}.
\end{equation*}

\begin{theorem}\label{thm:basis-of-types}
    Let \( A = A(X_1,\dots,X_k) \) be a type with free variables.
    Assume \( \CQ \)-modules \( M_1,\dots,M_k \in \classicalCQ \) with orthogonal pseudo-representable bases \( \mathcal{B}_1,\dots,\mathcal{B}_k \), respectively.
    Then \( \Base{A}[\vec{\mathcal{B}}] \) is an orthogonal pseudo-representable basis for \( \sem{A}(M_1,\dots,M_k) \).
\eqed
\end{theorem}

 \section{Interpretation of Quantum FPC and Adequacy}
\label{sec:interpretation}

This section defines the interpretation of terms of quantum FPC and proves the soundness and adequacy of the interpretation.
We often write \( A \) for the interpretation \( \sem{A} \) of type \(A\).

\subsection{Interpretation}
The interpretation can be straightforwardly given since Quantum FPC is just a linear \( \lambda \)-calculus with additional constants, and \( \classicalCQ \) has been shown to be a model of linear logic.
Recall that type variables never freely occur in every type judgement \( !\Gamma, \Delta \vdash t \colon A \).
The interpretation of recursive types has already been discussed in \cref{sec:recursive}, and the interpretation of quantum constants are
\begin{equation*}
    \sem{\ket{0}} \defe \yoneda(\varphi_0)
    \qquad
    \sem{U} \defe \yoneda(\psi_U)
    \quad\mbox{and}\quad
    \sem{\QMeas} \defe (\mathit{inl} \circ \yoneda(\theta_0)) + (\mathit{inr} \circ \yoneda(\theta_1))
\end{equation*}
where superoperators 
\( \varphi_0 \in \CQ(1,2) \), 
\( \psi_U \in \CQ(2^n,2^n) \)
 and 
\( \theta_0,\theta_1 \in \CQ(2,1) \)
 are those defined in \cref{fig:quantum-one-step-reduction}.
For a closure \( \QCloS = \QCloE{\vec{q}}{\varphi}{(c_i = \TBang{t_i})_{i = 1}^k}{t} \), its interpretation \( \sem{\QCloS} \in \classicalCQ(\TyUnit,\TyUnit) \) is defined as \( \sem{t} \circ (\yoneda(\varphi) \otimes \sem{\TBang{t_1}} \otimes \dots \otimes \sem{\TBang{t_k}}) \).

In contrast to the interpretations of \( \ket{0} \) and \( U \), which are well-defined in \( \pshCQ \) and \( \basedCQ \), the interpretation of \( \QMeas \) is possible only in \( \classicalCQ \) because of the definedness issue of the sum in the interpretation.
In the coproduct \( M \coprod N \) in \( \CQ \) and \( \basedCQ \), the sum of elements in different components is always undefined.
The coproduct \( M + N \) of \( M \) and \( N \) in the classical model \( \classicalCQ \) is \( \neg\neg(M \coprod N) = \neg(M \coprod N) \multimap \yoneda(1) \) and it inherits the ``convex closedness'' from \(\yoneda(1)\), which plays an essential role in the definedness of the sum.
\begin{lemma}
    \( (\mathit{inl} \circ \yoneda(\theta_0)) + (\mathit{inr} \circ \yoneda(\theta_1)) \) is defined in \( \classicalCQ(\TyQubit,\TyUnit+\TyUnit) \).
\eqed
\end{lemma}
\begin{proof}
    By a general convergence criterion for sums of morphisms in \( \classicalCQ \), it suffices to show
    \begin{equation*}
        f \circ ((\mathit{inl} \circ \yoneda(\theta_0)) \ptensor \yoneda(n)) + f \circ ((\mathit{inr} \circ \yoneda(\theta_1)) \ptensor \yoneda(n))
    \end{equation*}
    converges for every \( n \in \CQ \) and \( f \in \pshCQ((\TyUnit+\TyUnit) \ptensor \yoneda(n), \yoneda(1)) \).
    Letting \( g_0 \defe f \circ (\mathit{inl} \ptensor \yoneda(n)) \) and \( g_1 \defe f \circ (\mathit{inr} \ptensor \yoneda(n)) \), it suffices to check the convergence of
    \begin{equation*}
        g_0 \circ (\yoneda(\theta_0) \ptensor \yoneda(n)) + g_1 \circ (\yoneda(\theta_1) \ptensor \yoneda(n))
        \quad\Keq\quad
        g_0 \circ \yoneda(\theta_0 \otimes \ident_n) + g_1 \circ \yoneda(\theta_1 \otimes \ident_n)
    \end{equation*}
    for \( g_0, g_1 \in \pshCQ(\yoneda(n), \yoneda(1)) \).
    Let \( \gamma_0,\gamma_1 \in \CQ(n,1) \) such that \( g_i = \yoneda(\gamma_i) \) (for \( i = 0,1 \)).
    Then, the problem is to ensure that the completely positive map
    \begin{equation*}
        \gamma_0 \circ (\theta_0 \otimes \ident_n) + \gamma_1 \circ (\theta_1 \otimes \ident_n)
    \end{equation*}
    belongs to \( \CQ(2 \otimes n, 1) \), \ie~its operator norm \( \le 1 \).
    Assume \( x \in \CQ(1,2n) \), regarded as a positive self-adjoint \( (2n \times 2n) \)-matrix.
    Since \( \theta_0 + \theta_1 \in \CQ(2,1) \) is the trace \( \left( \begin{array}{cc} a & b \\ c & d \end{array} \right) \mapsto a + d \), we have \( 1 \ge \trace(x) = \trace((\theta_0 \otimes \ident_n)(x)) + \trace((\theta_1 \otimes \ident_n)(x)) \).
    As \( \gamma_0 \) and \( \gamma_1 \) are trace-non-increasing, \( 1 \ge \trace((\gamma_0 \circ (\theta_0 \otimes \ident_n))(x)) + \trace((\gamma_1 \circ (\theta_1 \otimes \ident_n))(x)) = \trace(((\gamma_0 \circ (\theta_0 \otimes \ident_n)) + (\gamma_1 \circ (\theta_1 \otimes \ident_n)))(x))\).
    This means that \( (\gamma_0 \circ (\theta_0 \otimes \ident_n)) + (\gamma_1 \circ (\theta_1 \otimes \ident_n)) \) has the operator norm \( \le 1 \).
\end{proof}

\begin{remark}
    One can compare the canonical matrix representation of the interpretation of a term in \( \classicalCQ \) with the interpretation in the Pagani-Selinger-Varilon model~\cite{Pagani2014}, which is a matrix by definition.
    Although we have not yet checked the details, we conjecture that the two interpretations coincide on terms without recursive types.
    Furthermore, we expect that the Pagani-Selinger-Valiron model~\cite{Pagani2014} can interpret recursive types as well, providing a fully abstract model of Quantum FPC.
\qed
\end{remark}

\subsection{Soundness}
The soundness of the interpretation easily follows from the fact that \( \classicalCQ \) is a model of linear logic and that \( \mathit{fold} \) and \( \mathit{unfold} \) are the inverses of each other.
\begin{theorem}[Soundness]\label{thm:soundness-finite-step}
    If \( \QCloS \stackrel{*}{\rightsquigarrow} \sum_i \QCloS_i \),
    then \( \sem{\QCloS} = \sum_i \sem{\QCloS_i} \).
\end{theorem}
\begin{proof}
    It suffices to prove the claim for the single-step reduction \( \QCloE{\vec{p}}{\varphi}{\Xi}{t} \rightsquigarrow \sum_i \QCloE{\vec{p}_i}{\varphi_i}{\Xi_i}{t_i} \).
    The single-step case can be shown by the standard technique since \( \classicalCQ \) is a model of linear logic.
    Since \( \vec{p}=\varphi \) and \( (c_i = \TBang{t_i})_{i = 1}^k \) parts in a closure \( \QCloE{\vec{p}}{\varphi}{(c_i = \TBang{t_i})_{i = 1}^k}{t} \) can be seen as the let-binding \( \mathtt{let}\,\vec{p} = \varphi, c_1 = \TBang{t_1},\dots, c_k = \TBang{t_k}\,\mathtt{in}\,t \), this feature is in the scope of the standard technique.
    This let-binding feature also shows the preservation of the interpretation by the rules for \( \ket{0} \) and \(U\).
    The rule for the measurement \( \QMeas \) relies on the distributive law of composition.
\end{proof}

\begin{corollary}\label{cor:sundness-program}
    \( \mathrm{Pr}(t \Downarrow \TUnitVal) \le \sem{t} \) for every program \( {}\vdash t \colon \TyUnit \). 
\eqed
\end{corollary}

\subsection{Adequacy}\label{sec:adequacy}
The adequacy proof is based on finite approximations.
In the case of the \( \lambda Y \)-calculus, it is often useful to view the fixed-point operator \( Y \) as the limit of its finite approximation \( Y_n \) (where \( Y_n\,f = f\,(f\,(\dots f(\mathit{diverge}))\,\dots) \), the \( n \)-fold application of \( f \)). Since the finite approximation \( Y_n \) can be expressed by a term of the simply-typed \( \lambda \)-calculus, the term obtained by substituting \( Y \) for \( Y_n \) is a simply-typed \( \lambda \)-term, so its adequacy is a consequence of the soundness. Then the adequacy for terms with fixed-points follows from the continuity of the interpretation.
Our strategy is similar, we shall approximate the linear exponential comonad instead of (type-level) recursion.

We consider an extension of the calculus, in which the exponential construct \( \TBang{t} \) is annotated by \( \alpha \in \Nat \cup \{ \infty \} \). The annotation represents the number of times this term is usable.
Hence
\begin{align*}
    \QCloE{\vec{q}}{\varphi}{(\Xi, c=\TBang{t}_{k+1},\Xi')}{E[\TDer\,c]}
    &\rightsquigarrow
    \QCloE{\vec{q}}{\varphi}{(\Xi, c=\TBang{t}_{k}, \Xi')}{E[t]}
    \\
    \QCloE{\vec{q}}{\varphi}{(\Xi, c=\TBang{t}_{0}, \Xi')}{E[\TDer\,c]}
    &\rightsquigarrow
    \QCloE{\vec{q}}{0}{(\Xi, c=\TBang{t}_{0}, \Xi')}{E[t]}
\end{align*}
(where \( \infty + 1 = \infty \)).
By the former rule, \( (\TDer\,c) \) is replaced with \( t \) and the annotation decreases by \( 1 \).
The latter case tries to use \( c \) beyond the limit, resulting in a meaningless term whose interpretation is \( 0 \).
Even in the latter case, we do not stop the reduction itself in order to ease the comparison with the original operational semantics.

\newcommand{\BProj}{\mathsf{P}}
The interpretation of \( \TBang{t}_\alpha \colon {!}A \) is given as follows:
\begin{equation*}
    \textstyle
    \sem{\TBang{t}_\alpha} \defe \BProj^{(<\alpha+1)} \sem{\TBang{t}}
    \qquad
    \mbox{where}
    \quad
    \BProj^{(<\alpha')} \defe \sum_{b \in \Base{!A}, |b| < \alpha'} \ket{b}\bra{b}
    \:\:\colon\:\:
    {!}A \longrightarrow {!}A,
\end{equation*}
where \( |b| \) for \( b = (a_1\dots a_k) \in \Base{!A} \) is the length \( k \) of \( b \).
It is the standard interpretation followed by the projection to the specified length.
Note that \( \sem{\TBang{t}_\infty} = \sem{\TBang{t}} = \lim_{k \to \infty} \sem{\TBang{t}_k} \).
\begin{proposition}
    The soundness holds for the extended calculus:
    \( \QCloS \stackrel{*}{\rightsquigarrow} \sum_i \QCloS_i \) implies \( \sem{\QCloS} = \sum_i \sem{\QCloS_i} \).
\end{proposition}
\begin{proof}
    The preservation of the semantics by the additional duplication-controlled reduction rules follows from     \(
        (\Der_A \otimes \ident_{!A}) \circ \delta \circ \BProj^{(<\alpha+1)} = (\Der_A \otimes \BProj^{(<\alpha)}) \circ \delta
    \) and \( (\Der \otimes \ident_{!A}) \circ \delta \circ \BProj^{(0)} = 0 \).
\end{proof}

A term \( t \) is \emph{finitary} if all annotations it involves are finite.
The \emph{size} of a finitary term \( t \) is defined almost as usual but
\( \mathit{size}(\TBang{t}_k) \defe 1 + k \times (\mathit{size}(t) + 1) \).
The size of \( \QCloE{\vec{q}}{\varphi}{(c_i = \TBang{t_i}_{\alpha_i})_{i = 1}^k}{t} \) is \( \mathit{size}(t) + \sum_{i = 1}^k (\mathit{size}(\TBang{t_i}_{\alpha_i}) - 1) \).
A closure \( \QCloE{\vec{q}}{\varphi}{\Xi}{t} \) is \emph{zero} if \( \varphi = 0 \).
The reduction of a finitary term terminates in the following sense.
\begin{lemma}
    Let \( t \) be a finitary term.
    Then \( [t] \stackrel{*}{\rightsquigarrow} \sum_i \QCloS_i \) such that each \( \QCloS_i \) is either value or zero.
\end{lemma}
\begin{proof}
    The following property holds for the single-step reduction: if \( \QCloS \rightsquigarrow \sum_i \QCloS_i \), then \( \mathit{size}(\QCloS) > \mathit{size}(\QCloS_i) \) or \( \QCloS_i \) is zero.
    So reducing the closure \( [t] \) at most \( \mathit{size}(t) \) steps yields the desired situation.
\end{proof}

\begin{corollary}
\label{cor:finiteAdequacy}
    \( \mathrm{Pr}(t \Downarrow \TUnitVal) = \sem{t} \) for a finitary program \( {}\vdash t \colon \TyUnit \).
\qed
\end{corollary}

\begin{theorem}[Adequacy]
    \( \mathrm{Pr}(t \Downarrow \TUnitVal) = \sem{t} \) for a program \( {}\vdash t \colon \TyUnit \).
\end{theorem}
\begin{proof}
    Let \( t \) be a program containing \( k \) occurrences of \( \TBang{-} \) constructs, each of which is annotated by \( \alpha_i \in \Nat\cup\{\infty\} \) (\(i =1, \dots, k\)).
    Let \( A \defe \{\, (\beta_1,\dots,\beta_k) \in \Nat^k \mid \forall i.\ 
\beta_i \le \alpha_i
 \,\} \).
    For \( \vec{\beta} \in A \), we write \( t_{\vec{\beta}} \) as the term obtained by replacing the annotation \( \alpha_i \) with \( \beta_i \).
    Then \( \sem{t} = \bigvee^{\uparrow} \sem{t_{\vec{\beta}}} \) by the continuity of all constructs.
    The term \( t_{\vec{\beta}} \) is finitary, and hence \( \sem{t_{\vec{\beta}}} = \mathrm{Pr}(t_{\vec{\beta}} \Downarrow \TUnitVal) \) by \cref{cor:finiteAdequacy}.
    It is easy to see that \( \mathrm{Pr}(t_{\vec{\beta}} \Downarrow \TUnitVal) \le \mathrm{Pr}(t \Downarrow \TUnitVal) \).
    Hence \( \sem{t} \le \mathrm{Pr}(t \Downarrow \TUnitVal) \).
    The other direction is just \cref{cor:sundness-program}.
\end{proof}

\begin{remark}\label{rem:recursive-types-in-linear-setting}
    \citet{Lindenhovius2021} discusses a \textbf{CPO}-LNL model and its adequacy for a linear \( \lambda \)-calculus with recursive types.
    Unfortunately their general result is not applicable to our setting by two reasons.
    First its adequacy proof requires that, for every \( f_1, g_1 \colon I \longrightarrow A \) and \( f_2, g_2 \colon I \longrightarrow B \), if \( f_1 \neq 0 \neq g_1 \), \( f_2 \neq 0 \neq g_2 \) and \( f_1 \otimes f_2 \le g_1 \otimes g_2 \), then \( f_1 \le g_1 \) and \( f_2 \le g_2 \) (see \cite[Definition~7.1 and Remark~7.2]{Lindenhovius2021}), but \( \classicalCQ \) does not satisfy this property (consider, \( A = B = I \), \( f_1 = f_2 = (1/2)\ident \), \( g_1 = \ident \) and \( g_2 = (1/4)\ident \)).
    Second the reduction of Quantum FPC is branching.
    In fact, our attempt to give an adequacy proof based on logical relation has not been succeeded because of branching.
\qed
\end{remark}
 \section{Full Abstraction}
\label{sec:full-abstraction}
This section proves the full abstraction (\cref{thm:fullabstraction}) of \( \classicalCQ \) with respect to Quantum FPC.
The proof strategy is essentially the same as the proof of the full abstraction result for probabilistic PCF~\cite{Ehrhard2014,Ehrhard2018}, of which the technique is later applied to the quantum \( \lambda \)-calculus by \citet{Clairambault2020}, although our presentation may look slightly different.

The proof is based on the full definability argument.
A model is \emph{fully definable} if every morphism \( f \in \classicalCQ(\sem{A},\sem{B}) \) is the denotation of a term \( x \colon A \vdash t_f \colon B \).
A fully definable, well-pointed and adequate model is trivially fully abstract.
Unfortunately \( \classicalCQ \) is not fully definable, so we introduce a weaker variant sufficient for full abstraction and prove it.
\begin{definition}[Power series, definable-as-a-coefficient]
    Let \( A \) and \( B \) be closed types.
    A \emph{power series} in \( \classicalCQ(A,B) \) is a formal power series \( F(\xi) = \sum_{\alpha \in \Nat^k} F_\alpha \xi^\alpha \) with coefficient \( F_\alpha \) taken from \( \classicalCQ(A,B) \).
    The indeterminant \( \xi = (\xi_1,\dots,\xi_k) \) ranges over \( \Delta_k \defe \{\, (r_1,\dots,r_k) \in [0,1]^k \mid r_1 + \dots + r_k \le 1 \,\} \) and \( \xi^\alpha \) means \( \xi_1^{\alpha_1} \dots \xi_k^{\alpha_k} \).
The power series \( F \) is \emph{definable} if there exists a family \( (t_{\vec{r}})_{\vec{r}\in\Delta_k} \) of terms of type \( x \colon A \vdash t_{\vec{r}} \colon B \) such that \( F(\vec{r}) = \sem{t_{\vec{r}}} \) for every \( \vec{r} \in \Delta_k \).
    A morphism \( f \in \classicalCQ(A,B) \) is \emph{definable as a coefficient} if there exists a definable power series \( F(\xi) = \sum_{\alpha \in \Nat^k} F_\alpha \xi^\alpha \) in \( \classicalCQ(A,B) \) (with \(k \ge 0\)) such that \( f = F_{\alpha} \) for some \( \alpha \).
\qed  
\end{definition}
We extend operations on morphisms such as \( \otimes \) and \( ({-})^{!} \) to those on formal power series.
For example, given formal power series \( F \) and \( G \), the formal power series \( F \otimes G \) and \( F^! \) are unique ones satisfying \( (F \otimes G)(\vec{r}) = (F(\vec{r})) \otimes (G(\vec{r})) \) and \( (F^!)(\vec{r}) = (F(\vec{r}))^! \) for every \( \vec{r} \in \Delta_k \).
The existence of such a power series follows from a formal calculation in which the formal \( \sum \) is treated as the \(\Sigma\)-monoid sum.

The next lemma demonstrates the relevance of this notion to ``define'' a morphism to \( !A \).
For \( (r_1,\dots,r_k) \in \Delta_k \) and terms \( {} \vdash u_i \colon A \) (\( i = 1,\dots,k \)), we write \( r_1 \cdot u_1 + \dots + r_k \cdot u_k \) for the term that executes \( u_i \) with probability \( r_i \) (and immediately diverges with probability \( 1 - \sum_{i} r_i \)).
This term is definable by appropriately using fresh qubits and measurements.

\begin{lemma}\label{lem:power-of-power-series}
    Let \( A \) be a closed type, \( a_1,\dots,a_k \in \Base{A} \) and \( f_i \in \classicalCQ(\TyUnit, \BObj{a_i}) \) for \(i = 1,\dots,k\).
    Assume that \( (a_1 \dots a_k) \in \Base{!A} \) and that \( \ket{a_i} f_i \colon \TyUnit \longrightarrow A \) is definable as a coefficient for every \( i = 1,\dots,k \).
    Then \( \ket{a_1 \dots a_k} (f_1 \otimes \dots \otimes f_k) \colon \TyUnit \longrightarrow {!A} \) is definable as a coefficient.
\end{lemma}
\begin{proof}
    For each \( i \), let \( F_i(\xi_i) = \sum_{\alpha_i} F_{i,\alpha_i} \xi_i^{\alpha_i} \) be a definable power series such that \( \ket{a_i} f_i = F_{i,\alpha_i^{\star}} \) for some \( \alpha_i^{\star} \).
    Let \( \zeta = (\zeta_1,\dots,\zeta_k) \) be fresh indeterminates and \( F(\zeta,\xi_1,\dots,\xi_k) \defe (\zeta_1 F_1 + \dots + \zeta_k F_k)^! \).
    The image of \( F \) by the (extended) post-composition of the embedding \( !A \hookrightarrow \prod_n A^{\otimes n} \) is
    \begin{equation*}
        \textstyle\sum\nolimits_{n = 0}^{\infty} (\zeta_1 F_1 + \dots + \zeta_k F_k)^{\otimes n}
        \quad=\quad
        \textstyle\sum\nolimits_{n = 0}^{\infty} \sum\nolimits_{\varpi \colon \{ 1,\dots,n \} \to \{ 1, \dots, k \}} (\zeta_{\varpi(1)} F_{\varpi(1)} \otimes \dots \otimes \zeta_{\varpi(n)} F_{\varpi(n)}).
    \end{equation*}
    So its \( \zeta_1 \dots \zeta_k \) coefficient, which is still a power series, is
    \(
        \sum\nolimits_{\sigma \in \mathfrak{S}_k}
        (F_{\sigma(1)} \otimes \dots \otimes F_{\sigma(k)})
    \)
    and its \( \zeta_1 \dots \zeta_k \xi_{1}^{\alpha_1^{\star}} \mspace{-2mu}{\dots}\, \xi_{k}^{\alpha_k^{\star}} \) coefficient is
    \(
        \sum\nolimits_{\sigma \in \mathfrak{S}_k}
        (\ket{a_{\sigma(1)}}f_{\sigma(1)} \otimes \dots \otimes \ket{a_{\sigma(k)}} f_{\sigma(k)})
        =
        \Proj^{(k)}\ket{a_1 \dots a_k} (f_1 \otimes \dots \otimes f_k)
    \).
Given \( u^{(i)}_{\vec{r}_i} \) that defines \( F_i(\vec{r}_i) \), the term \( \TBang{p_1 \cdot u^{(1)}_{\vec{r}_1} + \dots + p_k \cdot u^{(k)}_{\vec{r}_k}} \) defines \( F(p_1,\dots,p_k,\vec{r}_1,\dots,\vec{r}_k) \).
    So \( F \) is a definable power series and \( \ket{a_1\dots a_k}(f_1 \otimes \dots \otimes f_k) \) is its \( \zeta_1 \dots \zeta_k \xi_{1}^{\alpha_1^{\star}} \mspace{-2mu}{\dots}\, \xi_{k}^{\alpha_k^{\star}} \) coefficient.
\end{proof}

Not all but sufficiently many elements are indeed definable as coefficients.
We prove that elements definable as coefficients generate the whole set of morphisms in the following sense.
A subset \( X \subseteq \CPM(n,m) \) \emph{generates \( \CPM(n,m) \)} if the least subset \( \overline{X} \subseteq \CPM(n,m) \) that contains \( X \) and is closed under positive weighted sum (\ie~\( p_1 x + p_2 y \) for \( p_1,p_2 \ge 0 \) and \( x,y \in \overline{X} \)) and subtraction (\ie~if \( x + y = z \) and \( x,z \in \overline{X} \), then \( y \in \overline{X} \)) is a dense subset of \( \CPM(n,m) \).
A subset \( X \subseteq \classicalCQ(A,B) \) \emph{generates \( \classicalCQ(A,B) \)} if \( \{ \varphi \mid (\ket{b}\varphi\bra{a}) \in X \} \) generates \( \CPM(\#\BObj{a}, \#\BObj{b}) \) for every \( a \in \Base{A} \) and \( b \in \Base{B} \).
Let \( \DaC{A, \bra{a}} \defe \{\, \varphi \in \CPM(\#\BObj{a}, 1) \mid (\varphi \cdot \bra{a}) \mbox{ is definable as coefficient} \,\} \) and \( \DaC{A, \ket{a}} \defe \{\, \varphi \in \CPM(1, \#\BObj{a}) \mid (\ket{a} \cdot \varphi) \mbox{ is definable as coefficient} \,\} \).
Below, we often use \cref{thm:cpm-representation}.

\begin{theorem}
    Let \( A \) be a type and \( a \in \Base{A} \).
    Then \( \DaC{A, \bra{a}} \) generates \( \CPM(\#\BObj{a},1) \) and \( \DaC{A, \ket{a}} \) generates \( \CPM(1, \#\BObj{a}) \).
\end{theorem}
\begin{proof}
By induction on \( a \in \Base{A} \).
    Here we focus on the latter claim and show interesting cases.

    (Case \( A = A_1 \otimes A_2 \)):
    Then \( a = a_1 \otimes a_2 \) with \( a_1 \in \Base{A_1} \) and \( a_2 \in \Base{A_2} \).
Since \( \DaC{A_1, \ket{a_1}} \) and \( \DaC{A_2, \ket{a_2}} \) generate \( \CPM(1,\#\BObj{(a_1)}) \) and \( \CPM(1,\BObj{(a_2)}) \) by the induction hypothesis, \( \DaC{A_1, \ket{a_1}} \otimes \DaC{A_2, \ket{a_2}} = \big\{\varphi_1 \otimes \varphi_2 \mid \varphi_i \in \DaC{A_i, \ket{a_i}}\big\}\) generates \( \CPM(1, \#\BObj{(a_1)} \otimes \#\BObj{(a_2)}) \).
    So it suffices to prove that \( \DaC{A_1, \ket{a_1}} \otimes \DaC{A_2, \ket{a_2}} \subseteq \DaC{A, \ket{a_1 \otimes a_2}} \).
    Assume \( \varphi_1 \in \DaC{A_1, \ket{a_1}} \) and \( \varphi_2 \in \DaC{A_2, \ket{a_2}} \).
    We have definable power series \( S_1(\xi_1) = \sum_{\alpha_1 \in \Nat^{k_1}} S_{1,\alpha_1} \xi_1^{\alpha_1} \) with \( \exists \alpha_1'. S_{1,\alpha_1'} = \ket{a_1}\varphi_1 \) and \( S_2(\xi_2) = \sum_{\alpha_2 \in \Nat^{k_2}} S_{2,\alpha_2} \xi_2^{\alpha_2} \) with \( \exists \alpha_2'. S_{2,\alpha_2'} = \ket{a_2}\varphi_2 \).
    Then \( (S_1 \otimes S_1)(\xi_1\xi_2) \defe \sum_{\alpha_1 \in \Nat^{k_1}, \alpha_2 \in \Nat^{k_2}} (S_{1,\alpha_1} \otimes S_{2,\alpha_2}) \xi_1^{\alpha_1} \xi_2^{\alpha_2} \) is definable: \( (u^{(1)}_{\vec{r}_1} \otimes u^{(2)}_{\vec{r}_2})_{\vec{r}_1\vec{r}_2} \) is a witness when \( (u^{(1)}_{\vec{r}_1})_{\vec{r}_1} \) and \( (u^{(2)}_{\vec{r}_2})_{\vec{r}_2} \) define \( S_1 \) and \( S_2 \).
    The power series \( S \) has \( \ket{a_1 \otimes a_2}(\varphi_1 \otimes \varphi_2) \) as the coefficient of \( \xi_1^{\alpha_1'}\xi_2^{\alpha_2'} \).

    (Case \( A = !B \)):
    This case can be proved easily by \cref{lem:power-of-power-series}.

    (Case \( A = \mu X. B \)):
    Then \( a = \TFold(b) \) for some \( b \in \Base{B[\mu X.B/X]} \).
    Perhaps surprisingly, this case immediately follows from the induction hypothesis since the proof is by induction on an element \( a \) of the canonical basis, not induction on type \( A \).
\end{proof}

\begin{lemma}
    Let \( A \) be a type
and \( f,g \in \classicalCQ(\TyUnit,A) \).
    If \( f \neq g \), then \( \sem{u} \circ f \neq \sem{u} \circ g \) for some
\( x \colon A \vdash u \colon \TyUnit \).
\end{lemma}
\begin{proof}
    Assume \( f,g \in \classicalCQ(\TyUnit,A) \) and \( f \neq g \).
    There exists \( a \in \Base{A} \) such that \( \bra{a}f \neq \bra{a}g \) in \( \classicalCQ(\TyUnit,\BObj{a}) \).
    Then \( \varphi\bra{a}f \neq \varphi\bra{a}g \) for some \( \varphi \in \CPM(\#\BObj{a},1) \) (as \(\classicalCQ\) is (co-)well-pointed).
    Since \( \DaC{A, \bra{a}} \) generates \( \CPM(\#\BObj{a},1) \), there exists \( \varphi_0 \in \DaC{A,\bra{a}} \) such that \( \varphi_0\bra{a}f \neq \varphi_0\bra{a}g \).
Let \( S(\xi) = \sum_{\alpha \in \Nat^k} S_\alpha \xi^\alpha \) be a definable power series on \( \classicalCQ(A,\TyUnit) \) such that \( S_{\alpha_0} = \varphi_0\bra{a} \) for some \( \alpha_0 \in \Nat^k \).
    Then \( (Sf)(\xi) \defe \sum_\alpha (S_\alpha f) \xi^\alpha \) is a
formal power series with non-negative real coefficients, and its value \( (Sf)(\vec{r}) \) on \( \vec{r} \in \Delta_k \) is the composite \( (S(\vec{r}) \circ f) \in \classicalCQ(\TyUnit,\TyUnit) \) of \( f \in \classicalCQ(\TyUnit,A) \) and \( S(\vec{r}) \in \classicalCQ(A,\TyUnit) \).
    In particular \( (Sf)(\vec{r}) \) absolutely converges on every \( \vec{r} \in [0,1]^k \), and thus it defines an analytic function on \( \Delta_k \).
    This argument applies to \( (Sg)(\xi) \defe \sum_\alpha (S_\alpha g) \xi^\alpha \).
    These series \( (Sf)(\xi) \) and \( (Sg)(\xi) \) have different coefficients of \( \xi^{\alpha_0} \), hence define different functions.
    So there exists \( \vec{r} \in \Delta_k \) such that \( (Sf)(\vec{r}) \neq (Sg)(\vec{r}) \).
    By the definability of \( S(\xi) \), there exists a term \( y \colon A \vdash u \colon \TyUnit \) such that \( \sem{u} = S(\vec{r}) \).
    Then \( \sem{u} \circ f = S(\vec{r}) \circ f = (Sf)(\vec{r}) \neq (Sg)(\vec{r}) = \sem{u} \circ g \) as required.
\end{proof}

\begin{theorem}\label{thm:fullabstraction}
    \( \classicalCQ \) is a fully abstract model of Quantum FPC.
\qed
\end{theorem}

 \section{Cone Norm, Revisited}\label{sec:norm}
This section revisits a norm-based approach to higher-order quantum computation in \citet{Selinger2004a}, which looks natural but fails to provide a correct model.
We analyse this attempt from the viewpoint of superoperator modules and propose a correction.

\subsection{Background: Normed Cones}
\label{sec:defining-cone-norm}
A subset \( C \subseteq V \) of a real vector space \( V \) is a \emph{convex cone} if \( C \) is closed under the sum and the multiplication of non-negative reals.
A \emph{normed cone} is a convex cone \( C \subseteq V \) equipped with \emph{norm}, which is a function \( \opnorm{{-}} \colon C \longrightarrow \Real_{\ge 0} \) that satisfies (1) \( \opnorm{x + y} \le \opnorm{x} + \opnorm{y} \), (2) \( \opnorm{rx} = r \opnorm{x} \), (3)~\( \opnorm{x} = 0 \Rightarrow x = 0 \) and (4) \( \opnorm{x} \le \opnorm{x + y} \)
for every \( x,y \in C \) and \( r \in \Real_{\ge 0} \).
An example of a convex cone is the subset \( P_n \subseteq \Mat_n(\Complex) \) of positive self-adjoint matrices (where \( \Mat_n(\Complex) \) is regarded as a \( (2 n^2) \)-dimensional real vector space) and the trace \( \trace \colon P_n \longrightarrow \Real_{\ge 0} \) is a norm on \( P_n \).
Another example of a normed cone is \( \CPM(n,m) \) equipped with the operator norm \( \opnorm{\varphi} \defe \sup \{\, \trace(\varphi(x)) \mid x \in P_n, (\trace\,x) \le 1 \,\} \) (with respect to the trace norm).

Normed cone was studied in \citet{Selinger2004a} in an attempt to extend the superoperator model \( \CQ \) to a higher-order setting.
Recall that the subset \( \CQ(n,m) \subseteq \CPM(n,m) \) is defined in terms of the trace norms: \( (\varphi \in \CQ(n,m)) \Leftrightarrow (\forall x \in P_n. (\trace\,x) \ge (\trace(\varphi(x)))) \).
However the trace norm is not the only interesting norm on \( P_n \).
For example, we have \( \CPM(n,m) \cong \CPM(1, nm) \cong P_{nm} \) and hence \( P_{nm} \) has the norm corresponding to the operator norm on \( \CPM(n,m) \) (with respect to the trace norm).
The norm on \( P_{nm} \) via the above isomorphism actually differs from the trace norm, so it is natural to consider norms on \( P_{nm} \) other than the trace norm in a higher-order setting.

The category \( \CQ' \)~\cite{Selinger2004a} of \emph{normed matrix spaces} is defined as follows.\footnote{The original category \( \CQ' \) given by \citet[Section~4.1]{Selinger2004a} is actually a full subcategory of the category defined here.  In the original definition, an idempotent \( \varphi \colon P_n \longrightarrow P_n \) is limited to those induced by biproducts \( (\oplus_{i=1}^{k} P_{n_i}) \leftrightarrows P_{n} \) where \( n = \sum_{i} n_i \).  The difference is, however, not significant for the discussion of this section.}
An object of \( \CQ' \) is \( (n, \varphi, \opnorm{-}) \) where \( \varphi \colon P_n \longrightarrow P_n \) is an idempotent and \( \opnorm{-} \) is a norm on \( P_n \) such that \( \varphi \) is norm-non-increasing.
A morphism \( \psi \colon (n,\varphi,\opnorm{-}) \longrightarrow (n',\varphi',\opnorm{-}') \) is a completely positive map \( \psi \in \CPM(n,n') \) that respects the idempotent \( \psi = \varphi' \circ \psi \circ \varphi \) and is norm-non-increasing \( \opnorm{x} \ge \opnorm{\psi(x)}' \) for every \( x \in P_n \).
A pair \( (n, \varphi, \opnorm{-}) \) and \( (n', \varphi', \opnorm{-}') \) induces a norm on \( \CPM(n,n') \), namely the operator norm \( \opnorm{\psi} \defe \sup \{\, \opnorm{\psi(x)}' \mid x \in P_n, \opnorm{x} \le 1 \,\} \), and hence a norm \( \opnorm{-}'' \) on \( P_{nm} \) via the isomorphism \( \CPM(n,n') \cong \CPM(1,nn') \cong P_{nn'} \).
Let \( (n,\varphi, \opnorm{-}) \multimap (n',\varphi',\opnorm{-}') \defe (nn', \psi, \opnorm{-}'') \) where \( \psi \) is the idempotent corresponding to \( \varphi \multimap \varphi' \).
This extends to a functor with the left adjoint \( (-) \otimes (n,\varphi,\opnorm{-}) \).
Concretely \( (n,\varphi,\opnorm{-}) \otimes (n',\varphi',\opnorm{-}') \defe (nn', \varphi \otimes \varphi', \opnorm{-}'') \) where
\begin{equation}
    \textstyle
    \opnorm{z}''
    \quad\defe\quad
    \inf \{\, \sum_i \opnorm{x_i} \opnorm{y_i}' \mid z \le \sum_i x_i \otimes y_i, x_i \in P_n, y_i \in P_m \,\}.
    \label{eq:non-entangled-tensor-norm}
\end{equation}
\citet{Selinger2004a} proved that \( \CQ' \) is a \( * \)-autonomous category.

However \( \CQ' \) has a fatal problem~\citet[Section~4.3]{Selinger2004a}: \( (2, \ident, \trace) \otimes (2, \ident, \trace) \neq (4,\ident, \trace) \).
Actually an entangled state \( z \in P_{2 \times 2} \) with \( (\trace\,z) \le 1 \) has the norm \( > 1 \) with respect to the norm on \( (2, \ident,\trace) \otimes (2, \ident,\trace) \).
Since \( P_4 \) with the trace norm \( \le 1 \) precisely describes the definable values of type \( \TyQubit \otimes \TyQubit \) in a first-order quantum programming language~\cite{Selinger2004}, the tensor product in \( \CQ' \) cannot be a model of even a first-order quantum programming language.
The tensor product in \( \CQ' \) fails to deal with pairs of entangled values, because \cref{eq:non-entangled-tensor-norm} defines the norm of \( z \) in terms of the best entanglement-free approximation (see~\cite[Section~4.2.2]{Selinger2004a}).

\subsection{Families of Norms and Pseudo-Representable Modules}
\label{sec:family-norms}
We propose an alternative to \( \CQ' \) inspired by our module model \( \classicalCQ \).
Let \( \fdCQ \) be the full subcategory of \( \classicalCQ \) consisting of \emph{finite dimensional} \( \CQ \)-modules, \ie~\( M \in \classicalCQ \) with a finite basis \( (\LQT_b, \ket{b}, \bra{b})_{b \in B} \) (\( B \) is a finite set).
We give a norm-based characterisation of \( \fdCQ \).

Because every \( M \in \fdCQ \) is a splitting object of an idempotent \( f \colon \LQT \longrightarrow \LQT \) on a pseudo-representable \( \CQ \)-module \( \LQT \in \classicalCQ \), it suffices to find a norm-based description of a pseudo-representable \( \CQ \)-module \( \LQT \) with \( \LQT \cong \neg\neg \LQT \).
For \( x \in \CPM(n,\#\LQT) \), we define the norm \( \qnorm{x}{\LQT}^{(n)} \) by \(
    \qnorm{x}{\LQT}^{(n)}
    \defe
    \inf \left\{\, r \in \Real_{\ge 0} ~\middle|~ x \in r \cdot \LQT_n \,\right\}
\)
where \( r \cdot X \defe \{ r x \mid x \in X \} \) for \( X \subseteq \CPM(m,\ell) \).

\begin{lemma}\label{lem:linear:property-of-induced-norm}\label{lem:negated-lqt-cone-norm}
The family \( (\qnorm{-}{\LQT}^{(n)})_n \) given above satisfies the following conditions.
    \begin{enumerate}
        \item \( \qnorm{{-}}{\LQT}^{(n)} \) is a norm on \( \CPM(n,\#\LQT) \) for each \( n \).
        \item There exists \( B \) such that \( \opnorm{x} \le B \qnorm{x}{\LQT}^{(n)} \) for every \( n \) and \( x \in \CPM(n,\#\LQT) \).
        \item \( \CQ \)-action is norm-non-increasing: \( \forall x \in \CPM(n,\#\LQT). \forall \varphi \in \CQ(m,n).\: \qnorm{x \circ \varphi}{\LQT}^{(m)} \le \qnorm{x}{\LQT}^{(n)} \).
        \item For every \( x \in \CPM(n,\ell) \),\footnote{This condition corresponds to \( \LQT \cong \neg\neg\LQT \).  More precisely, it is obtained from \( \qnorm{-}{\LQT}^{(n)} = \qnorm{-}{\neg\neg\LQT}^{(n)} \) by rewriting the right-hand-side using the definition.}
\begin{align*}
            \qnorm{x}{\LQT}^{(n)} &= \sup \left\{\, \varphi \circ (\ident_m \otimes x) \circ \psi ~~\middle|~~
            \begin{aligned}
                m,k \in \Nat,\:
                \varphi \in \CPM(m \otimes \ell, 1),\:
                \psi \in \CPM(1, m \otimes k) \\
                \forall y \in \LQT_k. \varphi \circ (\ident_m \otimes y) \circ \psi \le 1
            \end{aligned}
            \right\}.
            \qed
\end{align*}
\end{enumerate}
\end{lemma}

Conversely a family \( (\annorm{-}^{(n)})_n \) of norms on \( \CPM(n,\ell) \) that satisfies the conditions in \cref{lem:linear:property-of-induced-norm} determines a pseudo-representable \( \LQT \in \classicalCQ \) given by
\(
    \LQT_n \defe \{ x \in \CPM(n,\ell) \mid \annorm{x}^{(n)} \le 1 \}
\).
This observation motivates the following definition of the category \( \CQ^{\P} \), an alternative to \( \CQ' \).
Its object is a natural number \( \ell \) together with an idempotent \( \varphi \colon P_{\ell} \longrightarrow P_{\ell} \) and a family \( (\annorm{-}^{(n)})_n \) of norms on \( \CPM(n,\ell) \) that satisfies the conditions in \cref{lem:linear:property-of-induced-norm} such that \( \varphi \) is norm-non-increasing.
Its morphism from \( (\ell_1, \varphi_1, (\annorm{-}_1^{(n)})_n) \) to \( (\ell_2, \varphi_2, (\annorm{-}_2^{(n)})_n) \) is a completely positive map \( \psi \in \CPM(\ell_1,\ell_2) \) such that \( \psi = \varphi_2 \circ \psi \circ \varphi_1 \) and \( \annorm{\psi \circ x}_2^{(n)} \le \annorm{x}_1^{(n)} \) for every \( x \in \CPM(n,\ell) \).
\begin{theorem}\label{thm:norm-category-equivalent}
    \( \CQ^{\P} \) is equivalent to \( \fdCQ \).
\qed
\end{theorem}

Since the category \( \fdCQ \cong \CQ^{\P} \) is closed under \( \multimap \) and \( \otimes \) in \( \classicalCQ \),
the quantum linear \( \lambda \)-calculus in \citet{Selinger2008} can be interpreted in \( \fdCQ \) and \( \CQ^{\P} \), which are actually fully abstract.
The functor \( \CQ^{\P} \longrightarrow K(\CPM) \) (where \( K(\CPM) \) is the Karoubi envelope of \( \CPM \))
that forgets the norm preserves \( \multimap \) and \( \otimes \), so the norm in \( \CQ^{\P} \) captures the total probability invariant of the \( \CPM \) interpretation given by \citet{Selinger2008}.

\subsection{Physically-Closed Module}
Our models \( \fdCQ \) and \( \CQ^{\P} \) provide a new perspective to Selinger's \( \CQ' \)~\cite{Selinger2004a}.
Given a norm \( \annorm{-} \) on \( \CPM(1,\ell) \), let \( \LQT_{\annorm{-}} \) be the minimum pseudo-representable \( \CQ \)-module with \( \#\LQT_{\annorm{-}} = \ell \) such that \( \{ y \in \CPM(1,\ell) \mid \annorm{y} \le 1 \} \subseteq  (\LQT_{\annorm{-}})_1 \) (ordered by the component-wise set-inclusion).
This gives a strong monoidal fully faithful functor \( \CQ' \longrightarrow \fdCQ \).

\begin{theorem}\label{thm:selingers-q-prime}
    \( \CQ' \) is isomorphic to a full subcategory of \( \fdCQ \).
    The embedding is strong monoidal.
\eqed
\end{theorem}

Let us compare the interpretations of linear types in \( \fdCQ \) with \( \CQ' \) via the above embedding.
Since \( (\LQT_{\annorm{-}})_n \) is the convex closure of \( \{\, x \circ \varphi \mid \varphi \in \CQ(n,1), x \in \CQ(1,n), \annorm{x} \le 1 \,\} \)
and \( (x \circ \varphi) \) is equivalent to \( x \otimes \varphi \colon 1 \otimes n \longrightarrow \ell \otimes 1 \) via the canonical isomorphisms, a value in \( (\LQT_{\annorm{-}})_n \) has only classical correlation to the environment.
In this sense, \( \LQT_{\annorm{-}} \) is \emph{physically-closed}.
This is in contrast to typical objects in \( \fdCQ^{\P} \) such as \( \sem{\TyQubit} = (\ident \colon \yoneda(2) \longrightarrow \yoneda(2)) \), which has elements with non-classical correlation with the environment, \eg~\( \ident_2 \in \CQ(2,2) \).
This physical closedness of modules from \( \CQ' \) explains why \( \CQ' \) is not suitable to interpret the quantum \( \lambda \)-calculus.

 \section{Related Work}
\label{sec:related-work}

Recently \citet{Clairambault2019a} gave an intensional game model for quantum programs that addresses the ``total probability'' by using the intensional nature of the model.
Using this game model, \citet{Clairambault2020} proved the full abstraction of the game model and the Pagani-Selinger-Valiron model~\cite{Pagani2014}.
A key of the proof is the convergence of a certain class of \( \CPM \)-coefficient formal power series, which intuitively follows from the bound ``total probability'' \( \le 1 \).

\citet{Malherbe2010} studied the (standard, \(\Set\)-enriched) presheaf category over the superoperator model \( \CQ \) and constructed a model of a quantum \( \lambda \)-calculus.
Malherbe's presheaf model demonstrates, to some extent, the usefulness of presheaves in modelling quantum programs, but the model has some drawbacks such as the absence of recursion.

\citet{Hasuo2011,Hasuo2011a} developed a categorical model of higher-order quantum programs based on the Geometry of Interaction interpretation.
A main drawback of their model is that their tensor product does not allow entanglement~\cite[Remark~IV.1]{Hasuo2011}, behaving like the tensor product in Selinger's \( \CQ' \)~\cite{Selinger2004a}.

\citet{Cho2016} studied a model of quantum computation in operator algebras.
Operator algebras are closely related to a foundation of quantum physics, and we would like to understand the relationship between his model and our models.
\cref{sec:norm} is a first step towards this direction.

The Pagani-Selinger-Valiron model~\cite{Pagani2014} stems from research of constructions of models of linear logic.
The category of \emph{weighted relation} by \citet{Laird2013} is a general construction and is one of the ancestors of the Pagani-Selinger-Valiron model.
\citet{Tsukada2017,Tsukada2018} discuss the Pagani-Selinger-Valiron model in relation to the \emph{Taylor expansion}~\cite{Ehrhard2008} in differential lambda calculus.
The module model of this paper should have some relationship to the Taylor expansion, but the relationship is not yet known.
\citet{Tsukada2017,Tsukada2018} are also inspired by their previous papers~\cite{Tsukada2015,Tsukada2016} concerning \emph{game semantics}~\cite{Abramsky2000,Hyland2000}, so their work should have some relationship to \citet{Clairambault2020} that proves the full abstraction of the Pagani-Selinger-Valiron model by relating this model with the game model~\cite{Clairambault2019a}. 

The difference between the Pagani-Selinger-Valiron model and ours can be summarised as follows: Our model \( \pshCQ \) is the enriched presheaf over \( \CQ \), whereas the Pagani-Selinger-Valiron model is the hom-set-wise completion\footnote{The completion here means a free embedding of a given algebra \( M \) with partially defined sum into \( \overline{M} \), an algebra with totally-defined countable sum.  \citet{Pagani2014} used \emph{\(D\)-completion}.}
of \( \CQ \) followed by the Cauchy-completion.\footnote{More precisely the Cauchy-completion in the enriched sense.  The hom-set-wise completion yields \( \overline{\CPM} \), which is complete-\( \Sigma \)-monoid enriched.  The Cauchy-completion in the complete-\( \Sigma \)-monoid enriched sense is the completion by adding countable biproducts and splitting objects of idempotents.}
The enriched Cauchy-completion is known to give a full subcategory of the enriched presheaf category, so the crucial difference is the hom-set-wise completion.
The hom-set-wise completion introduces \( \infty \)'s, of which existence makes some arguments of this paper unapplicable.
One point is the analysis of the embedding-projection pairs, in which the cancelability \( x + z = y + z \Rightarrow x = y \) plays an important role.  Note that an algebra with totally-defined countable sum cannot be cancellable since \( 0 + 1 + 1 + \dots = 1 + 1 + \cdots \) (infinite sums of \( 1 \)) implies \( 0 = 1 \).
The second point is the full abstraction proof, in which the convergence (\ie~not being \( \infty \)) of a real-coefficient power series plays an essential role.
Nevertheless we believe that the Pagani-Selinger-Valiron model provides a fully abstract model of Quantum FPC.

\citet{Staton2015} studied a first-order quantum computation from the viewpoint of algebraic theory, giving an axiomatisation of the equational theory of (a variant of) the superoperator model.

As for recursive types in quantum programming languages, \citet[Section~3]{Rennela2020} discussed an approach to model recursive types by using the CPO-enrichment of a category of \( W^* \)-algebras.
Subsequently \citet{Pechoux2020} provided a model of a quantum calculus with inductive data types and \citet{Jia2022} a model of a quantum calculus with quantum inductive types and classical recursive types.
But even in the latter, the quantum language is separated from the classical part and this separation played a crucial role in the model construction.

\citet{AndresMartinez2023} studied another algebra with partial infinite sums to develop a model of a quantum programming language with loops.
Their algebra has weaker axioms than \( \Sigma \)-monoid, satisfying only a weaker form of the associativity law.
An algebra with partial infinite sum and weaker associativity law can also be found in \citet{Manes1986}.
Their algebra is further required to be \( \omega \)-complete (in the sense of \cref{def:sigma-monoid}) in order to apply techniques from the domain theory.
They studied a category enriched by their algebra; the connection between their and our approaches deserves further investigation. 

Denotational models of recursive types in the (non-quantum) linear \( \lambda \)-calculus was studied recently by \citet{Lindenhovius2021}.
They introduced a class of models and proves the adequacy of the models.
Unfortunately their result is not applicable to our setting (see \cref{rem:recursive-types-in-linear-setting}).

\tk{todo: write something about models of circuit description language.  What is interesting is that some of them used presheaves.  But our model considers the ``sum'' of circuits, so our approach may not be satisfactory to the community.}
 \section{Conclusion and Future Work}
\label{sec:conc}
We introduced the category \( \pshCQ \) of modules over superoperator and shows its relevance to quantum programs.
We caved out a model \( \classicalCQ \) of classical linear logic from \( \pshCQ \) and proved that \( \classicalCQ \) is an adequate and fully abstract model of Quantum FPC, a quantum \( \lambda \)-calculus with recursive types.

We expect that the construction of this paper is applicable to other \( \SMon \)-enriched symmetric monoidal category.
Given a \( \SMon \)-enriched monoidal category \( \mathcal{C} \) that describes the first-order effectful computation, the enriched presheaf category \( \widehat{\mathcal{C}} \) should be a model of higher-order computation with the same effect as \( \mathcal{C} \) under a mild assumption; the details are left for future work.
\asd{infinite dimensional model?} 
\subsection*{Acknowledgement}
The authors would also like to thank the anonymous reviewers for their helpful and constructive comments.
This work was supported by JSPS KAKENHI Grant Numbers JP20H05703
and ROIS NII Open Collaborative Research 2023-23FA05.

\bibliography{library}

\ifwithappendix
\clearpage
\appendix
\section{Supplementary Materials for Section~3}
\tk{\cref{sec:model}}

\begin{notation}
    Let \( M \) be a \( \Sigma \)-monoid and assume that \( \sum_i x_i = y \) holds in \( M \).
    To make explicit the ambient \( \Sigma \)-monoid \( M \), we sometimes write as \( M \models \sum_i x_i = y \).
    Similar notation is used for the canonical pre-order: \( M \models x \le y \) means that \( x,y \in M \) and \( x \le y \) holds in \( M \).
\eqed
\end{notation}

\subsection{On the Sum of Completely Positive Maps}
Let us first recall the definition of the (finite and infinite) sum of a family in \( \CPM(n,m) \).
A finite sum in \( \CPM(n,m) \) is always defined and an infinite sum \( \sum_{i \in \Nat} \varphi_i \) is defined as the limit \( \lim_{n \to \infty} \sum_{i=0}^n \varphi_i \) with respect to the standard topology on \( \Mat_n(\Complex) \cong \Complex^{n \times n} \).
\begin{lemma}\label{lem:cpm-sigma-monoid}
    \( \CPM(n,m) \) with the above sum is a \( \Sigma \)-monoid.
\end{lemma}
\begin{proof}[Proof]
    The key is the fact that the topological limit \( \lim_{n \to \infty} \varphi_i \) of an increasing chain \( \varphi_0 \le \varphi_1 \le \cdots \) is the least upper bound with respect to the order \( \le \)~\cite[Lemma~2.9]{Selinger2004a}; note that the order \( \le \) associated to the finite sum in \( \CPM(n,m) \) is the \emph{L\"owner order} on positive self-adhoint matrices.

    The sum \( \sum_{i \in \Nat} \varphi_i \) is preserved by a rearrangement \( \sigma \colon \Nat \longrightarrow \Nat \) because of the positivity of \( \varphi_i \).
    We show that a rearrangement \( \sigma \colon \Nat \longrightarrow \Nat \) does not affect the limit.
    Assume \( \varphi = \lim_{n \to \infty} \sum_{i=0}^n \varphi_i \).
    For each \( m \), let \( n = \max\{\sigma(j)\mid j\le m\} \).
    Then \( \sum_{j=0}^m \varphi_{\sigma(j)} \le \sum_{i=0}^{n} \varphi_i \le \varphi \) and hence \( \varphi \) is an upper bound of \( \sum_{j=0}^m \varphi_{\sigma(j)} \).
    By \cite[Remark in Section~4.1]{Selinger2004a}, an increasing sequence with an upper bound has the least upper bound.
    So \( \lim_{m \to \infty} \sum_{j=0}^m \varphi_{\sigma(j)} \le \varphi \).
    The same argument shows the reverse direction.
    
    The associativity and other axioms are easy to prove.
\end{proof}

\subsection{On Partial Actions}
Every \( \CQ \)-module has a \emph{partial action} of \( \Real_{\ge 0} \).
Given \( x \in M_n \) and \( r \ge 0 \), its action \( r \cdot x \) is a partially defined expression.
It is defined as \( (r/N) \cdot x + \dots + (r/N) \cdot x \) (\( N \) components) for some \( N \in \Nat \) with \( (r/N) \in [0,1] \), whose value is independent of the choice of \( N \in \Nat \).
Similarly \( M \) has a partial action of \( \CPM \): given \( x \in M_n \) and \( \varphi \in \CPM(m,n) \), the partially defined expression \( (x \cdot \varphi) \) is defined as \( (x \cdot (\varphi/N)) + \dots + (x \cdot (\varphi/N)) \) (\( N \) components).

\subsection{On Monomorphisms in $\pshCQ$}
A \( \CQ \)-module morphism \( f \colon M \longrightarrow N \) is a \emph{monomorphism} if \( f \circ g = f \circ h \) implies \( g = h \) for every \( \CQ \)-module \( L \) and \( g,h \colon L \longrightarrow M \).

\begin{lemma}\label{lem:monomorphism-characterisation}
    Let \( f \colon M \longrightarrow N \) be a morphism in \( \pshCQ \).
    The following conditions are equivalent.
    \begin{itemize}
        \item \( f \) is a monomorphism.
        \item \( f_n \in \SMon(M_n, N_n) \) is injective for each \( n \).
    \end{itemize}
\end{lemma}
\begin{proof}
    Suppose that \( f \) is a monomorphism.
    We show that \( f_n \colon M_n \longrightarrow N_n \) is an injection.
    Let \( x,y \in M_n \) and assume that \( f_n(x) = f_n(y) = z \in N_n \).
    By the Yoneda Lemma, \( x \) and \( y \) can be identified with morphisms \( \hat{x}, \hat{y} \colon \yoneda(n) \longrightarrow M \) such that \( x = \hat{x}_n(\ident_n) \) and \( y = \hat{y}_n(\ident_n) \).
    Then, for \( \varphi \in \CQ(m,n) \),
    \begin{align*}
        (f \circ \hat{x})_m(\varphi)
        &=
        ((f \circ \hat{x})_n(\ident_n)) \cdot \varphi
        \\
        &=
        (f_n(\hat{x}_n(\ident_n))) \cdot \varphi
        \\
        &=
        f_n(x) \cdot \varphi
        \\
        &=
        f_n(y) \cdot \varphi
        \\
        &=
        ((f \circ \hat{y})_n(\ident_n)) \cdot \varphi
        \\
        &=
        (f \circ \hat{y})_m(\varphi).
    \end{align*}
    Since \( \varphi \) is arbitrary, \( f \circ \hat{x} = f \circ \hat{y} \) and thus \( \hat{x} = \hat{y} \).
    Therefore \( x = \hat{x}_n(\ident_n) = \hat{y}_n(\ident_n) = y \).

    Suppose that \( f_n \) is injective for each \( n \).
    Let \( g,h \) be \( \CQ \)-module morphisms \( L \longrightarrow M \) and assume \( f \circ g = f \circ \varphi \).
    Then \( (f \circ g)_n = f_n \circ g_n \) and hence \( f_n \circ g_n = f_n \circ h_n \) for each \( n \).
    Since \( \Sigma \)-monoid homomorphisms are functions on underlying sets and the composition is the functional composition, if \( f_n \) is injective, then \( g_n = h_n \).
\end{proof}

By identifying a monomorphism \( f \colon M \hookrightarrow N \) and a family \( (\{ f_n(x) \in N_n \mid x \in M_n \} \subseteq N_n)_n \) of subsets, a monomorphism is called a \emph{\(\CQ\)-submodule} or \emph{submodule}.

We introduce some subclass of submodules.
Recall the canonical pre-order of a \( \Sigma \)-monoid, given by \( x \le y \defp \exists z. x+z=y \).
\begin{definition}
    Let \( \iota \colon M \hookrightarrow N \) be a \( \CQ \)-submodule.
    \begin{itemize}
        \item It is \emph{sum-reflecting} if, for every \( n \in \Nat \), \( x \in M_n \) and family \( (x_i)_i \) over \( M_n \), if \( \iota(x) = \sum_i \iota(x_i) \) in \( N_n \), then \( x = \sum_i x_i \) holds in \( M_n \).
        \begin{align*}
            &
            \forall n \in \Nat. \forall x \in M_n. \forall (x_i)_i \subseteq M_n.
            \\
            &\qquad
            N_n \models \iota(x) = \sum_i \iota(x_i)
            \quad\Longrightarrow\quad
            M_n \models x = \sum_i x_i.
        \end{align*}
        \item It is \emph{downward-closed} if \( M_n \subseteq N_n \) is a downward-closed subset for each \( n \), \ie,
        \begin{align*}
            &
            \forall n \in \Nat. \forall x \in M_n. \forall y \in N_n.
            \\
            &\qquad
            N_n \models y \le \iota(x)
            \quad\Longrightarrow\quad
            \exists y_0 \in M_n. \: y = \iota(y_0)
        \end{align*}
        \item It is \emph{hereditary} if it is sum-reflecting and downward-closed.
\eqed
    \end{itemize}
\end{definition}

\begin{remark}
    Given \( M \hookrightarrow N \), we shall often regard an element \( x \in M_n \) as an element of \( N_n \).
    But this convention is applied only if the omission of the embedding would not be confusing, and we sometimes prefer explicitly writing the embedding.
    A hereditary submodule is easy to handle, so the embedding of a hereditary submodule is mostly omitted.
    A confusing case is a non-downward-closed submodule, including the case of a regular submodule (\ie~an equaliser).
\qed
\end{remark}

\begin{lemma}\label{lem:some-monomorphisms-closed-under-composition}
    The class of sum-reflecting monomorphisms (resp.~hereditary monomorphisms) are closed under the composition.
\end{lemma}
\begin{proof}
    Let \( f \colon M \hookrightarrow N \) and \( g \colon N \hookrightarrow L \) be monomorphisms.
    
    Assume that \( f \) and \( g \) are sum-reflecting.
    Assume \( L \models g f(x) = \sum_i g f(x_i) \).
    By the sum-reflection of \( g \), we have \( N \models f(x) = \sum_i f(x_i) \).
    By the sum-reflection of \( f \), we have \( M \models x = \sum_i x_i \).

    Assume that \( f \) and \( g \) are hereditary.
    We have seen that \( g \circ f \) is sum-reflecting.
    Assume \( L \models y \le g f(x) \), \ie~\( L \models y + z = g f(x) \).
    By the downward-closedness of \( g \), we have \( y = g(y_0) \) and \( z = g(z_0) \) for some \( y_0 \) and \( z_0 \) in \( N \).
    Hence \( L \models g(y_0) + g(z_0) = g f(x) \).
    Since \( g \) is sum-reflecting, \( N \models y_0 + z_0 = f(x) \).
    By the downward-closedness of \( f \), we have \( y_0 = f(y_1) \) for some \( y_1 \) in \( M \).
\end{proof}

\tk{todo: metavariable renaming}
\begin{lemma}\label{lem:diagonal-monotone}
    Let \( I \) be a countable set and
    \( N, N_i \) and \( N', N'_i \) be \( \CQ \)-modules (\( i \in I \)).
    Assume morphisms \( f, f_i, e_i, \varphi_i, e'_i, \varphi'_i \) (\(i\in I\)) such that
    \begin{equation*}
        \xymatrix{
            N \ar[d]^f \ar[r]^{\varphi_i} & N_i \ar[d]^{f_i}
            \\
            N' \ar[r]^{\varphi'_i} & N'_i
        }
    \end{equation*}
    and
    \begin{equation*}
        \xymatrix{
            N \ar[d]^f & N_i \ar[d]^{f_i} \ar[l]^{e_i}
            \\
            N' & N'_i \ar[l]^{e'_i}
        }
    \end{equation*}
    commute for every \( i \in I \).
    Furthermore, assume that
    \begin{equation*}
        \ident_N = \sum_{i \in I} e_i \circ \varphi_i
    \end{equation*}
    and
    \begin{equation*}
        \ident_{N'} = \sum_{i \in I} e'_i \circ \varphi'_i
    \end{equation*}
    \begin{enumerate}
        \item If \( f_i \) is a monomorphism for every \( i \), then \( f \) is a monomorphism.
        \item If \( f_i \) is a sum-reflecting monomorphism for every \( i \), then \( f \) is a sum-reflecting monomorphism.
        \item If \( f_i \) is a hereditary monomorphism for every \( i \), then \( f \) is a hereditary monomorphism.
    \end{enumerate}
\end{lemma}
\begin{proof}
    (1)
    Let \( g, h \colon X \longrightarrow N \) and assume \( f \circ g = f \circ h \).
    Then
    \begin{equation*}
        \varphi'_i \circ f \circ g 
        =
        \varphi'_i \circ f \circ h
    \end{equation*}
    for every \( i \).
    By an assumed commuting diagram,
    \begin{equation*}
        f_i \circ \varphi_i \circ g 
        =
        f_i \circ \varphi_i \circ h
    \end{equation*}
    for every \( i \).
    Since \( f_i \) is monic,
    \begin{equation*}
        \varphi_i \circ g 
        =
        \varphi_i \circ h
    \end{equation*}
    for every \( i \), and thus
    \begin{equation*}
        e_i \circ \varphi_i \circ g 
        =
        e_i \circ \varphi_i \circ h.
    \end{equation*}
    Since \( \ident_N = \sum_i e_i \circ \varphi_i \),
    \begin{equation*}
        g
        =
        (\sum_i e_i \circ \varphi_i) \circ g
        \Kle
        \sum_i (e_i \circ \varphi_i \circ g)
    \end{equation*}
    and
    \begin{equation*}
        h
        =
        \sum_i (e_i \circ \varphi_i \circ h).
    \end{equation*}
    Hence \( g = h \).

    (2)
    Assume \( N' \models \sum_j f(x_j) = f(x) \) for some \( x_j,x \in N_n \).
    Then \( N_i' \models \sum_j \varphi'_i(f(x_j)) = \varphi'_i(f(x)) \) for every \( i \in I \).
    So \( N_i' \models \sum_j f_i(\varphi_i(x_j)) = f_i(\varphi_i(x)) \).
    Since \( f_i \) is sum-reflecting, \( N_i \models \sum_j \varphi_i(x_j) = \varphi_i(x) \).
    Since \( i \in I \) is arbitrary, \( N \models \sum_{i,j} (e_i \circ \varphi_i)(x_j) \Keq \sum_i (e_i \circ \varphi_i)(x) \).
    Both sides are defined since so is \( x = \sum_i (e_i \circ \varphi_i)(x) \).
    Therefore \( N \models \sum_j x_j = \sum_{i,j} (e_i \circ \varphi_i)(x_j) = \sum_i (e_i \circ \varphi_i)(x) = x \).

    (3)
    Assume that \( N' \models y \le f(x) \) for some \( x \in N_n \) and \( y \in N'_n \).
    Then \( N' \models y + z = f(x) \) for some \( z \in N'_n \).
    Then \( N'_i \models \varphi'_i(y) + \varphi'_i(z) = \varphi'_i(f(x)) = f_i(\varphi_i(x)) \) for every \( i \in I \).
    Since \( f_i \) is downward-closed, there exist \( y_i \) and \( z_i \) such that \( f_i(y_i) = \varphi'_i(y) \) and \( f_i(z_i) = \varphi'_i(z) \).
    Since \( y = \sum_i e_i'(\varphi_i'(y)) \) by the assumption, \( N' \models y = \sum_i e_i'(f_i(y_i)) = \sum_i f(e_i(y_i)) \).
    Similarly \( N' \models z = \sum_i f(e_i(z_i)) \).
    Therefore \( N' \models \sum_i (f(e_i(y_i)) + f(e_i(z_i))) = y + z = f(x) \).
    Since \( f \) is sum-reflecting (by the proof of case (2)), we have \( N \models \sum_i (e_i(y_i) + e_i(z_i)) = x \).
    In particular \( N \models \IsDef{(\sum_i e_i(y_i))} \).
    Let \( y' \defe \sum_i e_i(y_i) \).
    Then \( y = f(y') \) as desired.
\end{proof}

\begin{lemma}\label{lem:equaliser-sum-reflecting}
    An equaliser is a sum-reflecting submodule.
\end{lemma}
\begin{proof}
    Let \( f,g \colon M \longrightarrow N \).
    It is not difficult to see that the equaliser of \( f \) and \( g \) is the sum-reflecting submodule \( L \) of \( M \) consisting of
    \begin{equation*}
        L_k
        \quad\defe\quad
        \{\, x \in M_k \mid f(x) = g(x) \,\}.
    \end{equation*}
\end{proof}

\subsection{Basic Properties of $\pshCQ$}
\begin{lemma}\label{lem:real-action-invertible}
    Let \( M \) be a \( \CQ \)-module, \( n \in \Nat \) and \( x, y \in M_n \).
    If \( r\,x = r\,y \) for some \( r > 0 \), then \( x = y \).
\end{lemma}
\begin{proof}
    Assume that \( r^{-1} \) is a natural number.
    If it is not the case, choose \( r' \in (0,1] \) such that \( (r'r)^{-1} \) is a natural number.
    Of course, \( r\,x = r\,y \) implies \( (r'r)\,x = (r'r)\,y \).

    Let \( N = (1/r) \).
    Then \( \sum_{i=1}^N r = 1 \) holds in \( [0,1] \).
    Hence
    \begin{align*}
        x
        &=
        1\,x
        \\
        &=
        (\sum_{i=1}^N r) x
        \\
        &\Kle
        \sum_{i=1}^N r\,x
    \end{align*}
    and similarly
    \begin{equation*}
        y \quad=\quad \sum_{i=1}^N r\,y.
    \end{equation*}
    Hence \( r\,x = r\,y \) implies \( x = y \).
\end{proof}

\begin{lemma}\label{lem:tensor-peak}
    Let \( M \) and \( N \) be \( \CQ \)-modules.
    For every \( k \) and \( x \in (M \ptensor N)_k \), there exist \( y \in M_n \), \( z \in N_m \) and \( \varphi \in \CQ(k, n \otimes m) \) such that \( x \le (y \ptensor z) \cdot \varphi \).
\end{lemma}
\begin{proof}
    Let \( L \) be the hereditary submodule of \( M \ptensor N \) consisting of elements that satisfies the condition of this lemma.
    Note that \( L \) contains all elements of the form \( y \ptensor z \), \( y \in M_n \) and \( z \in N_m \).
    Then \( L \) satisfies the universal property of the tensor product.
    Given a bilinear map \( f \in \Bilin(M,N; K) \), we have a morphism \( L \hookrightarrow M \ptensor N \stackrel{\hat{f}}{\longrightarrow} K \), where \( \hat{f} \) is the morphism given by the universal property of \( M \ptensor N \).
    Conversely, given a morphism \( g \colon L \longrightarrow K \), we have a bilinear map \( f \) defined by \( f_{n,m}(x,y) \defe g_{n \otimes m}(x \ptensor y) \).
    Hence \( L \) and \( M \ptensor N \) are canonically isomorphic.
    Since the embedding \( L \hookrightarrow M \ptensor N \) is identity on the elements of the form \( y \ptensor z \), this embedding is one direction of the canonical embedding, which is an isomorphism.
    So the embedding is bijection on each component.
    This means that every element of \( M \ptensor N \) satisfies the condition of this lemma.
\end{proof}

\subsection{Finite Completion}
\begin{definition}[Finite completeness]
    A \( \CQ \)-module \( M \) is \emph{finitely complete} if the finite sum \( \sum_{i=1}^k x_i \) is defined in \( M_n \) for every \( n \) and \( x_1,\dots,x_k \in M_n \).
\eqed
\end{definition}
The forgetful functor from the full subcategory of finitely complete \( \CQ \)-modules has the left adjoint, which we call the \emph{finite completion}.
We write \( \FComp{M} \) for the finite completion of \( M \).

\begin{lemma}
    Every \( \CQ \)-module \( M \) has the finite completion \( \FComp{M} \).
\end{lemma}
\begin{proof}
    Note that every \( \Sigma \)-monoid has the finite completion.
    We write \( \FComp{X} \) for the finite completions of a \( \Sigma \)-monoid \( X \).
    Given a \( \Sigma \)-monoid morphism \( f \in \SMon(X, Y) \) to a finitely complete \( \Sigma \)-monoid \( Y \), the canonical map \( \FComp{X} \longrightarrow Y \) in \( \SMon \) is written as \( f^{\dagger} \).
    We also note that an element \( x \in \FComp{X} \) in the finite completion can be written as a finite sum \( x = \sum_{i = 1}^n \iota(x_i) \) of \( x_i \in X \) (where \( \iota \in \SMon(X, \FComp{X}) \) is the embedding).

    Let \( M \) be a \( \CQ \)-module.
    Its finite completion is given by \( (\FComp{M})_n \defe \FComp{M_n} \).
    Let \( \iota_n \in \SMon(M_n, \FComp{M_n}) \) be the embedding.
    Given a \( \CQ \)-morphism \( h \in \CQ(m,n) \), the action \( ({-}) \cdot h \in \SMon(\FComp{M_n}, \FComp{M_m}) \) is the canonical map \( (\iota_m(({-}) \cdot h))^{\dagger} \).
    It is easy to see that this satisfies the requirements for action.

    The \( \CQ \)-module \( \FComp{M} \) is obviously finitely complete.
    It is not difficult to see that the embedding \( \iota = (\iota_n)_n \) is a \( \CQ \)-module morphism, \ie~it preserves the action of \( \CQ \)-morphisms.
    
    Let \( N \) be a finitely complete \( \CQ \)-module and assume \( \alpha \colon M \longrightarrow N \).
    Then \( \alpha_n \in \SMon(M_n, N_n) \) is a \( \Sigma \)-monoid morphism to a finitely complete \( \Sigma \)-monoid, and hence it has a unique extension \( \alpha_n^{\dagger} \in \SMon(\FComp{M_n}, N_n) \).
    We show that the family \( \alpha^{\dagger} = (\alpha_n^{\dagger})_n \) is a \( \CQ \)-module morphism.
    Let \( h \in \CQ(m,n) \) and \( x \in (\FComp{M})_n = \FComp{M_n} \).
    Then \( x = \sum_{i = 1}^n \iota_n(x_i) \) for some \( x_i \in M_n \).
    Let us write \( \bullet \) for the action in \( \FComp{M} \).
    Then
    \begin{align*}
        x \bullet h
        &=
        (\sum_{i = 1}^n \iota_n(x_i)) \bullet h
        \\
        &=
        (\iota_m(({-})\cdot h))^{\dagger}(\sum_{i = 1}^n \iota_n(x_i))
        \\
        &\Kle
        \sum_{i = 1}^n (\iota_m(({-})\cdot h))^{\dagger}(\iota_n(x_i))
        \\
        &=
        \sum_{i = 1}^n \iota_m(x_i\cdot h)
    \end{align*}
    and thus
    \begin{align*}
        \alpha_m^{\dagger}(x \bullet h)
        &=
        \alpha_m^{\dagger}(\sum_{i = 1}^n \iota_m(x_i\cdot h))
        \\
        &\Kle
        \sum_{i = 1}^n \alpha_m^{\dagger}(\iota_m(x_i \cdot h))
        \\
        &=
        \sum_{i = 1}^n \alpha_m(x_i \cdot h)
        \\
        &=
        \sum_{i = 1}^n \alpha_n(x_i) \cdot h.
    \end{align*}
    Since
    \begin{align*}
        \alpha_n^{\dagger}(x)
        &=
        \alpha_n^{\dagger}(\sum_{i = 1}^n \iota_n(x_i))
        \\
        &\Kle
        \sum_{i = 1}^n \alpha_n^\dagger(\iota_n(x_i))
        \\
        &=
        \sum_{i = 1}^n \alpha_n(x_i),
    \end{align*}
    we have
    \begin{align*}
        \alpha_n^{\dagger}(x) \cdot h
        &=
        (\sum_{i = 1}^n \alpha_n(x_i)) \cdot h
        \\
        &\Kle
        \sum_{i = 1}^n \alpha_n(x_i) \cdot h.
    \end{align*}
    So \( \alpha^\dagger \) preserves the \( \CQ \)-action.

    Obviously \( \alpha^{\dagger} \circ \iota = \alpha \) as expected.
    The uniqueness of \( \alpha^{\dagger} \) can be proved component-wise using the same property for \( \Sigma \)-monoids.
\end{proof}
Since the finite completion \( \FComp{M} \) exists for every \( \CQ \)-module \( M \), by the universality of the finite completion, \( \FComp{({-})} \) can be extended to a functor.

\begin{lemma}\label{lem:finite-completion-of-yoneda}
    \( \FComp{\yoneda(n)} \cong \CPM({-}, n) \).
\end{lemma}
\begin{proof}
    We prove that the inclusion \( \iota \colon \yoneda(n) \longrightarrow \CPM({-}, n) \) is the finite completion.
    Assume a finitely complete \( \CQ \)-module \( M \) and a morphism \( f \colon \yoneda(n) \longrightarrow M \).

    Given \( \varphi \in \CPM({-},n)_m = \CPM(m,n) \), let \( \FComp{f}(\varphi) \defe \sum_{1 \le i \le N} f(\varphi/N) \) where \( N \) is a sufficiently large natural number such that \( \varphi/N \in \CQ(m,n) \).
    The sum is defined since \( M \) is finitely complete, and it is independent of the choice of \( N \).
    Obviously \( \FComp{f} \circ \iota = f \).

    Assume \( g \circ \iota = f \).
    For \( \varphi \in \CPM(m,n) \),
    \begin{equation*}
        g(\varphi)
        = g(\sum_{1 \le i \le N} \varphi/N)
        \Kle \sum_{1 \le i \le N} g(\varphi/N)
        = \sum_{1 \le i \le N} g(\ident_n) \cdot (\varphi/N)
        = \sum_{1 \le i \le N} f(\ident_n) \cdot (\varphi/N)
    \end{equation*}
    for sufficiently large \( N \).
    So \( g \) is completely determined by \( f(\ident_n) \) and thus \( g = \FComp{f} \).
\end{proof}

\begin{lemma}\label{lem:completion-hereditary}
    The finite completion \( M \longrightarrow \FComp{M} \) is a hereditary \( \CQ \)-submodule.
\end{lemma}
\begin{proof}
    We first prove the following claim.
    \begin{claim*}
        For each \( n \in \Nat \), there exist a \( \CQ \)-module \( M^{(n)}_\infty \) and a \( \CQ \)-module morphsim \( \alpha^{(n)} \colon M \longrightarrow M^{(n)}_\infty \) that satisfies the following conditions.
        \begin{itemize}
            \item \( M^{(n)}_\infty \) is countably complete (\ie~all countable sums are defined).
            \item The \( n \)-th component \( \alpha^{(n)}_n \) of \( \alpha^{(n)} \) is a sum-reflecting monomorphism.
        \end{itemize}
    \end{claim*}
    \begin{proof}
        For a \( \Sigma \)-monoid \( X \), let \( X_{\infty} \) be the \( \Sigma \)-monoid that additionally has the infinity \( \infty \) and all undefined sum in \( X \) converges to \( \infty \).
        That means, if \( \IsUndef{(\sum_i x_i)} \) in \( X \), then \( \sum_i x_i = \infty \) in \( X_\infty \).
        For a family \( (y_i)_i \) on \( X_{\infty} \) containing \( \infty \), their sum is again \( \infty \).
        Then \( X_{\infty} \) is a countably complete \( \Sigma \)-monoid.
    
        Given a \( \CQ \)-module, let \( M^{(n)}_{\infty} \) be the \( \CQ \)-module defined as follows.
        Its \( m \)-th component is \( (M^{(n)}_{\infty})_m \defe \SMon(\CQ(n,m), (M_n)_{\infty}) \).
        Given \( h \in \CQ(k,m) \) and \( f \in (M^{(n)}_\infty)_m = \SMon(\CQ(n,m), (M_n)_{\infty}) \), the action is \( (f \cdot h)(h') \defe f(h \circ h') \) for every \( h' \in \CQ(n,k) \).
Then \( M^{(n)}_\infty \) is a countably complete \( \CQ \)-module.
    
        There exists a \(\CQ\)-module morphism \( \alpha^{(n)} \colon M \longrightarrow M^{(n)}_\infty \) defined by, 
        for \( x \in M_m \) and \( h \in \CQ(n,m) \),
        \begin{equation*}
            \alpha^{(n)}_m(x)(h) \defe x \cdot h.
        \end{equation*}
        We check that this mapping preserves the action.
        Assume \( x \in M_m \) and \( h \in \CQ(k,m) \).
        Then \( \alpha^{(n)}_k(x \cdot h), \alpha^{(n)}_m(x) \cdot h \in (M^{(n)}_{\infty})_k = \SMon(\CQ(n,k), (M_n)_{\infty}) \).
        For every \( h' \in \CQ(n,k) \),
        \begin{align*}
            \alpha^{(n)}_k(x \cdot h)(h')
            &=
            (x \cdot h) \cdot h'
            \\
            &=
            x \cdot (h \circ h')
            \\
            &=
            \alpha^{(n)}_m(x)(h \circ h')
            \\
            &=
            (\alpha^{(n)}_m(x) \cdot h)(h').
        \end{align*}
        It is easy to see the preservation of sums.
        So \( \alpha^{(n)} \) is a \( \CQ \)-module morphism.

        The former has already been proved.
        We prove the latter.
        Note that \( \alpha^{(n)}_n(x)(\ident_n) = x \cdot \ident_n = x \) for every \( x \in M_n \).

        \( \alpha^{(n)}_n \) is an injection.
        Actually, if \( x,y \in M_n \) and \( x \neq y \), then \( \alpha^{(n)}_n(x)(\ident_n) = x \cdot \ident_n = x \neq y = \alpha^{(n)}_n(y) \).
        Since an injection is a monomorphism in \( \SMon \), \( \alpha^{(n)}_n \) is a monomorphism.

        \( \alpha^{(n)}_n \) is sum-reflecting.
        Assume that \( \alpha^{(n)}(x) = \sum_i \alpha^{(n)}(x_i) \).
        By the definition of the hom-sets of \( \SMon \), we have \( \alpha^{(n)}(x)(\ident_n) = \sum_i \alpha^{(n)}(x_i)(\ident_n) \) in \( (M_n)_{\infty} \).
        This means \( x = \sum_i x_i \) in \( (M_n)_\infty \).
        Since \( x \neq \infty \), the same equation holds in \( M_n \).
\qedhere[\textbf{Claim}]
    \end{proof}

    Let \( \iota \colon M \longrightarrow \FComp{M} \) be the finite completion.
    By the universality of \( \FComp{M} \) and finite completeness of \( M^{(n)}_\infty \), there exists \( \beta^{(n)} \colon \FComp{M} \longrightarrow M^{(n)}_\infty \) such that \( \beta^{(n)} \circ \iota = \alpha^{(n)} \).

    We prove that \( \iota \) is a monomorphism.
    The \( n \)-th component \( \beta^{(n)}_n \circ \iota_n \) of \( \beta^{(n)} \circ \iota \) must be injective since so is \( \alpha^{(n)}_n \).
    Because \( n \) is arbitrary, \( \iota \) consists of injections.
    By \cref{lem:monomorphism-characterisation}, \( \iota \) is a monomorphism.

    We prove that \( \iota \) is sum-reflecting.
    Assume that \( \iota(x) = \sum_i \iota(x_i) \) in \( \FComp{M}_n \) for some \( x, x_i \in M_n \).
    Then \( \alpha^{(n)}(x) = \beta^{(n)}\iota(x) = \beta^{(n)}(\sum_i \iota(x_i)) \Kle \sum_i \beta^{(n)}\iota(x_i) = \sum_i \alpha^{(n)}(x_i) \).
    Since \( \alpha^{(n)} \) is sum-reflecting, \( x = \sum_i x_i \) holds in \( M_n \).

    We prove that \( \iota \) is downward-closed.
    Assume \( y \le \iota(x) \) in \( \FComp{M}_n \) for some \( x \in M_n \) and \( y \in \FComp{M}_n \).
    By definition, \( y + z \le \iota(x) \) for some \( z \in \FComp{M}_n \).
    Recall that \( \FComp{M}_n \) is generated by \( M_n \).
    Hence \( y = \sum_i \iota(y_i) \) and \( z = \sum_j \iota(z_j) \) for some families \( (y_i)_i, (z_j)_j \) on \( M_n \).
    We have \( \sum_i \iota(y_i) + \sum_j \iota(z_j) = \iota(x) \) in \( \FComp{M}_n \).
    Since \( \iota \) is sum-reflecting, \( \sum_i y_i + \sum_j z_j = x \) holds in \( M_n \).
    So \( \sum_i y_i \) is defined in \( M_n \), and obviously \( \iota(\sum_i y_i) = y \).

\end{proof}

\begin{lemma}
    The finite completion is a faithful functor.
\end{lemma}
\begin{proof}
    Because the unit \( M \longrightarrow \FComp{M} \) of the adjunction is a monomorphism (\cref{lem:completion-hereditary}).
\end{proof}

\begin{lemma}\label{lem:tensor-preserves-finite-completeness}
    For finitely complete \( \CQ \)-modules \( M \) and \( M' \), its tensor product \( M \otimes M' \) is also finitely complete.
\end{lemma}
\begin{proof}
    Let \( (x_i)_{i = 1}^n \) be a finite family over \( (M \otimes N)_k \).
    By \cref{lem:tensor-peak}, \( x_i \le (y_i \otimes y'_i) \cdot h_i \) for each \( i \) where \( y_i \in M_{m_i} \), \( y'_i \in M'_{m'_i} \) and \( h_i \in \CQ(k, m_i \otimes m'_i) \).
    Let \( m = \max_i m_i \) and \( m' = \max_i m'_i \).

    Note that, for \( \ell, \ell' \in CQ \) with \( \ell \le \ell' \), there is a retraction \( \ell \lhd \ell' \) in \( \CQ \), \ie~a pair \( \mathsf{p} \in \CQ(\ell',\ell) \) and \( \mathsf{e} \in \CQ(\ell,\ell') \) of morphisms such that \( \mathsf{p} \circ \mathsf{e} = \ident_\ell \).
    Indeed \( \mathsf{e} \) maps an \( (\ell \times \ell) \)-matrix \( X \) to the \( (\ell' \times \ell') \)-matrix whose top-left corner is \( X \) (and other entries are \( 0 \)) and \( \mathsf{p} \) extracts the top-left \( (\ell \times \ell) \)-submatrix.

    Let \( (\mathsf{p}_i, \mathsf{e}_i) \) be an retraction \( m_i \lhd m \) and \( (\mathsf{p}'_i, \mathsf{e}'_i) \) an retraction \( m'_i \lhd m' \).
    Then \( (\mathsf{p}_i \otimes \mathsf{p}'_i, \mathsf{e}_i \otimes \mathsf{e}'_i) \) is a retraction \( (m_i \otimes m'_i) \lhd (m \otimes m') \).
    So we have
    \begin{align*}
        (y_i \otimes y'_i) \cdot h_i
        &=
        (y_i \otimes y'_i) \cdot ((\mathsf{p}_i \otimes \mathsf{p}'_i) \circ (\mathsf{e}_i \otimes \mathsf{e}'_i) \circ h_i)
        \\
        &=
        ((y_i \cdot \mathsf{p}_i), (y'_i \cdot \mathsf{p}'_i)) \cdot ((\mathsf{e}_i \otimes \mathsf{e}'_i) \circ h_i).
    \end{align*}
    Let
    \begin{align*}
        z_i &\defe y_i \cdot \mathsf{p}_i \\
        z'_i &\defe y'_i \cdot \mathsf{p}'_i \\
        g_i &\defe (\mathsf{e}_i \otimes \mathsf{e}'_i) \circ h_i.
    \end{align*}
    Then \( x_i \le (z_i \otimes z'_i) \cdot g_i \), where \( z_i \in M_m \), \( z'_i \in M'_m \) and \( g_i \in \CQ(k, m \otimes m') \).

    By the finite completeness of \( M \) and \( M' \), both \( z = (\sum_{i=1}^n z_i) \) and \( z' = \sum_{i = 1}^n z'_i \) are defined.
    Again by the finite completeness, \( nz = \sum_{i = 1}^n z \) is also defined.
    Since \( \CQ(k, m \otimes m') \) is convex, \( g \defe \sum_{i = 1}^n \frac{1}{n} g_i \) is defined.
    Then
    \begin{align*}
        (nz \otimes z') \cdot g
        &=
        (z \otimes z') \cdot ng
        \\
        &\Kle
        \sum_{i,i'=1}^n (nz_i \otimes z'_{i'}) \cdot (\sum_{i''=1}^n g_i)
        \\
        &\Kle
        \sum_{i,i',i''=1}^n (z_i \otimes z'_{i'}) \cdot g_{i''}.
    \end{align*}
    As its partial sum, \( \sum_{i=1}^n (z_i \otimes z'_i) \cdot g_i \) is defined.
    Since \( x_i \le (z_i \otimes z'_i)\cdot g_i \), the sum \( \sum_i x_i \) is also defined.
\end{proof}

\begin{lemma}\label{lem:implication-preserves-completeness}
    If \( N \) is finitely complete, then so is \( M \multimap N \).
\end{lemma}
\begin{proof}
    Since the sum in \( M \multimap N \) is defined by the point-wise sum.
\end{proof}

Since \( \FComp{M} \otimes \FComp{N} \) is finitely complete by \cref{lem:tensor-preserves-finite-completeness}, there exists a canonical map \( \FComp{M \otimes N} \longrightarrow \FComp{M} \otimes \FComp{N} \).
It is actually an isomorphism.
\begin{lemma}\label{lem:tensor-preserves-completion}
    Let \( \iota_M \colon M \longrightarrow \FComp{M} \) and \( \iota_N \colon N \longrightarrow \FComp{N} \) be finite completions.
    Then \( (\iota_M \otimes \iota_N)^{\dagger} \colon \FComp{M \otimes N} \longrightarrow \FComp{M} \otimes \FComp{N} \) is an isomorphism.
\end{lemma}
\begin{proof}
    We give the inverse.
    Let \( \iota_{M \otimes N} \colon M \otimes N \longrightarrow \FComp{M \otimes N} \) be the finite completion.
    It induces \( M \longrightarrow (N \multimap \FComp{M \otimes N}) \) and \( \FComp{M} \longrightarrow (N \multimap \FComp{M \otimes N}) \) by the finite completeness of \( N \multimap \FComp{M \otimes N} \) (\cref{lem:implication-preserves-completeness}).
    So we have \( \FComp{M} \otimes N \longrightarrow \FComp{M \otimes N} \).
    Similarly we have \( \alpha \colon \FComp{M} \otimes \FComp{N} \longrightarrow \FComp{M \otimes N} \).
    Note that \( \alpha(x \otimes y) = \iota(x \otimes y) \) on \( (x \otimes y) \in (M \otimes N)_n \).

    By computing the image of generators \( x \otimes y \), \( x \in M_m \) and \( y \in N_n \), of \( M \otimes N \), we can check that the following diagram commutes:
    \begin{equation*}
        \xymatrix@C=100pt{
            M \otimes N \ar[r]^{\iota} \ar[ddr]_{\iota} & \FComp{M \otimes N} \ar[d]^{(\iota_M \otimes \iota_N)^\dagger}
            \\
            & \FComp{M} \otimes \FComp{N} \ar[d]^{\alpha}
            \\
            & \FComp{M \otimes N}
        }.
    \end{equation*}
    By the universal property of \( \FComp{M \otimes N} \), we have \( \ident = \alpha \circ (\iota_M \otimes \iota_N)^{\dagger} \).

    The other equation \( (\iota_M \otimes \iota_N)^{\dagger} \) can be obtained similarly, since \( \FComp{M} \otimes \FComp{N} \) is generated by \( \iota_M(x) \otimes \iota_N(y) \), \( x \in M_m \) and \( y \in N_n \).
    \tk{todo}
\end{proof}

\tk{todo: define generators.  Note that if the generator is closed under the \( \CQ \)-action, then each \( M_n \) is generated by a sum.}

\subsection{Orthogonal Factorisation System}
A morphism \( f \colon M \longrightarrow N \) is a \emph{covering morphism} if, for every \( n \in \Nat \) and \( y \in N_n \), there exists \( x \in M_n \) such that \( N \models y \le f(x) \).
\begin{lemma}\label{lem:orthogonal-factorisation-system}
    \( (\{\mbox{covering morphisms}\}, \{\mbox{hereditary monos}\}) \) is an orthogonal factorisation system.
\end{lemma}
\begin{proof}
    We first check the lifting property.
    Assume a commutative square
    \begin{equation*}
        \xymatrix{
            M \ar[d]^{c} \ar[r]^f & L \ar[d]^{m} \\
            N \ar[r]^g & K
        }
    \end{equation*}
    such that \( m \) is a hereditary monomorphism and \( c \) is a covering morphism.
    Then, for every \( n \in \CQ \), it induces a commutative square
    \begin{equation*}
        \xymatrix{
            M_n \ar[d]^{c_n} \ar[r]^{f_n} & L_n \ar[d]^{m_n} \\
            N_n \ar[r]^{g_n} & K_n
        }
    \end{equation*}
    where \( m_n \) is a hereditary monomorphism and \( c_n \) is a covering morphism in \( \SMon \).
    So we have a unique diagonal fill-in \( k_n \colon N_n \longrightarrow L_n \).\tk{to do: provide a detail in a paper}
    The family \( (k_n)_n \) preserves the \( \CQ \)-action since so does the composite \( (m_n \circ k_n)_n = (g_n)_n \) and \( m_n \) is an injection.
    So \( (k_n)_n \) is a \( \CQ \)-module homomorphism.

    Assume that \( m \colon L \longrightarrow K \) is a morphism satisfying the right lifting property for every covering morphism.
    \begin{itemize}
        \item
            To prove \( m \) is monic, let \( a,b \in L_n \) and assume \( m_n(a) = m_n(b) \).
            Let \( M = \yoneda(n) \coprod \yoneda(n) \) and \( N = \yoneda(n) \).
            Then \( M \) is generated by two elements \( x,y \in M_n \) and \( N \) is generated by a single element \( z \in N_n \).
            Let \( c \) be the map defined by \( c = [\ident_N, \ident_N] \) or, equivalently, \( x \mapsto z \) and \( y \mapsto z \).
            Obviously \( c \) is a covering morphism.
            Let \( f \) and \( g \) be \( CQ \)-module homomorphisms given by \( x \stackrel{f}{\mapsto} a \), \( y \stackrel{f}{\mapsto} b \) and \( z \stackrel{g}{\mapsto} m_n(a) \).
            Then the square commutes, so the diagonal fill-in \( k \colon N \longrightarrow L \) exists.
            So
            \begin{equation*}
                a = f_n(x) = k_n(c_n(x)) = k_n(c_n(y)) = f_n(y) = b.
            \end{equation*}
            Since \( n \) and \( a,b \in L_n \) are arbitrary, \( m_n \) is an injection for every \( n \).
            So \( m \) is a mono.
        \item 
            To prove \( m \) is downward-closed, let \( a,b \in K_n \) and assume \( a \le b \) (\ie~\( a + a' = b \)) and \( b \in L_n \) (here we regard \( L_n \subseteq K_n \) via the injection \( m_n \)).
            Let \( M = \yoneda(n) \) and \( N = M \times M \).
            So \( M \) is the \( \CQ \)-module generated by an element \( x \in M_n \) and \( N \) is the \( \CQ \)-module generated by elements \( y,z,u \in N_n \) and the relation \( y + z = u \).
            Consider a morphism \( c \) given by \( x \mapsto u \) or, equivalently, \( c = \langle \ident_M, \ident_M \rangle \).
            This is a covering morphism.
            Let \( f \colon M \longrightarrow L \) be the morphism given by \( x \mapsto a \) and \( g \colon N \longrightarrow K \) be the morphism given by \( y \mapsto a \), \( z \mapsto a' \) and \( u \mapsto b \).
            Then the square commutes, so there exists a diagonal fill-in \( k \colon N \longrightarrow L \).
            Since \( a = m_n(k_n(y)) \) and \( m_n \) is regarded as an inclusion, we conclude that \( a \) belongs to \( L_n \).
        \item 
            To prove \( m \) is sum-reflecting, let \( (a_i)_i \) be a family of elements in \( L_n \), \( b \in L_n \) and assume \( m_n(b) = \sum_i m_n(a_i) \) in \( K_n \).
            Let \( M = \yoneda(n) + \coprod_{i} \yoneda(n) \) and \( N = \prod_{i} \yoneda(n) \).
            So \( M \) is the \( \CQ \)-module generated by \( x_i,y \in M_n \) (and the sum of generators is undefined), and \( N \) is the \( \CQ \)-module presented by the generators \( x'_i, y' \in N_n \) and the relation \( y' = \sum_i x'_i \).
            Consider the morphism \( c \colon M \longrightarrow N \) defined by \( x_i \mapsto x'_i \) and \( y \mapsto y' \) (note that \( c \) is well-defined since the domain is the coproduct).
            The image of \( c \) contains the ``maximum element'' \( y' \), so \( c \) is a covering morphism.
            Let \( f \colon M \longrightarrow L \) be the mapping given by \( x_i \mapsto a_i \) and \( y \mapsto b \), and \( g \colon N \longrightarrow K \) be the mapping given by \( x'_i \mapsto m_n(a_i) \) and \( y' \mapsto m_n(b) \).
            Then the square commutes, so there exists a diagonal fill-in \( k \colon N \longrightarrow L \).
            Then
            \begin{equation*}
                b
                =
                f(y)
                =
                k(c(y)) 
                =
                k(y')
                =
                k(\sum_i x'_i)
                =
                \sum_i k(x'_i)
                =
                \sum_i k(c(x_i))
                =
                \sum_i f(x_i)
                =
                \sum_i a_i.
            \end{equation*}
    \end{itemize}

    Assume that \( c \colon M \longrightarrow N \) is a morphism satisfying the left lifting property for every hereditary monomorphism.
    Let \( L \hookrightarrow N \) be the sum-reflecting submodule of \( N \) consisting of elements covered by the image of \( c \), \ie,
    \begin{equation*}
        L_n \defe \{ x \in N_n \mid \exists y \in M_n. x \le f_n(y) \}.
    \end{equation*}
    It is easy to see that the restriction of the operations in \( N \) to \( L \) forms a \( \CQ \)-module.
    By definition, \( m \colon L \hookrightarrow N \) is a hereditary monomorphism.
    Consider the following square
    \begin{equation*}
        \xymatrix{
            M \ar[d]^{c} \ar[r]^c & L \ar[d]^{m} \\
            N \ar[r]^{\ident} & N
        }.
    \end{equation*}
    This square commutes, so there exists a diagonal fill-in \( k \colon N \longrightarrow L \).
    Regarding \( L_n \) as a subset of \( N_n \), the morphism \( k_n \colon N_n \longrightarrow L_n \) is the identity.
    So \( L = N \).
\end{proof}

\subsection{Proof of \cref{prop:local-presentability}}
\begin{claim*}[of \cref{prop:local-presentability}]
    \( \pshCQ \) is locally \( \aleph_1 \)-presentable.
\end{claim*}

We prove that \( \pshCQ \) is locally \( \aleph_1 \)-presentable.

Recall that a \( \Sigma \)-monoid can be axiomatised~\cite{Tsukada2022} as a single-sorted partial Horn theory~\cite{Palmgren2007}.
A \( \CQ \)-module \( M \in \pshCQ \) consists of a family \( (M_n)_{n \in \Nat} \) of \( \Sigma \)-monoids, which is naturally an \( \Nat \)-sorted algebra.
Its operations are the sums for each sort \( n \), and unary operation \( ({-}) \cdot f \) of sort \( n \to m \) for each \( f \in \CQ(m,n) \).
The axioms are \( \Sigma \)-monoid axioms for each sort \( n \) together with \( (x \cdot f) \cdot g = x \cdot (f \circ g) \) and \( (\sum_i x_i) \cdot (\sum_j f_j) \Kle \sum_{i,j} x_i \cdot f_j \).
It is not difficult to see that \( \widehat{\CQ} \) is equivalent to the category of models of this \( \Nat \)-sorted algebras.

\subsection{Proof of \cref{thm:free-exponential}}
\begin{claim*}[of \cref{thm:free-exponential}]
    \( \pshCQ \) has the cofree exponential comonad.
\end{claim*}

As stated immediately before the statement, this theorem is a consequence of the following facts.
\begin{itemize}
    \item \( \pshCQ \) is locally presentable (\cref{prop:local-presentability}).
    \item \( \pshCQ \) is a symmetric monoidal-closed category (\cref{sec:intuitionistic:day-tensor,sec:intuitionistic:function-space}).
    \item Every locally-presentable (symmetric) monoidal-closed category has the cofree cocommutative comonoids~\cite[Remarks~1 in Section~2.7]{Porst2008}.
\end{itemize}

 \section{Supplementary Materials for Section~4}
\tk{\cref{sec:matrix}}

\subsection{Basic fact}
\begin{lemma}\label{lem:basis-transitive}
    Assume that \( M \) has an \( \mathcal{O} \)-basis and every \( N \in \mathcal{O} \) has an \( \mathcal{O}' \)-basis.
    Then \( M \) has an \( \mathcal{O}' \)-basis.
\end{lemma}
\begin{proof}
    Assume a \( \mathcal{O} \)-basis \( (\BObj{b}, \ket{b}, \bra{b})_{b \in B} \) of \( M \).
    Choose a \( \mathcal{O}' \)-basis \( (\BObj{a}, \ket{a}, \bra{a})_{a \in A_b} \) for each \( \BObj{b} \).
    Then, for every \( b \in B \) and \( a \in B_a \),
    \begin{align*}
        \ket{b} \circ \ket{a} &\colon \BObj{a} \longrightarrow \BObj{b} \longrightarrow M \\
        \bra{a} \circ \bra{b} &\colon M \longrightarrow \BObj{b} \longrightarrow \BObj{a}.
    \end{align*}
    Since
    \begin{equation*}
        \ident_M
        \quad=\quad
        \sum_{b \in B} \ket{b} \bra{b}
        \quad=\quad
        \sum_{b \in B} \ket{b} (\sum_{a \in A_b} \ket{a}\bra{a}) \bra{b}
        \quad\Kle\quad
        \sum_{b \in B, a \in A_b} \ket{b}\ket{a}\bra{a}\bra{b}.
    \end{equation*}
    Hence \( (\BObj{b,a}, \ket{b,a}, \bra{b,a})_{b \in B, a \in A_b} \) given by \( \BObj{b,a} \defe \BObj{a} \), \( \ket{b,a} \defe \ket{b}\ket{a} \) and \( \bra{b,a} \defe \bra{a} \bra{b} \) is an \( \mathcal{O}' \)-basis for \( M \).
\end{proof}

\subsection{Proof of \cref{thm:cpm-representation}}
\begin{claim*}[of \cref{thm:cpm-representation}]
    Let \( \LQT \) and \( \LQT' \) be pseudo-representable \( \CQ \)-modules.
    Then
\begin{equation*}
        \CQ(\LQT, \LQT')
        \quad\cong\quad
        \{ \varphi \in \CPM(\#\LQT,\#\LQT') \mid \forall n. \forall x \in \mathcal{L}_n. \varphi \circ x \in \mathcal{L}'_n \}
        \qquad\mbox{(as \( [0,1] \)-modules).}
    \end{equation*}
Here \( \varphi \in \CPM(\#\LQT,\#\LQT') \) satisfying the above condition corresponds to a morphism \( f_\varphi = (f_{\varphi,n})_n \colon \LQT \longrightarrow \LQT' \) given by \( f_{\varphi,n}(x) \defe \varphi \circ x \) for every \( n \) and \( x \in \LQT_n \).
\end{claim*}

Let \( \ell \) and \( \ell' \) be \( \#\LQT \) and \( \#\LQT' \), respectively.

Clearly \( f_\varphi \) is a morphism from \( \LQT \) to \( \LQT' \).

It suffices to show that every morphism \( g \in \pshCQ(\LQT,\LQT') \) is represented by a completely positive map \( \psi \in \CPM(\ell,\ell') \).
Since \( \LQT \) is pseudo-representable, \( r\,\ident_\ell \in \LQT_\ell \) for some \( r \in (0,1] \).
Let \( \psi_0 \defe g_\ell(r\,\ident_\ell) \in \LQT'_\ell \subseteq \CPM(\ell,\ell') \) and \( \psi \defe (1/r) \psi_0 \in \CPM(\ell,\ell') \).
Let \( x \in \LQT_n \subseteq \CPM(n,\ell) \).
We show that \( g_n(x) = \psi \circ x \).
Let \( r' \in (0,1] \) such that \( r'\,x \in \CQ(n,\ell) \).
Then
\begin{equation*}
    rr'\,g_n(x) = g_n(rr'\,x) = g_\ell(r\,\ident_\ell) \cdot (r'\,x) = (r\,\psi_0) \cdot (r'\,x).
\end{equation*}
Since the \( \CQ \)-action to \( \LQT \) is the composition in \( \CPM \), we have \( rr'\,g_n(x) = rr'\,(\psi \circ x) \) and thus \( g_n(x) = \psi \circ x \) since \( rr' \neq 0 \) (\cf~\cref{lem:real-action-invertible}).

\subsection{Proof of \cref{lem:linear:lqt-function-and-tensor}}
\begin{claim*}[of \cref{lem:linear:lqt-function-and-tensor}]
    Let \( \LQT \) and \( \LQT' \) be pseudo-representable \( \CQ \)-modules.
Let \( (\LQT \rightarrowtriangle \LQT') \hookrightarrow \CPM({-}, \#\LQT \multimap \#\LQT') \) and \( (\LQT \otensor \LQT') \hookrightarrow \CPM({-}, \#\LQT \otimes \#\LQT') \) be the hereditary \(\CQ\)-submodules given by
    \begin{align*}
        (\LQT \rightarrowtriangle \LQT')_n &:= \{\, \varphi \in \CPM(n, \#\LQT \multimap \#\LQT') \mid
        \forall m. \forall x \in \mathcal{L}_m. \mathbf{ev} \circ (\varphi \otimes x) \in \mathcal{L}'_{n \otimes m} \,\} \\
        (\LQT \otensor \LQT')_n &:= \{\, (x \otimes x') \circ \varphi \in \CPM(n, \#\LQT \otimes \#\LQT') \mid
        x \in \mathcal{L}_m, x' \in \mathcal{L}'_{m'}, \varphi \in \CQ(n, m \otensor m') \,\}.
    \end{align*}
    Then
\( \LQT \rightarrowtriangle \LQT' \) and \( \LQT \otensor \LQT' \) are pseudo-representable.
Furthermore
        \( (\LQT \rightarrowtriangle \LQT') \cong (\LQT \multimap \LQT') \) and \( (\LQT \otensor \LQT') \cong \LQT \ptensor \LQT' \).
\end{claim*}

The proof is long.
We split the whole arguments into lemmas.

\begin{lemma}\label{lem:appx:lqt-tensor1}
    Let \( \LQT \) and \( \LQT' \) be pseudo-representable \( \CQ \)-modules.
    Then \( \LQT \otensor \LQT' \) is pseudo-representable.
\end{lemma}
\begin{proof}
    Let \( \ell \) and \( \ell' \) be the underlying objects of \( \LQT \) and \( \LQT' \).
    If \( B \) and \( B' \) are upper bounds for \( \LQT \) and \( \LQT' \), respectively, then \( B B' \) is an upper bound for \( \LQT \otensor \LQT' \).
    If \( r\,\ident_\ell \in \LQT_\ell \) and \( r'\,\ident_{\ell'} \in \LQT_{\ell'} \) for some \( r,r' > 0 \), then \( rr'\,(\ident_{\ell} \otimes \ident_{\ell'}) \in (\LQT \otensor \LQT')_{\ell \otimes \ell'} \).
\end{proof}

\begin{lemma}\label{lem:appx:lqt-tensor2}
    \( \LQT \otensor \LQT' \) is the
    tensor product \( \LQT \ptensor \LQT' \) in \( \pshCQ \).
\end{lemma}
\begin{proof}
    Let \( \ell \) and \( \ell' \) be the underlying objects of \( \LQT \) and \( \LQT' \), respectively.
    We show that \( \LQT \otensor \LQT' \) is a representing object of \( \Bilin(\LQT, \LQT'; {-}) \).

    Assume a morphism \( \alpha \colon \LQT \otensor \LQT' \longrightarrow M \) in \( \widehat{\CQ} \).
    For \( x \in \LQT_n \) and \( x' \in \LQT'_{n'} \), we define \( \beta_{n,n'}(x,x') \defe \alpha_{n \otimes n'}(x \ptensor x') \in M_{n \otimes n'} \).
    Then \( \beta = (\beta_{n,n'})_{n,n'} \in \Bilin(\LQT, \LQT'; M) \).
    It is easy to see that this mapping \( \alpha \mapsto \beta \) is natural in \( M \).

    Conversely, given \( (\beta_{n,m})_{n,m} \in \Bilin(\LQT, \LQT'; M) \),
    we define a \( \CQ \)-module morphism \( \alpha = (\alpha_m)_{m} \colon \LQT \otensor \LQT' \longrightarrow M \) in \( \widehat{\CQ} \) by \( \alpha_m((x \ptensor x') \circ h) \defe \beta_{n,n'}(x,x') \cdot h \) for \( x \in \LQT_n \), \( x' \in \LQT_{n'} \) and \( h \in \CQ(m, n \otimes n') \).
    We prove that \( \alpha \) is actually a \( \CQ \)-module morphism.
Suppose \( (x \ptensor x') \circ h = \sum_i (x_i \ptensor x'_i) \circ h_i \) holds in \( (\LQT \otensor \LQT')_m \), where \( x \in \LQT_n \), \( x' \in \LQT'_{n'} \), \( x_i \in \LQT_{n_i} \) and \( x'_i \in \LQT'_{n'_i} \).
Let \( B \) and \( B' \) be upper bounds of \( \LQT \) and \( \LQT' \), respectively.
    Then \( (x/B \ptensor x'/B') \circ h = \sum_i (x_i/B \ptensor x'_i/B') \circ h_i \) holds in \( \CQ(m, \ell \otimes \ell') \).
    If \( r\,\ident_\ell \in \LQT_\ell \) and \( r'\,\ident_{\ell'} \in \LQT'_{\ell'} \),
    \begin{align*}
        &\frac{rr'}{BB'} \alpha_m((x \ptensor x') \circ h)
        \\
        \quad=\quad&
        \frac{rr'}{BB'} \beta_{n,n'}(x, x') \cdot h
        \\
        \quad=\quad&
        \beta_{n,n'}\big((r\,\ident_{\ell}) \circ (x/B), (r'\,\ident_{\ell'}) \circ (x'/B')) \cdot h
        \\
        \quad=\quad&
        \beta_{\ell,\ell'}(r\,\ident_\ell, r'\,\ident_{\ell'}) \cdot ((x/B \otimes x'/B') \circ h))
        \\
        \quad=\quad&
        \beta_{\ell,\ell'}(r\,\ident_\ell, r'\,\ident_{\ell'}) \cdot (\sum_i (x_i/B \otimes x'_i/B') \circ h_i)
        \\
        \quad\Kle\quad&
        \sum_i \beta_{\ell,\ell'}(r\,\ident_\ell, r'\,\ident_{\ell'}) \cdot ((x_i/B \otimes x'_i/B') \circ h_i)
        \\
        \quad=\quad&
        \sum_i \beta_{n_i,n'_i}((r/B)\,x_i, (r'/B')\,x'_i) \cdot h_i
        \\
        \quad=\quad&
        \sum_i \alpha_m\big(((r/B) x_i \ptensor (r'/B')x'_i) \circ h_i\big)
        \\
        \quad=\quad&
        \frac{rr'}{BB'} \sum_i \alpha_m\big((x_i \ptensor x'_i) \circ h_i).
    \end{align*}
    Since \( rr'/BB' > 0 \), we have \( \alpha_m((x \ptensor x') \circ h) = \sum_i \alpha_m((x_i \ptensor x'_i) \circ h_i) \).
    Hence \( \alpha_m \) is well-defined and a \( \Sigma \)-monoid homomorphism for each \( m \).
    It is easy to see that \( \alpha = (\alpha_m)_m \) respects the \( \CQ \)-action.
    Hence \( \alpha \) is a \( \CQ \)-module homomorphism.

    So we have a bijection \( \Bilin(\LQT, \LQT'; M) \cong \widehat{\CQ}(\LQT \otensor \LQT', M) \) natural in \( M \).
    It is actually a \( \Sigma \)-monoid homomorphism.
\end{proof}

The analysis of the linear function space is even harder.
Let us start from a topological observation needed to the proof.

Let \( \SelfAdjoint(\Mat_n(\Complex)) \) and \( \SelfAdjoint_+(\Mat_{n}(\Complex)) \) be the sets of self-adjoint matrices and of self-adjoint and positive matrices in \( \Mat_{n}(\Complex) \), the set of all \( (n \times n) \)-matrices.
\begin{lemma}\label{lem:self-ajdoint-closedness}
    \( \SelfAdjoint_+(\Mat_n(\Complex)) \) is a closed subset of \( \Mat_n(\Complex) \) (in the topological sense).
\end{lemma}
\begin{proof}
    Let us regard \( \Mat_n(\Complex) \) as a \( 2n^2 \)-dimensional \( \Real \)-vector space.
    Since it is finite dimensional, all norms on \( \Mat_n(\Complex) \) induce the same topology.
    
    Let \( A = (s_{i,j} + t_{i,j} \sqrt{-1})_{1 \le i,j \le n} \) be a matrix.
    It is self-adjoint if
    \begin{itemize}
        \item \( t_{i,i} = 0 \) for every \( i \),
        \item \( s_{i,j} = s_{j,i} \) for every \( i \neq j \), and
        \item \( t_{i,j} = - t_{j,i} \) for every \( i \neq j \).
    \end{itemize}
    Each of the above conditions defines a closed subset, as its intersection,
    \( \SelfAdjoint(\Mat_n(\Complex)) \) is a closed subset of \( \Mat_n(\Complex) \).
    The positivity means that \( v^* A v \ge 0 \) for every \( v \in \Complex^n \).
    Since \( \{ A \in \Mat_n(\Complex) \mid v^* A v \ge 0 \} \) is a closed subset for each \( v \in \Complex^n \), as the intersection of closed subsets, \( \SelfAdjoint_+(\Mat_n(\Complex)) \) is closed.
\end{proof}

\begin{lemma}\label{lem:cpm-trace-norm-ball-compact}
    \( \{ f \in \CPM(1, \ell \multimap \ell') \mid \opnorm{f} = 1 \} \) is compact.
\end{lemma}
\begin{proof}
    We use the following facts.
    \begin{itemize}
        \item The operator norm of \( f \in \CPM(1, n) \) is the trace of the corresponding element \( f \in \SelfAdjoint_+(\Mat_n(\Complex)) \).
        \item The trace \( \trace x \) on \( \SelfAdjoint_+(\Mat_n(\Complex)) \) is extended to a norm on \( \Mat_n(\Complex) \).
          This norm is called the \emph{trace norm}, defined as follows.
          For every \( b \in \SelfAdjoint_+(\Mat_n(\Complex)) \), there exists a unique \( c \in \SelfAdjoint_+(\Mat_n(\Complex)) \) such that \( cc = b \).
          This \( c \) is written as \( \sqrt{b} \).
          For \( a \in \Mat_n(\Complex) \), the trace norm of \( a \) is defined as \( \trace \sqrt(a^* a) \).
        \item On a finite dimensional vector space, any norms are equivalent to each other.  That means, if \( \opnorm{-}_1 \) and \( \opnorm{-}_2 \) are norms on the space, then for some \( 0 < C < C' \), we have \( \forall x. C \opnorm{x}_1 \le \opnorm{x}_2 \le C' \opnorm{x}_1 \).
    \end{itemize}

    Therefore \( \{ f \in \CPM(1, \ell \multimap \ell') \mid \opnorm{f} = 1 \} \) is a bounded subset with respect to the standard Euclidean norm (regarded as \( (\ell\ell')^2 \)-dimensional \( \Complex \)-vector space, for example).
    By the continuity of the trace norm, its unit sphere is closed.
    Since \( \CPM(1, \ell \multimap \ell') \) is closed by \cref{lem:self-ajdoint-closedness}, the set in the statement is a bounded closed subset of \( \Complex^{(\ell\ell')^2} \).
    Hence it is compact.
\end{proof}

\begin{lemma}\label{lem:appx:lqt-function1}
    Let \( \mathcal{L} \) and \( \mathcal{L}' \) be pseudo-representable \( \CQ \)-modules.
    Then \( \mathcal{L} \rightarrowtriangle \mathcal{L}' \) is pseudo-representable.
\end{lemma}
\begin{proof}
    Let \( \ell \) and \( \ell' \) be the underlying objects of \( \mathcal{L} \) and \( \mathcal{L}' \).

    We show that \( r\,\ident_{\ell \multimap \ell'} \in (\LQT \rightarrowtriangle \LQT')_{\ell \multimap \ell'} \) for some \( r > 0 \).
    Let \( B \) be an upper bound of \( \{ \opnorm{x} \mid n \in \Nat, x \in \LQT_n \} \).
    Then, for every \( m \) and \( x \in \LQT_m \),
    \begin{equation*}
        \opnorm{\mathbf{ev} \circ (\ident \otimes x)} \le \opnorm{\mathbf{ev}}\,\opnorm{\ident}\,\opnorm{g} \le \opnorm{\mathbf{ev}}\,\opnorm{\ident}\,B.
    \end{equation*}
    So letting \( 1/r = \opnorm{\mathbf{ev}}\,\opnorm{\ident}\,B \), we have \( r\,\ident \in (\LQT \rightarrowtriangle \LQT')_{\ell \multimap \ell'} \).

    We prove that there exists \( C \) such that \( f \in (\LQT \rightarrowtriangle \LQT')_n \) implies \( \opnorm{f} \le C \).
    Since \( \{ f \in \CPM(1,  \ell \multimap \ell') \mid \opnorm{f} = 1 \} \) is compact (\cref{lem:cpm-trace-norm-ball-compact}) and the mapping \( f \mapsto \opnorm{\mathbf{ev} \circ (f \otensor \ident_{\ell'})} \) is continuous,
    the image of the map of the unit sphere
    \begin{equation*}
        \{ \opnorm{\mathbf{ev} \circ (x \otensor \ident_{\ell'})} \mid f \in \CPM(1, \ell \multimap \ell'), \opnorm{f} = 1 \}
    \end{equation*}
    has the minimum \( m \in [0,\infty) \).
    Since \( \mathbf{ev} \circ (f \otensor \ident_{\ell'}) = \Lambda^{-1}(f) \), where \( \Lambda \colon \CPM(1, \ell \multimap \ell') \stackrel{\cong}{\longrightarrow} \CPM(\ell, \ell') \) is the canonical bijection, we have
    \begin{align*}
        \opnorm{\mathbf{ev} \circ (f \otensor \ident_{\ell'})} = 0
        \quad
        &\Longrightarrow\quad
        \mathbf{ev} \circ (f \otensor \ident_{\ell'}) = 0
        \\ &\Longrightarrow\quad
        \Lambda^{-1}(f) = 0
        \\ &\Longrightarrow\quad
        f = 0.
    \end{align*}
    Hence \( m \neq 0 \) and \( \opnorm{\mathbf{ev} \circ (f \otensor \ident_{\ell'})} \ge m \opnorm{f} \) for every \( f \in \CPM(1, \ell \multimap \ell') \).
    Assume that \( g \in (\LQT \rightarrowtriangle \LQT')_n \).
    There exists \( y \in \CQ(1, n) \) such that \( \opnorm{g} = \opnorm{g \circ y} \) (since \( \CQ(1,n) \) is compact and \( \opnorm{g} = \sup_{y \in \CQ(1,n)} \opnorm{g \circ y} \)).
    Then \( g \circ y \in (\LQT \rightarrowtriangle \LQT')_1 \) and thus \( \mathbf{ev} \circ ((g \circ y) \otimes (r'\,\ident_{\ell'})) \in \LQT'_{\ell'} \) for some \( r' > 0 \).
    For an upper bound \( B' \) for \( \LQT' \), we have
    \begin{equation*}
        B'
        \ge \opnorm{\mathbf{ev} \circ ((g \circ y) \otimes (r'\,\ident_{\ell'}))}
        \ge r'm \opnorm{g \circ  y}
        = r'm \opnorm{g}.
    \end{equation*}
    So \( \opnorm{g} \le B'/r'm \) for every \( g \in (\LQT \rightarrowtriangle \LQT')_n \).
\end{proof}

\begin{lemma}\label{lem:appx:lqt-function2}
    \( \LQT \rightarrowtriangle \LQT' \) is the linear implication \( \LQT \multimap \LQT' \) in \( \widehat{\CQ} \).
\end{lemma}
\begin{proof}
    Let \( \ell \) and \( \ell' \) be the underlying objects of \( \LQT \) and \( \LQT' \), respectively.
    We have \( \Sigma \)-monoid isomorphisms
    \begin{align*}
        (\LQT \multimap \LQT')_n
        &\cong
        \widehat{\CQ}(\yoneda(n), \LQT \multimap \LQT')
        \\
        &\cong
        \widehat{\CQ}(\yoneda(n) \otimes \LQT, \LQT')
        \\
        &\cong
        \Bilin(\yoneda(n), \LQT; \LQT').
    \end{align*}
    The action of \( h \in \CQ(m,n) \) to \( \alpha \in \Bilin(\yoneda(n), \LQT; \LQT') \) is \( (\alpha \cdot h)({-},{-}) \defe \alpha(h \circ {-}, {-}) \).

    Let \( B \) be an upper bound of the norm for \( \LQT \), \( r > 0 \) be a real number such that \( r\,\ident_\ell \in \LQT_{\ell} \) and \( \Lambda \colon \CPM(n \otimes \ell, \ell') \longrightarrow \CPM(n, \ell \multimap \ell') \) be the canonical map.

    Given \( \alpha \in \Bilin(\yoneda(n), \LQT; \LQT') \), we have
    \begin{equation*}
        \alpha_{n,\ell}(\ident_n, r\,\ident_{\ell}) \in \LQT'_{n \otimes \ell} \subseteq \CPM(n \otimes \ell, \ell')
    \end{equation*}
    and \( (1/r) \Lambda(\alpha_{n,\ell}(\ident_n, r\,\ident_{\ell})) \in \CPM(n,\ell \multimap \ell') \).
    That \( (1/r) \Lambda(\alpha_{n,\ell}(\ident_n, r\,\ident_{\ell})) \in (\LQT \rightarrowtriangle \LQT')_n \) follows from, for every \( x \in \LQT_{m} \),
    \begin{align*}
        &\mathrm{ev} \circ ((1/r) \Lambda(\alpha_{n,\ell}(\ident_n, r\,\ident_{\ell})) \otimes x)
        \\
        ={}&
        (1/r)\,\alpha_{n,\ell}(\ident_n, r\,\ident_{\ell}) \circ (\ident_n \otimes x)
        \\
        ={}&
        (B/r)\,\alpha_{n,\ell}(\ident_n, r\,\ident_{\ell}) \circ (\ident_n \otimes (x/B))
        \\
        ={}&
        (B/r)\,\alpha_{n,m}(\ident_n, (r/B)\,x)
        \\
        ={}&
        \alpha_{n,m}(\ident_n, x)
    \end{align*}
    and \( \alpha_{n,m}(\ident_n, x) \in \LQT'_{n \otimes m} \).
    The mapping \( \alpha \mapsto (1/r) \Lambda(\alpha_{n,\ell}(\ident_n, r\,\ident_{\ell})) \) preserves the \( \CQ \)-action and the sum.
    So it is a \( \CQ \)-module homomorphism \( (\LQT \multimap \LQT') \longrightarrow (\LQT \rightarrowtriangle \LQT') \).

    Conversely, given \( g \in (\LQT \rightarrowtriangle \LQT')_n \), we define \( \Psi_n(g) \in \Bilin(\yoneda(n), \LQT; \LQT') \) by
    \begin{equation*}
        (\Psi_n(g))_{m,m'}(h, x) \defe \mathrm{ev} \circ ((g \circ h) \otimes x)
    \end{equation*}
    for \( h \in \CQ(m, n) \) and \( x \in \LQT_{m'} \).
    It is easy to see that \( \Psi_n(g) \in \Bilin(\yoneda(n), \LQT; \LQT') \).
    The family \( \Psi = (\Psi_n)_{n} \) preserves the action, \ie~\( \Psi_n(g) \cdot k = \mathrm{ev} \circ ((g \circ k \circ ({-}) \otimes ({-})) = \Psi_{n'}(g \circ k) \) for every \( k \in \CQ(n',n) \).
    Hence \( \Psi \) is a \( \CQ \)-module morphism \( (\LQT \rightarrowtriangle \LQT') \longrightarrow (\LQT \multimap \LQT') \).
    
    The above constructed \( \CQ \)-modules morphisms are the inverses of each other, the desired isomorphism holds.
\end{proof}

\begin{proof}[Proof of \cref{lem:linear:lqt-function-and-tensor}]
    A consequence of \cref{lem:appx:lqt-tensor1,lem:appx:lqt-tensor2,lem:appx:lqt-function1,lem:appx:lqt-function2}.
\end{proof}

\subsection{Proof of \cref{prop:pseudo-basis-implies-basis}}
\begin{claim*}[of \cref{prop:pseudo-basis-implies-basis}]
    A \(\CQ\)-module \( M \) has a pseudo-representable basis if and only if \( M \in \basedCQ \).
\end{claim*}

We use the following lemma.
\begin{lemma}\label{lem:lqt-basis}
    A pseudo-representable \( \CQ \)-module has a representable basis.
\end{lemma}
\begin{proof}
    Let \( \LQT \hookrightarrow \CPM({-}, \ell) \) be a pseudo-representable module, \( B \) be an upper bound of the norm for \( \LQT \) and \( r > 0 \) be a real number such that \( r\,\ident_{\ell} \in \LQT_\ell \).
    We can assume without loss of generality that \( N \defe B/r \) is a natural number.
    If it is not the case, take any \( B' > B \) such that \( B'/r \) is a natural number.

    Then \( \LQT \) has a basis \( (n_i, e_i, \varphi_i)_{i = 1,...,N} \) where \( n_i = \ell \), \( e_i = r\,\ident_\ell \) and \( \varphi_i(x) = (1/B)\,x \) for every \( i \).
    Actually \( x = \sum_{i = 1}^{N} (r\,\ident_\ell) \cdot ((1/B)\,x) \) for every \( n \) and \( x \in \LQT_n \).
\end{proof}

Assume that \( M \) has a pseudo-representable basis \( (\LQT_i, e_i, f_i)_{i \in I} \).
By \cref{lem:lqt-basis}, \( \LQT_i \) has a representable basis \( (n_{i,j}, e_{i,j}, f_{i,j})_{j \in J_i} \).
Then
\begin{align*}
    & \ident_M
    \\
    &=
    \sum_i e_i \circ f_i
    \\
    &=
    \sum_i e_i \circ \ident_{\LQT_i} \circ f_i
    \\
    &=
    \sum_i e_i \circ (\sum_{j \in J_i} e_{i,j} \circ f_{i,j}) \circ f_i
    \\
    &\Kle
    \sum_{i \in I, j \in J_i} e_i \circ e_{i,j} \circ f_{i,j} \circ f_i
\end{align*}
Because
\begin{align*}
    e_i \circ e_{i,j} &\colon \yoneda(n_{i,j}) \longrightarrow M \\
    f_{i,j} \circ f_i &\colon M \longrightarrow \yoneda(n_{i,j}),
\end{align*}
\( (n_{i,j}, e_i \circ e_{i,j}, f_{i,j} \circ f_i)_{i \in I,j\in J_i} \) is a representable basis.

\subsection{Proof of \cref{lem:based-additive-multiplicative}}
\begin{claim*}[of \cref{lem:based-additive-multiplicative}]
    Let \( M,N \in \basedCQ \) and \( (\BObj{a}, \ket{a}, \bra{a})_{a \in \Base{M}} \) and \( (\BObj{b}, \ket{b}, \bra{b})_{b \in \Base{N}} \) be pseudo-representable bases.
    Then \( M \times N \), \( M \amalg N \), \( M \ptensor N \) and \( M \multimap N \) have the following bases.
    \begin{gather*}
        \Base{M \times N} \defe \{\, (a \times \bullet) \mid a \in \Base{M} \,\} \cup \{\, (\bullet \times b) \mid b \in \Base{N} \,\}
        \\
        \Base{M \amalg N} \defe \{\, \TInl(a) \mid a \in \Base{M} \,\} \cup \{\, \TInr(b) \mid b \in \Base{N} \,\}
        \\
        \Base{M \ptensor N} \defe \{\, a \ptensor b \mid a \in \Base{M}, b \in \Base{N} \,\}
        \\
        \Base{M \multimap N} \defe \{\, a \multimap b \mid a \in \Base{M}, b \in \Base{N} \,\}
        \\
        \BObj{(a \times \bullet)} \defe \BObj{a}
        \qquad
        \ket{a \times \bullet} \defe \Inj_1 \circ \ket{a}
        \qquad
        \bra{a \times \bullet} \defe \bra{a} \circ \Proj_1
        \\
        \BObj{(\bullet \times b)} \defe \BObj{b}
        \qquad
        \ket{\bullet \times b} \defe \Inj_2 \circ \ket{b}
        \qquad
        \bra{\bullet \times b} \defe \bra{b} \circ \Proj_2
        \\
        \BObj{(\TInl(a))} \defe \BObj{a}
        \qquad
        \ket{\TInl(a)} \defe \Inj_1 \circ \ket{a}
        \qquad
        \bra{\TInl(a)} \defe \bra{a} \circ \Proj_1
        \\
        \BObj{(\TInr(b))} \defe \BObj{b}
        \qquad
        \ket{\TInr(b)} \defe \Inj_2 \circ \ket{b}
        \qquad
        \bra{\TInr(b)} \defe \bra{b} \circ \Proj_2
        \\
        \BObj{(a \ptensor b)} \defe (\BObj{a} \ptensor \BObj{b})
        \qquad
        \ket{a \ptensor b} \defe (\ket{a} \ptensor \ket{b})
        \qquad
        \bra{a \ptensor b} \defe (\bra{a} \ptensor \bra{b})
        \\
        \BObj{(a \multimap b)} \defe (\BObj{a} \multimap \BObj{b})
        \qquad
        \ket{a \multimap b} \defe (\bra{a} \multimap \ket{b})
        \qquad
        \bra{a \multimap b} \defe (\ket{a} \multimap \bra{b}).
    \end{gather*}
    The above bases are orthogonal if so are \( \Base{M} \) and \( \Base{N} \).
\end{claim*}

By a direct calculation of the required equation: \( \ident = \sum_{b \in \Base{M}} \ket{a} \bra{a} \).
We prove the case of the tensor product as an example.

Since \( \sum_{a \in \Base{M}} \ket{a}\bra{a} = \ident_M \) and \( \sum_{b \in \Base{N}} \ket{b}\bra{b} = \ident_N \) and \( \ptensor \) is \( \SMon \)-enriched,
\begin{align*}
    \ident_{M \ptensor N}
    &=
    \ident_M \ptensor \ident_N
    \\
    &=
    \left( \sum_{a \in \Base{M}} \ket{a}\bra{a} \right) \ptensor \left( \sum_{b \in \Base{N}} \ket{b}\bra{b} \right)
    \\
    &\Kle
    \sum_{a \in \Base{M}, b \in \Base{N}} (\ket{a}\bra{a})\ptensor(\ket{b}\bra{b})
    \\
    &=
    \sum_{a \in \Base{M}, b \in \Base{N}} (\ket{a} \ptensor \ket{b}) (\bra{a} \ptensor\bra{b})
    \\
    &=
    \sum_{c \in \Base{M \ptensor N}} \ket{c} \bra{c}
\end{align*}
Assume that \( \Base{M} \) and \( \Base{N} \) are orthogonal.
Assume \( c,c' \in \Base{M \ptensor N} \) and \( c \neq c' \).
By definition of \( \Base{M \ptensor N} \), we have \( c = a \ptensor b \) and \( c' = a' \ptensor b' \).
Since \( c \neq c' \), we have \( a \neq a' \) or \( b \neq b' \).
If \( a \neq a' \), then \( \braket{a}{a'} = 0 \) by the orthogonality of \( \Base{M} \).
So \( \braket{c}{c'} = (\bra{a} \ptensor \bra{b})(\ket{a'} \ptensor \ket{b'}) = \braket{a}{a'} \ptensor \braket{b}{b'} = 0 \).

 \section{The Structure of Cofree Exponential}
\tk{todo: need rework}
\subsection{Useful Lemmas}
\begin{lemma}\label{lem:cpm-identity-cover}
    For each \( n \in \Nat \), there exists \( m \), \( x_1, \dots, x_m \in \CPM(1,n) \) and \( h_1,\dots,h_m \in \CPM(n,1) \) such that
    \begin{equation*}
        \ident_n
        \quad\le\quad
        \sum_i x_i \circ h_i.
    \end{equation*}
\end{lemma}
\begin{proof}
    In this proof, we use \( i \) to mean \( \sqrt{-1} \).

    The case of \( n = 2 \) is essential.
    We define completely positive maps \( f_1,f_2,f_3,f_4 \in \CPM(2,1) \).
    For
    \begin{equation*}
        A =
        \left(\begin{array}{cc}
            a & b \\ c & d
        \end{array}\right),
    \end{equation*}
    we define the maps by
    \begin{align*}
        f_1 &\colon A \mapsto a \\
        f_2 &\colon A \mapsto d \\
        f_3 &\colon A \mapsto a+b+c+d \\
        f_4 &\colon A \mapsto a - ib + ic+d.
    \end{align*}

    Then
    \begin{align*}
        A &=
        \left(\begin{array}{cc}
            2 & 0 \\ 0 & 1
        \end{array}\right)
        \cdot f_1(A)
        \quad-\quad
        \left(\begin{array}{cc}
            2 & 1+i \\ 1-i & 2
        \end{array}\right)
        \cdot \frac{f_1(A)}{2}
        \\
        &+
        \left(\begin{array}{cc}
            1 & 0 \\ 0 & 2
        \end{array}\right)
        \cdot f_2(A)
        \quad-\quad
        \left(\begin{array}{cc}
            2 & 1+i \\ 1-i & 2
        \end{array}\right)
        \cdot \frac{f_2(A)}{2}
        \\
        &+
        \left(\begin{array}{cc}
            1 & 1 \\ 1 & 1
        \end{array}\right)
        \cdot \frac{f_3(A)}{2}
        \quad-\quad
        \left(\begin{array}{cc}
            1 & 0 \\ 0 & 1
        \end{array}\right)
        \cdot \frac{f_3(A)}{2}
        \\
        &+
        \left(\begin{array}{cc}
            1 & i \\ -i & 1
        \end{array}\right)
        \cdot \frac{f_4(A)}{2}
        \quad-\quad
        \left(\begin{array}{cc}
            1 & 0 \\ 0 & 1
        \end{array}\right)
        \cdot \frac{f_4(A)}{2}.
    \end{align*}
    Note that all matrices are positive and all functions are completely positive.
    Hence
    \begin{align*}
        A &\le
        \left(\begin{array}{cc}
            2 & 0 \\ 0 & 1
        \end{array}\right)
        \cdot f_1(A)
        \quad+\quad
        \left(\begin{array}{cc}
            1 & 0 \\ 0 & 2
        \end{array}\right)
        \cdot f_2(A)
        \\
        &+
        \left(\begin{array}{cc}
            1 & 1 \\ 1 & 1
        \end{array}\right)
        \cdot \frac{f_3(A)}{2}
        \quad+\quad
        \left(\begin{array}{cc}
            1 & i \\ -i & 1
        \end{array}\right)
        \cdot \frac{f_4(A)}{2}.
    \end{align*}
    Therefore
    \begin{align*}
        \ident_2 &\le
        \left(\begin{array}{cc}
            2 & 0 \\ 0 & 1
        \end{array}\right)
        \cdot f_1
        \quad+\quad
        \left(\begin{array}{cc}
            1 & 0 \\ 0 & 2
        \end{array}\right)
        \cdot f_2
        \\
        &+
        \left(\begin{array}{cc}
            1 & 1 \\ 1 & 1
        \end{array}\right)
        \cdot \frac{f_3}{2}
        \quad+\quad
        \left(\begin{array}{cc}
            1 & i \\ -i & 1
        \end{array}\right)
        \cdot \frac{f_4}{2}.
    \end{align*}
\end{proof}

\subsection{Flatness of Modules with Bases}
One of the difficulties in the module theory is that the tensor product \( ({-}) \otimes M \) does not necessarily preserve  monomorphisms.
A module \( M \) such that \( ({-)} \otimes M \) preserves monomorphisms are called a \emph{flat module} (in the standard module theory).

\begin{definition}[Flat, strongly flat]
    A \( \CQ \)-module \( M \) is \emph{flat} if \( ({-}) \otimes M \) preserves monomorphisms, \ie~\( \iota \otimes M \colon N \otimes M \longrightarrow L \otimes M \) is a monomorphism for every monomorphism \( \iota \colon N \hookrightarrow L \).
    A flat \( \CQ \)-module \( M \) is \emph{strongly flat} if \( ({-}) \otimes M \) preserves submodules, sum-reflecting submodules and hereditary submodules, \ie~\( \iota \otimes M \colon N \otimes M \longrightarrow L \otimes M \) is a submodule (resp.~sum-reflecting submodule, hereditary submodule) whenever so is \( \iota \colon N \hookrightarrow L \).
\end{definition}

If \( M \) is a dualisable object, then it is strongly flat.
In a symmetric monoidal category, an object \( X \) is \emph{dualisable} if there exist an object \( X^* \) and morphisms \( \eta \colon I \longrightarrow X \otimes X^* \) and \( \epsilon \colon X^* \otimes X \longrightarrow I \), where \( I \) is the tensor unit, such that both
\begin{equation*}
    X
    \:\cong\:
    I \otimes X
    \:\stackrel{\eta \otimes X}{\longrightarrow}\:
    X \otimes X^* \otimes X
    \:\stackrel{X \otimes \epsilon}{\longrightarrow}\:
    X \otimes I
    \:\cong\:
    X
\end{equation*}
and
\begin{equation*}
    X^*
    \:\cong\:
    X^* \otimes I
    \:\stackrel{X^* \otimes \eta}{\longrightarrow}\:
    X^* \otimes X \otimes X^*
    \:\stackrel{\epsilon \otimes X^*}{\longrightarrow}\:
    I \otimes X^*
    \:\cong\:
    X^*
\end{equation*}
are identities.
Then
\begin{equation*}
    \mathrm{Hom}(A \otimes X, B)
    \cong
    \mathrm{Hom}(A, X^* \otimes B)
\end{equation*}
natural in \( A \) and \( B \).

A dualisable object is flat.
Assume \( M \) is dualisable and \( \iota \colon X \hookrightarrow Y \).
Let \( \alpha,\beta \colon A \longrightarrow X \otimes M \) such that \( (\iota \otimes M) \circ \alpha = (\iota \otimes M) \circ \beta \).
By the transpose, \( \iota \circ (\alpha \otimes M^*) = \iota \circ (\beta \otimes M^*) \).
Since \( \iota \) is monic, \( \alpha \otimes M^* = \beta \otimes M^* \).
By the transpose followed by the composition of the counit, we have \( \alpha = \beta \).

Unfortunately many \( \CQ \)-modules of interest are not dualisable.
We introduce a weaker version of this notion.
\begin{definition}[Weakly dualisable]
    A \( \CQ \)-module \( M \) is \emph{weakly dualisable} if there exist a \( \CQ \)-module \( M^* \) and morphisms \( \eta \colon \yoneda(1) \longrightarrow X \otimes X^* \) and \( \epsilon \colon X^* \otimes X \longrightarrow \yoneda(1) \) such that both
    \begin{equation*}
        X
        \:\cong\:
        \yoneda(1) \otimes X
        \:\stackrel{\eta \otimes X}{\longrightarrow}\:
        X \otimes X^* \otimes X
        \:\stackrel{X \otimes \epsilon}{\longrightarrow}\:
        X \otimes \yoneda(1)
        \:\cong\:
        X
    \end{equation*}
    and
    \begin{equation*}
        X^*
        \:\cong\:
        X^* \otimes \yoneda(1)
        \:\stackrel{X^* \otimes \eta}{\longrightarrow}\:
        X^* \otimes X \otimes X^*
        \:\stackrel{\epsilon \otimes X^*}{\longrightarrow}\:
        \yoneda(1) \otimes X^*
        \:\cong\:
        X^*
    \end{equation*}
    are \( r \cdot ({-}) \) for some \( r > 0 \).
    (If \( r = 1 \), then \( M \) is dualisable.)
\end{definition}

\begin{lemma}
    \( \yoneda(n) \) is weakly dualisable.
\end{lemma}
\begin{proof}
    Recall that \( \CPM \) is a compact closed category, and hence every object \( n \in \CPM \) is dualisable.
    Hence there exists \( \eta \in \CPM(1, n \otimes n) \) and \( \epsilon \in \CPM(n \otimes n, 1) \) that satisfies the required laws in \( \CPM \) with \( r = 1 \).
    Let \( \eta' \defe \eta/\opnorm{\eta} \) and \( \epsilon' \defe \epsilon/\opnorm{\epsilon} \).
    Then \( \eta' \in \CQ(1, n \otimes n) \) and \( \epsilon' \in \CQ(n \otimes n, 1) \) witnesses that \( n \) in \( \CQ \) is weakly dualisable with \( n^* = n \) and \( r = \frac{1}{\opnorm{\eta}\opnorm{\epsilon}} \).
    Because the Yoneda embedding preserves the Day tensor product, by the Yoneda Lemma,
    \begin{equation*}
        \CQ(1, n \otimes n)
        \:\cong\:
        \pshCQ(\yoneda(1), \yoneda(n \otimes n))
        \:\cong\:
        \pshCQ(\yoneda(1), \yoneda(n) \otimes \yoneda(n))
    \end{equation*}
    and
    \begin{equation*}
        \CQ(n \otimes n, 1)
        \:\cong\:
        \pshCQ(\yoneda(n \otimes n), \yoneda(1))
        \:\cong\:
        \pshCQ(\yoneda(n) \otimes \yoneda(n), \yoneda(1)).
    \end{equation*}
    So the images of \( \eta' \) and \( \epsilon' \) by these isomorphisms satisfy the requirement.
\end{proof}

\begin{lemma}\label{lem:cpm-dualisable}
    \( \CPM({-}, n) \) is weakly dualisable.
\end{lemma}
\begin{proof}
    Obivously they are dualisable.
\end{proof}

\begin{lemma}\label{lem:weak-compact-closed-flat}
    Let \( M \) be a weakly dualisable \( \CQ \)-module.
    Then \( M \) is strongly flat.
\end{lemma}
\begin{proof}
    Similar to the case of dualisable object, which is mentioned before.

    Consider, for example, the case of sum-reflection.
    Assume a sum-reflecting submodule \( \iota \colon X \longrightarrow Y \) and
    \begin{equation*}
        (M \otimes \iota) \circ y = \sum_i (M \otimes \iota) \circ x_i.
    \end{equation*}
    Transposing the both sides, we have
    \begin{equation*}
        \iota \circ (\epsilon \otimes X) \circ (M^* \otimes y)
        =
        \sum_i \iota \circ (\epsilon \otimes X) \circ (M^* \otimes x_i).
    \end{equation*}
    By the sum-reflection of \( \iota \),
    \begin{equation*}
        (\epsilon \otimes X) \circ (M^* \otimes y)
        =
        \sum_i (\epsilon \otimes X) \circ (M^* \otimes x_i).
    \end{equation*}
    Then by transposing the both side in the other direction,
    \begin{equation*}
        (M \otimes \epsilon \otimes X) \circ (\eta \otimes y)
        =
        \sum_i (M \otimes \epsilon \otimes X) \circ (\eta \otimes x_i).
    \end{equation*}
    By the weak duality,
    \begin{equation*}
        r\, y
        =
        \sum_i r\, x_i.
    \end{equation*}
    Since \( r > 0 \), by \cref{lem:real-action-invertible},
    \begin{equation*}
        y
        =
        \sum_i x_i.
    \end{equation*}
\end{proof}

\begin{lemma}\label{lem:based-strongly-flat}
    If \( M \in \basedCQ \), then \( M \) is strongly flat.
\end{lemma}
\begin{proof}
    Suppose that \( \smod \) has a basis \( (n_i, e_i, \varphi_i)_{i \in I} \) and let \( \iota \colon \smodb \hookrightarrow \smodb' \) be a monomorphism.
    Then
    \begin{equation*}
        \xymatrix@C=120pt{
            N \otimes M \ar[d]^{\iota \otimes M} \ar[r]^{N \otimes \varphi_i} & N_i \otimes \yoneda(n_i) \ar[d]^{\iota \otimes \yoneda(n_i)}
            \\
            N' \otimes M \ar[r]^{N' \otimes \varphi_i} & N'_i \otimes \yoneda(n_i)
        }
    \end{equation*}
    and
    \begin{equation*}
        \xymatrix@C=120pt{
            N \otimes M \ar[d]^{\iota \otimes M} & N_i \otimes \yoneda(n_i) \ar[l]^{N \otimes e_i} \ar[d]^{\iota \otimes \yoneda(n_i)}
            \\
            N' \otimes M & N'_i \otimes \yoneda(n_i) \ar[l]^{N' \otimes e_i}
        }
    \end{equation*}
    commute.
    The morphism \( \iota \otimes \yoneda(n_i) \) is monic by \cref{lem:weak-compact-closed-flat}, so \( \iota \otimes M \) is monic by \cref{lem:diagonal-monotone}(1).
    
    If \( \iota \) is a sum-reflecting submodule (resp.~a hereditary submodule), then \( \iota \ptensor \yoneda(n_i) \) is a sum-reflecting submodule (resp.~a hereditary submodule) by \cref{lem:weak-compact-closed-flat}, hence so is \( \iota \ptensor M \) by \cref{lem:diagonal-monotone}(2) (resp.~\cref{lem:diagonal-monotone}(3)).
\end{proof}

\begin{lemma}\label{lem:based-completion-strongly-flat}
    For \( M \in \basedCQ \), its finite completion \( \FComp{M} \) is strongly flat.
\end{lemma}
\begin{proof}
    Assume a representable basis \( (n_b, \ket{b}, \bra{b})_{b \in B} \) of \( M \).
    Then \( \sum_b \ket{b} \bra{b} = \ident_M \) by definition of basis.
    By the functoriality of the finite completion, \( \sum_b \FComp{\ket{b}} \circ \FComp{\bra{b}} = \ident_{\FComp{M}} \).
    So \( \FComp{M} \) has a \( \{ \FComp{\yoneda(n)} \mid n \in \Nat \} \)-basis \( (\FComp{\yoneda(n_b)}, \FComp{\ket{b}}, \FComp{\bra{b}})_{b \in B} \).
    Since \( \FComp{\yoneda(n)} \cong \CPM({-}, n) \) (\cref{lem:finite-completion-of-yoneda}) and \( \CPM({-},n) \) is (weakly) dualisable (\cref{lem:cpm-dualisable}), the same argument as \cref{lem:based-strongly-flat} proves the claim.
\end{proof}

\subsection{Exponential}
We fix a \( \CQ \)-module \( M \) and assume that \( M \cong \neg M_0 \) for some \( M_0 \).

Let \( M^{\odot n} \hookrightarrow M^{\ptensor n} \) be the equaliser of \( n! \) permutations \( \sigma \colon M^{\ptensor n} \longrightarrow M^{\ptensor n} \).
By the comonoid structure of \( \cofreeexp M \), we have the canonical map \( d^{(n)} \colon \cofreeexp M \longrightarrow (\cofreeexp M)^{\ptensor n} \longrightarrow M^{\ptensor n} \).
By the cocommutativity of the comonoid structure of \( \cofreeexp M \), this map factors as \( \cofreeexp M \stackrel{\hat{d}^{(n)}}{\longrightarrow} M^{\odot n} \hookrightarrow M^{\ptensor n} \).
Let \( \cofreeexp M \stackrel{\bar{d}^{(n)}}{\longrightarrow} \S^{(n)} M \hookrightarrow M^{\odot n}\) be the covering-hereditary factorisation of \( \hat{d}^{(n)} \).
\begin{equation*}
    \xymatrix{
        \cofreeexp M \ar[rrr]^{\delta^{(n)}} \ar[d]_{\bar{d}^{(n)}} \ar[dr]^{\hat{d}^{(n)}} \ar[drrr]^{d^{(n)}} & & & (\cofreeexp M)^{\ptensor n} \ar[d]^{\mathit{der}^{\ptensor n}} \\
        \S^{(n)} M \ar@{^{(}->}[r] & M^{\odot n} \ar@{^{(}->}[rr] & & M^{\ptensor n}
    }
\end{equation*}
Let \( d \defe \sum_n \Inj^{(n)} d^{(n)} \colon \cofreeexp M \longrightarrow \prod_n M^{\ptensor n} \), \( \hat{d} \defe \sum_n \Inj^{(n)} \hat{d}^{(n)} \colon \cofreeexp M \longrightarrow \prod_n M^{\odot n} \) and \( \cofreeexp M \stackrel{\bar{d}}{\longrightarrow} \S M \hookrightarrow \prod_n M^{\odot n} \) be the covering-hereditary factorisation of \( \hat{d} \).

We prove that \( \cofreeexp M \cong \S M \).

Since
\begin{equation*}
    \xymatrix{
        {\cofreeexp}M \ar[r]^{\bar{d}} \ar[dr]_{\bar{d}^{(n)}} & \S M \ar@{^{(}->}[r] & \prod_n M^{\odot n} \ar[d]^{\Proj^{(n)}} \\
        & \S^{(n)} M \ar@{^{(}->}[r] & M^{\odot n}
    }
\end{equation*}
commutes, \( \bar{d} \) is a covering map and \( \S^{(n)} M \longrightarrow M^{\ptensor n} \) is a hereditary submodule, there exists a (unique) morphism \( \S M \longrightarrow \S^{(n)} M \), which we also write as \( \Proj^{(n)} \), such that the triangle and square in
\begin{equation*}
    \xymatrix{
        {\cofreeexp}M \ar[r]^{\bar{d}} \ar[dr]_{\bar{d}^{(n)}} & \S M \ar@{^{(}->}[r] \ar[d]^{\Proj^{(n)}} & \prod_n M^{\odot n} \ar[d]^{\Proj^{(n)}} \\
        & \S^{(n)} M \ar@{^{(}->}[r] & M^{\odot n}
    }
\end{equation*}
commute.
Both \( \bar{d} \) and \( \bar{d}^{(n)} \) belong to the left class (\ie~the class of covering maps) and \( \bar{d}^{(n)} = \Proj^{(n)} \circ \bar{d} \), by the cancellation property of the left class of an orthogonal factorisation system, \( \Proj^{(n)} \) belongs to the left class (\ie~\( \Proj^{(n)} \) is covering).

\begin{lemma}
    There exists a unique \( f \colon \S^{(n)} M \longrightarrow \S M \) such that
    \begin{equation*}
        \xymatrix{
            & \S M \ar@{^{(}->}[r] & \prod_n M^{\odot n} \\
            & \S^{(n)} M \ar[u]_f \ar@{^{(}->}[r] & M^{\odot n} \ar[u]^{\Inj^{(n)}}
        }
    \end{equation*}
    commutes.
\end{lemma}
\begin{proof}
    Since \( \S M \stackrel{\iota}{\hookrightarrow} \prod_{n} M^{\odot n} \) and \( \S^{(n)}M \stackrel{\iota}{\hookrightarrow} M^{\odot n} \stackrel{\Inj^{(n)}}{\hookrightarrow} \prod_n M^{\odot n} \) are hereditary submodules, it suffices to show that \( x \in (\S^{(n)}M)_k \) implies \( x \in (\S M)_k \).
    By explicitly writing the embeddings, the goal is to prove that, for every \( x \in (\S^{(n)}M)_k \), there exists \( y \in (\S M)_k \) such that \( (\Inj^{(n)}\iota')(x) = \iota(y) \).

    Assume \( x \in (\S^{(n)}M)_k \).
    Since \( \Proj^{(n)} \) is covering, there exists \( y_0 \in (\S M)_k \) such that \( \S^{(n)} M \models x \le \Proj^{(n)}(y_0) \).
    Then \( M^{\odot n} \models \iota'(x) \le (\iota' \Proj^{(n)})(y_0) = \Proj^{(n)}\iota(y_0) \).
    Hence \( \prod_n M^{\odot n} \models \Inj^{(n)}\iota'(x) \le \Inj^{(n)}\Proj^{(n)}\iota(y_0) \le \iota(y_0) \) since \( \Inj^{(n)}\Proj^{(n)} \le \sum_n \Inj^{(n)}\Proj^{(n)} = \ident \).
    As \( \iota \) is a hereditary submodule, there exists \( y_1 \) such that \( \Inj^{(n)}\iota'(x) = \iota(y_1) \) as required.
\end{proof}
We write \( \Inj^{(n)} \) for the above \( f \).
We have \( \Proj^{(n)} \circ \Inj^{(n)} = \ident_{\S^{(n)}M} \).

Before comparing \( \S M \) with \( \cofreeexp M \), we prove that \( \S M \) has a basis.
\begin{lemma}\label{lem:weak-convexity-on-tensors}
    Assume \( p_i \in [0,1] \) and \( \sum_{i=1}^n p_i \le 1 \).
    For \( f_i \colon N \longrightarrow M^{\ptensor n} \) (\( i = 1,\dots,n \)), we have \( \CQ(N, M^{\ptensor n}) \models \IsDef{(\sum_i p_i^{n+1} f_i)} \).
\end{lemma}
\begin{proof}
    Recall that \( M_k \) (as a \( [0,1] \)-module) is convex for every \( k \) since \( M \cong \neg M_0 \).
    \tk{todo: prove this fact}

    It suffices to show that \( \IsDef{(\sum_i p_i^{n+1} f_i(x))} \) in \( M^{\ptensor n}_k \) for every \( k \) and \( x \in N_k \).
    By \cref{lem:tensor-peak}, \( f_i(x) \le (y_{i,1} \ptensor \dots \ptensor y_{i,n}) \circ \varphi_i \) for some \( y_{i,j} \in M_{k_{i,j}} \) and \( \varphi_i \in \CPM(k, \prod_j k_{i,j}) \).
    We can assume without loss of generality that \( k_{i,j} = k_{i',j} \) because there exists a retraction from \( \ell \) to \( \ell' \) in \( \CPM \) if \( \ell \le \ell' \).
    Let \( k_j \defe k_{i,j} \) and \( k' \defe \prod_j k_j \).
    Since \( M_{n_j} \) is convex, \( \sum_i p_i y_{i,j} \) is defined for every \( j \).
    So \( ((\sum_i p_1 y_{i,1}) \ptensor \dots \ptensor (\sum_i p_n y_{i,n})) \cdot (\sum_i p_i \varphi_i) \) is defined in \( M^{\ptensor n} \).
    Hence its partial sum \( \sum_i p_i^{n+1} (y_{i,1} \ptensor \dots \ptensor y_{i,n}) \cdot \varphi_i \) is also defined.
    By the assumption \( f_i(x) \le (y_{i,1} \ptensor \dots \ptensor y_{i,n}) \cdot \varphi_i \), we have \( \IsDef{(\sum_i f_i(x))} \).
\end{proof}

\begin{lemma}\label{lem:basis-for-symmetric-tensor-power}
    \( M^{\odot n} \in \basedCQ \) for every \( n \).
\end{lemma}
\begin{proof}
    Let \( g \defe \sum_{\sigma \in \mathfrak{S}_n} (1/n!)^{n+1} \sigma \colon M^{\ptensor n} \longrightarrow M^{\ptensor n} \), which is defined by \cref{lem:weak-convexity-on-tensors}.
    Since \( \sigma \circ g = g \) for every \( \sigma \in \mathfrak{S}_n \), the morphism \( g \) can be factored as \( g = \iota \circ \bar{g} \), where \( \iota \colon M^{\odot n} \longrightarrow M^{\ptensor n} \) is the equaliser.
    We have \( \iota \circ \bar{g} \circ \iota = g \circ \iota = (1/n!)^n \iota \) from \( \sigma \circ \iota = \iota \).
    Hence \( \bar{g} \circ \iota = (1/n!)^n \ident_{M^{\odot n}} \) by the universality of \( \iota \).

    Therefore \( M^{\odot n} \) has a \( \{M^{\ptensor n}\} \)-basis \( (\BObj{b}, \ket{b}, \bra{b})_{b = 1}^{(n!)^n} \) where \( \BObj{b} \defe M^{\ptensor} \), \( \ket{b} \defe \bar{g} \) and \( \bra{b} \defe \iota \).
    Since \( M^{\ptensor n} \) has a representable basis, so is \( M^{\odot n} \) by \cref{lem:basis-transitive}.
\end{proof}

\begin{lemma}\label{lem:basis-for-exponential-power-series-components}
    \( \S^{(n)} M \in \basedCQ \) for every \( n \).
\end{lemma}
\begin{proof}
    Assume a pseudo-representable basis \( (\BObj{b}, \ket{b}, \bra{b})_{b \in \Base{M}} \) for \( M \).
    Let \( J \defe \Base{M}^n \).
    Given \( \chi = (b_1,\dots,b_n) \in J \), we define a hereditary submodule \( \S^{(\chi)}M \) of \( \bigotimes_{i=1}^{n} \BObj{b_i} \).
    Let \( \#\chi = \prod_{i} \#\BObj{b_i} \).
    Note that \( \bigotimes_{i=1}^{n} \BObj{b_i} \hookrightarrow \CPM({-}, \#\chi) \).
    Then
    \begin{align*}
        (\S^{(\chi)}M)_k
        &\defe
        \{\, x \in (\bigotimes_{i=1}^{n} \BObj{b_i})_k \mid
        \exists y \in (!M)_k. (\sum_{\sigma \in \mathfrak{S}_n} \frac{1}{n!} \sigma \circ (\ket{b_1} \otimes \dots \otimes \ket{b_n}) \cdot x) \le d^{(n)}(y) \,\}.
    \end{align*}
    Then \( \ket{\chi} \) is defined as
    \begin{equation*}
        \sum_{\sigma \in \mathfrak{S}_\ell} \frac{1}{n!} \sigma \circ (\ket{b_1} \otimes \dots \otimes \ket{b_n})
    \end{equation*}
    and the coordinate function \( \bra{\chi} \) is
    \begin{equation*}
        \bra{b_1} \otimes \dots \otimes \bra{b_n}.
    \end{equation*}
    Then \( \sum_{\chi \in J} \ket{\chi} \circ \bra{\chi} \) holds.

    So it suffices to show that each \( \S^{(\chi)}M \) is has a (pseudo-representable or representable) basis.
    We prove that \( \S^{(\chi)}M \) itself is pseudo-representable.
    It has an obvious upper bound.
    We would like to show that \( r\,\ident_{\#\chi} \in \S^{(\chi)} M \) for some \( r > 0 \).
    We use \cref{lem:cpm-identity-cover}: there exists \( a_1, \dots, a_w \in \CQ(1, \#\chi) \) and \( h_1,\dots,h_w \in \CQ(\#\chi,1) \) such that \( r\,\ident \le \sum_i a_i \cdot h_i \) for some \( r > 0 \).
    We can assume without loss of generality that \( \sum_i h_i \in \CQ(\#\chi,1) \) by multiplicating \( r' > 0 \) for both \( h_i \) and \( r\,\ident \) if necessarily.
    Since each element in \( \CPM(1,\#\chi) = \CPM(1, \otimes_i \#\BObj{b_i}) \) can be covered by a finite sum of simple tensors of \( \CPM(1,\#\BObj{b_i}) \), we can assume without loss of generality that \( a_i \) is a simple tensor \( a_i = a_{i,1} \otimes \dots \otimes a_{i,n} \), \( a_{i,j} \in \CQ(1, \#\BObj{b_j}) \).

    \begin{claim*}\label{lem:simple-tensor}
        Let \( z_{i,j} \in \CQ(1,\#\BObj{b_j}) \) for each \( i=1,\dots,\ell \) and \( j = 1,\dots,n \).
        Assume that \( \sum_{i,j} r\,\ket{b_j} \cdot z_{i,j} \) is defined in \( M_1 \) for some \( r > 0 \).
        Then \( (\S^{(\chi)} M)_1 \models \IsDef{(\sum_i r'\,(z_{i,1} \otimes \dots \otimes z_{i,n}))} \) for some \( r' > 0 \).
    \end{claim*}
    \begin{proof}
Let
        \begin{equation*}
            z \quad\defe\quad \sum_{i} (r\,\ket{b_1} \cdot z_{i,1} + \dots + r\,\ket{b_n} \cdot z_{i,\ell}),
        \end{equation*}
        which is defined in \( M_1 \) by the assumption.
        This element \( z \in M_1 \) can be identified with a morphism 
        \begin{equation*}
            z \colon \yoneda(1) \longrightarrow M
        \end{equation*}
        by the Yoneda lemma.
        Since \( \yoneda(1) \), as the unit of the tensor product, has the canonical comonoid structure, there exists a canonical comonoid map
        \begin{equation*}
            z^! \colon \yoneda(1) \longrightarrow {!}M
        \end{equation*}
        by the cofreeness of \( !M \).
        Then
        \begin{align*}
            d^{(n)}_{!M} \circ z^{!}
            &=
            z^{\ptensor n} \circ d^{(n)}_{\yoneda(1)}
            =
            z \ptensor \dots \ptensor z.
        \end{align*} 
        Since
        \begin{equation*}
            \sum_i r^n (\ket{b_1} \cdot z_{i,1}) \ptensor \dots \ptensor (\ket{b_n} \cdot z_{i,n})
            \quad\le\quad
            z \ptensor \dots \ptensor z
            \quad=\quad
            d^{(\ell)}(z^!)
        \end{equation*}
        holds in \( M^{\ptensor n} \), we have
        \begin{align*}
            &
            \sum_{\sigma \in \mathfrak{S}_n} \frac{1}{n!} \sigma \circ (\ket{b_1} \ptensor \dots \ptensor \ket{b_n}) \cdot (\sum_i r^n z_{i,1} \ptensor \dots \ptensor z_{i,n}))
            \\
            &\le \sum_{\sigma \in \mathfrak{S}_n} \frac{1}{n!} \sigma \circ (z \ptensor \dots \ptensor z)
            \\
            &= z \ptensor \dots \ptensor z
            \\
            &= d^{(n)}(z^!).
        \end{align*}
        So \( (\sum_i r^n z_{i,1} \otimes \dots \otimes z_{i,n})) \in (\S^{(\chi)}M)_1 \).
    \end{proof}

    By the above claim, \( \sum_i r' a_i = \sum_i r' (a_{i,1} \otimes \dots \otimes a_{i,\ell}) \) is in \( (\S^{(\chi)}M)_1 \) for some \( r' > 0 \).
    Since \( (\sum_i h_i) \in \CQ(\#\chi, 1) \), wa have \( (\sum_i r' a_i) \cdot (\sum_{i'} h_{i'}) \) is in \( (\S^{(\chi)}M_{\#\chi}) \).
    So, by the downward closedness of \( (\S^{(\chi)}M_{\#\chi}) \), we have \( (\sum_i r' a_i \cdot h_{i}) \in (\S^{(\chi)}M_{\#\chi}) \).
    Since \( rr' \ident \le (\sum_i r' a_i \cdot h_{i}) \), we have \( rr' \ident \in  (\S^{(\chi)}M_{\#\chi}) \).
\end{proof}

Now we analyse the comultiplication \( \delta \colon \cofreeexp M \longrightarrow \cofreeexp M \ptensor \cofreeexp M \).
As we have seen, the cofree comonoid \( \cofreeexp M \) induces a family of morphisms \( \delta^{(n)} \colon \cofreeexp M \longrightarrow (\cofreeexp M)^{\ptensor n} \) for each \( n \), where \( \delta^{(0)} \) is the counit.
Compositing \( \delta^{(n)} \) with \( \Der^{\otimes n} \colon (\cofreeexp M)^{\ptensor n} \longrightarrow M^{\ptensor n} \) and \( M^{\ptensor n} \longrightarrow \FComp{M}^{\ptensor n} \), we have \( \tilde{d}^{(n)} \colon \cofreeexp M \longrightarrow \FComp{M}^{\ptensor n} \), as well as \( d^{(n)} \), \( \hat{d}^{(n)} \) and \( \bar{d}^{(n)} \), for each \( n > 0 \).
\begin{equation*}
    \xymatrix{
        & & \cofreeexp M \ar[dll]_{\bar{d}^{(n)}} \ar[dl]^{\hat{d}^{(n)}} \ar[dr]_{d^{(n)}} \ar[drrr]^{\tilde{d}}\\
        \S^{(n)} M \ar@{^{(}->}[r] & M^{\odot n} \ar@{^{(}->}[rr] & & M^{\ptensor n} \ar@{^{(}->}[rr] & & \FComp{M}^{\ptensor n}
    }
\end{equation*} 
Then
\begin{equation*}
    \xymatrix{
        \cofreeexp M \ar[d]_\delta \ar[rr]^{\tilde{d}^{(n+m)}} & & \FComp{M}^{\ptensor (n+m)} \ar[d]^{\cong} \\
        \cofreeexp M \otimes \cofreeexp M \ar[rr]^{\tilde{d}^{(n)} \otimes \tilde{d}^{(m)}} & & \FComp{M}^{\ptensor n} \otimes \FComp{M}^{\ptensor m}
    }
\end{equation*}
commutes.
Let \( \tilde{d} \colon \cofreeexp M \longrightarrow \prod_{n} \FComp{M}^{\ptensor n} \) be the map defined by \( \tilde{d} = \sum_n \Inj^{(n)} \tilde{d}^{(n)} \).

Let \( \SymTensor{\FComp{M}}{n} \) be the equaliser of the \( n! \) permutations on \( \FComp{M}^{\otimes n} \), which exists since \( \pshCQ \) is complete.
By the finite completeness of \( \FComp{M} \), the equaliser \( \SymTensor{\FComp{M}}{n} \hookrightarrow \FComp{M}^{\otimes n} \) splits.
Let
\begin{equation*}
    \EqSplit{n}_0 \defe \sum_{\sigma \in \mathfrak{S}_n} \frac{1}{n!} \sigma
    \qquad\colon\quad \FComp{M}^{\ptensor n} \longrightarrow \FComp{M}^{\ptensor n}.
\end{equation*}
Since \( \sigma \circ \EqSplit{n} = \EqSplit{n} \) for each permutation \( \sigma \in \mathfrak{S}_n \), it factors though \( \SymTensor{\FComp{M}}{n} \hookrightarrow \FComp{M}^{\ptensor n} \):
\begin{equation*}
    \EqSplit{n}_0
    \quad=\quad
    \FComp{M}^{\ptensor n}
    \stackrel{\EqSplit{n}}{\longrightarrow}
    \SymTensor{\FComp{M}}{n}
    \hookrightarrow
    \FComp{M}^{\ptensor n}.
\end{equation*}
Then
\( \SymTensor{\FComp{M}}{n} \hookrightarrow \FComp{M}^{\otimes n} \stackrel{\EqSplit{n}}{\longrightarrow} \SymTensor{\FComp{M}}{n} \) is the identity.

By the cocommutativity, the image of \( \tilde{d}^{(n)} \) is symmetric.
So it factors as \( \cofreeexp M \stackrel{\check{d}^{(n)}}{\longrightarrow} \SymTensor{\FComp{M}}{(n)} \hookrightarrow \FComp{M}^{\ptensor n} \).
The current situation is
\begin{equation*}
    \xymatrix{
        C \ar[d]_\delta \ar[rr]^{\check{d}^{(n+m)}} & & \SymTensor{\FComp{M}}{(n+m)}  \ar@{^{(}-_>}[r] & \FComp{M}^{\ptensor (n+m)} \ar[d]^{\cong} \\
        C \otimes C \ar[rr]^{\check{d}^{(n)} \ptensor \check{d}^{(m)}} & & \SymTensor{\FComp{M}}{n} \otimes \SymTensor{\FComp{M}}{m} \ar@{^{(}-_>}[r] & \FComp{M}^{\ptensor n} \ptensor \FComp{M}^{\ptensor m}
    }.
\end{equation*}
Here \( \SymTensor{\FComp{M}}{n} \otimes \SymTensor{\FComp{M}}{m} \hookrightarrow \FComp{M}^{\ptensor n} \ptensor \FComp{M}^{\ptensor m} \) is monic since \( \SymTensor{\FComp{M}}{n} \) and \( \FComp{M}^{\ptensor m} \) are strongly flat.
The strong flatness of \( \FComp{M}^{\ptensor m} \) follows from the facts that \( \FComp{M}^{\ptensor m} \cong \FComp{M^{\ptensor m}} \) (\cref{lem:tensor-preserves-finite-completeness}) and that \( M^{\ptensor m} \in \basedCQ \) implies the strong flatness of \( \FComp{M^{\ptensor m}} \) (\cref{lem:based-completion-strongly-flat}).
The strong flatness of \( \SymTensor{\FComp{M}}{n} \) comes from a \( \{ \FComp{M}^{\ptensor n} \} \)-basis (given by the splitting of the equaliser).
We have a map \( \EqSplit{n} \ptensor \EqSplit{m} \) from \( \FComp{M}^{\ptensor (n+m)} \) to \( \SymTensor{\FComp{M}}{n} \ptensor \SymTensor{\FComp{M}}{m} \), which is the identity on \( \SymTensor{\FComp{M}}{(n+m)} \).
It determines the unique morphism \( \bar{\delta}^{(n,m)} \colon \SymTensor{\FComp{M}}{(n+m)} \longrightarrow \SymTensor{\FComp{M}}{n} \ptensor \SymTensor{\FComp{M}}{m} \) since \( \SymTensor{\FComp{M}}{n} \ptensor \SymTensor{\FComp{M}}{m} \hookrightarrow \FComp{M}^{\ptensor n} \ptensor \FComp{M}^{\ptensor m} \) is monic.
\begin{equation*}
    \xymatrix{
        \cofreeexp M \ar[d]_\delta \ar[rr]^{\check{d}^{(n+m)}} & & \SymTensor{\FComp{M}}{(n+m)}  \ar@{^{(}-_>}[r] \ar[d]^{\check{\delta}^{(n,m)}} & \FComp{M}^{\ptensor (n+m)} \ar[d]^{\cong} \\
        \cofreeexp M \ptensor \cofreeexp M \ar[rr]^{\check{d}^{(n)} \ptensor \check{d}^{(m)}} & & \SymTensor{\FComp{M}}{n} \ptensor \SymTensor{\FComp{M}}{m} \ar@{^{(}-_>}[r] & \FComp{M}^{\ptensor n} \ptensor \FComp{M}^{\ptensor m}
    }
\end{equation*}
The outer square and the right square are known to commute.
Again, because \( \SymTensor{\FComp{M}}{n} \otimes \SymTensor{\FComp{M}}{m} \hookrightarrow \FComp{M}^{\ptensor n} \ptensor \FComp{M}^{\ptensor m} \) is monic, the left square also commutes.

Since \( \EqSplit{n} \otimes \EqSplit{m} \) acts as the identity on the image of \( \SymTensor{\FComp{M}}{(n+m)} \hookrightarrow \FComp{M}^{\ptensor (n+m)} \), \( \check{\delta}^{(n,m)} \) is an injection on each component and hence monic.

Now we have
\begin{equation*}
    \xymatrix{
        \cofreeexp M \ar[d]_\delta \ar[rr]^{\bar{d}^{(n+m)}} & & \S^{(n+m)} M  \ar@{^{(}-_>}[r] & \SymTensor{\FComp{M}}{(n+m)} \ar@{^{(}-_>}[d]^{\bar{\delta}^{(n,m)}} \\
        \cofreeexp M \ptensor \cofreeexp M \ar[rr]^{\bar{d}^{(n)} \ptensor \bar{d}^{(m)}} & & \S^{(n)}M \ptensor \S^{(m)}M \ar@{^{(}-_>}[r] & \SymTensor{\FComp{M}}{n} \ptensor \SymTensor{\FComp{M}}{m}
    }
\end{equation*}
Note that \( \S^{(n)}M \ptensor \S^{(m)}M \hookrightarrow \SymTensor{\FComp{M}}{n} \ptensor \SymTensor{\FComp{M}}{m} \) is a hereditary submodule since \( \S^{(n)}M \hookrightarrow \SymTensor{M}{n} \) is hereditary by definition, \( \SymTensor{M}{n} \hookrightarrow \SymTensor{\FComp{M}}{n} \) is hereditary as it is a finite completion\tk{todo: prove.  One can prove this fact by using a weak splitting \( M^{\ptensor n} \longrightarrow \SymTensor{M}{n} \).} and four modules are strongly flat (\cref{lem:based-strongly-flat,lem:based-completion-strongly-flat}).
Since the rectangle commutes, \( \S^{(n)}M \ptensor \S^{(m)}M \hookrightarrow \SymTensor{\FComp{M}}{n} \ptensor \SymTensor{\FComp{M}}{m} \) is a hereditary submodule and \( \bar{d}^{(n+m)} \) is covering, there is a diagonal fill-in:\tk{todo: prove this claim}
\begin{equation*}
    \xymatrix{
        \cofreeexp M \ar[d]_\delta \ar[rr]^{\bar{d}^{(n+m)}} & & \S^{(n+m)}M  \ar@{^{(}-_>}[r] \ar[d]^{\delta^{(n,m)}} & \SymTensor{\FComp{M}}{(n+m)} \ar@{^{(}-_>}[d]^{\bar{\delta}^{(n,m)}} \\
        \cofreeexp M \ptensor \cofreeexp M \ar[rr]^{\bar{d}^{(n)} \otimes \bar{d}^{(m)}} & & \S^{(n)} M \ptensor \S^{(m)}M \ar@{^{(}-_>}[r] & \SymTensor{\FComp{M}}{n} \ptensor \SymTensor{\FComp{M}}{m}
    }
\end{equation*}
By this way, we have
\begin{equation*}
    \delta^{(n,m)}
    \quad\colon\quad
    \S^{(n+m)} M \longrightarrow \S^{(n)} M \ptensor \S^{(m)} M.
\end{equation*}
By construction, three rectangles in
\begin{equation*}
    \xymatrix{
        \cofreeexp M \ar[d]_\delta \ar[rr]^{\bar{d}^{(n+m)}} & & \S^{(n+m)}M  \ar@{^{(}-_>}[r] \ar[d]^{\delta^{(n,m)}} & M^{\ptensor (n+m)} \ar[d]^{\cong} \\
        \cofreeexp M \otimes \cofreeexp M \ar[rr]^{\bar{d}^{(n)} \otimes \bar{d}^{(m)}} & & \S^{(n)} M \ptensor \S^{(m)}M \ar@{^{(}-_>}[r] & M^{\ptensor n} \ptensor M^{\ptensor m}
    }
\end{equation*}
commute.

Let \( \delta \defe \sum_{n,m} (\Inj^{(n)} \ptensor \Inj^{(m)}) \delta^{(n,m)} \Proj^{(n)} \colon \S M \longrightarrow \S M \ptensor \S M \).
We prove that this sum is indeed defined.
Note that \( \S M \ptensor \S M \hookrightarrow (\prod_n \S^{(n)} M) \ptensor (\prod_m \S^{(m)} M) \hookrightarrow \prod_{n,m} (\S^{(n)}M \ptensor \S^{(m)}M) \) is a hereditary submodule, which we write as \( \iota \).
Then \( \delta' \defe \sum_{n,m} \iota \circ (\Inj^{(n)} \ptensor \Inj^{(m)}) \delta^{(n,m)} \Proj^{(n)} \) is well-defined.
It suffices to prove that the image of \( \delta' \) is in the image of \( \iota \).
Let \( x \in (\S M)_k \).
Then \( y \in (\cofreeexp M)_k \) and \( x \le \bar{d}(y) \) for some \( y \).
Consider \( \delta(y) \in (\cofreeexp M \otimes \cofreeexp M)_k \).
By \cref{lem:tensor-peak}, \( \delta(y) \le (y_1 \otimes y_2) \cdot \varphi \) for some \( y_1 \in (\cofreeexp M)_{k_1} \), \( y_2 \in (\cofreeexp M)_{k_2} \) and \( \varphi \in \CQ(k, k_1 \otimes k_2) \).
Then
\begin{align*}
    \delta'(x)
    &\le
    \delta'(\bar{d}(y))
    \\
    &\le
    \iota (\bar{d} \ptensor \bar{d})(\delta(y))
    \\
    &\le
    \iota (\bar{d} \ptensor \bar{d})(y_1 \ptensor y_2) \cdot \varphi
    \\
    &=
    \iota (\bar{d}(y_1) \ptensor \bar{d}(y_2)) \cdot \varphi.
\end{align*}
The right-most-term belongs to \( (\S M \ptensor \S M)_k \) and \( \iota \colon \S M \ptensor \S M \hookrightarrow (\prod_{n,m} \SymTensor{M}{n} \ptensor \SymTensor{M}{m} \) is hereditary, \( \delta'(x) \) belongs to \( \S M \ptensor \S M \) as expected.

Then it is not difficult to see that \( (\S M, \delta, \Proj^{(0)}) \) is a cocommutative comonoid.
It is a comonoid over \( M \) via \( \Proj^{(1)} \).
By the universality of \( \cofreeexp M \), there exists a comonoid morphism \( \alpha \colon \S M \longrightarrow \cofreeexp M \) such that \( \Der \circ \alpha = \Proj^{(1)} \).
Since any comonoid-over-\( M \) morphism preserves \( d \), and \( d_{\S M} = \ident_{\S M} \) by definition,
\begin{equation*}
    \S M \stackrel{\alpha}{\longrightarrow} !M \stackrel{d}{\longrightarrow} \S M
\end{equation*}
is the identity.
Since
\begin{equation*}
    !M \stackrel{d}{\longrightarrow} \S M \stackrel{d}{\longrightarrow} !M
\end{equation*}
preserves \( \Der \), by the universality of \( !M \), the composite is also the identity.
Hence canonically \( \S M \cong !M \), and the canonical morphism preserves the comonoid-over-\( M \) structure.

\begin{proof}[Proof of \cref{thm:exponential-based-setting-formal-power-series}]
    It is easy to see the claims by using the characterisation \( \cofreeexp M \cong \S M \hookrightarrow \prod_n M^{\ptensor n} \).
\end{proof}

 \section{Classical Structures}
Let \( \classicalCQ \hookrightarrow \basedCQ \) be the full subcategory of \( \basedCQ \) consisting of \( M \in \basedCQ \) such that the canonical morphism \( M \longrightarrow \neg\neg M \) is an isomorphism.
This section studies its structures.

\subsection{Classical Model as the Eilenberg-Moore Category of the Continuation Monad}
Another characterisation of \( \classicalCQ \) is as the Eilenberg-Moore category of the continuation monad \( \neg\neg \) on \( \basedCQ \).
This characterisation is useful because it involes the canonical adjunction between \( \basedCQ \) and \( \classicalCQ \).
This definition makes sense because \( M \in \basedCQ \) implies \( \neg\neg M \in \basedCQ \) (\cref{lem:based-additive-multiplicative}).

The continuation monad \( \neg\neg \) on \( \basedCQ \) is idempotent.
This fact significantly eases the analysis of \( \classicalCQ \).
To prove the idempotency of \( \neg\neg \) on \( \basedCQ \), we first analyse the action of \( \neg \) to pseudo-representable modules.

The negation on pseudo-representable modules can be characterised by using the \emph{orthogonality relation} as in many models of linear logic.
We write \( \mathcal{F} \subseteq \CPM({-}, \ell) \) to mean that \( \mathcal{F} = (\mathcal{F}_n)_n \) is a family of subsets \( \mathcal{F}_n \subseteq \CPM(n,\ell) \); note that \( \mathcal{F} \) does not need to be closed under the \( \CQ \)-action.
Families \( \mathcal{F} \subseteq \CPM({-}, \ell) \) and \( \mathcal{F}' \subseteq \CPM({-}, \neg \ell) \) are \emph{orthogonal}, written \( \mathcal{F} \mathrel{\bot} \mathcal{F}' \), if
\begin{align*}
    & \forall n,m.\ \forall x \in \mathcal{F}_n.\ \forall f \in \mathcal{F}'_m.\ \forall \varphi \in \CQ(1,n \otimes m).
    \quad \mathbf{ev} \circ (f \otimes x) \circ \varphi \le 1.
\end{align*}
For \( \mathcal{F} \subseteq \CPM({-}, \ell) \), let \( \mathcal{F}^{\bot} \subseteq \CPM({-}, \neg\ell) \) be
the family given by
\begin{align*}
    &
    \mathcal{F}^{\bot}_n \defe \{\, f \in \CPM(n,\neg\ell) \mid
    \forall m. \forall x \in \mathcal{F}_m. \forall \varphi \in \CQ(1, n \otimes m).\ (\mathbf{ev} \circ (f \otimes x) \circ \varphi) \le 1 \,\}.
\end{align*}
\begin{lemma}\label{lem:negation-of-pseudo-representable}
    Let \( \mathcal{F} \subseteq \CPM({-},\ell) \) and \( \mathcal{F}' \subseteq \CPM({-}, \neg\ell) \).
    \begin{enumerate}
        \item \( \mathcal{F} \mathrel{\bot} \mathcal{F}^{\bot} \).
        \item \( \mathcal{F}' \subseteq \mathcal{F}^{\bot}
        \Longleftrightarrow
        \mathcal{F} \mathrel{\bot} \mathcal{F}'
        \Longleftrightarrow
        \mathcal{F} \subseteq (\mathcal{F}')^{\bot} \).
        \item \( \mathcal{F} \subseteq \mathcal{F}^{\bot\bot} \) and \( \mathcal{F}^{\bot\bot\bot} = \mathcal{F}^{\bot} \).
\item \( \LQT^{\bot} = \neg \LQT \)\/ for a pseudo-representable \( \CQ \)-module \( \LQT \).
\qed
    \end{enumerate}
\end{lemma}
\begin{proof}
    (1)~Trivial.
    (2)~The first equivalence is trivial.  The second equivalence follows from the fact that the canonical morphism \( \ell \longrightarrow \neg\neg\ell \) in \( \CPM \) is the identity.
    (3)~\( \mathcal{F} \subseteq \mathcal{F}^{\bot\bot} \) follows from (1) and (2).
    In particular, \( \mathcal{F}^{\bot} \subseteq \mathcal{F}^{\bot\bot\bot} \).
    By the anti-monotonicity of \( ({-})^{\bot} \) (\ie~\( \mathcal{F}_1 \subseteq \mathcal{F}_2\) implies \( \mathcal{F}_1^\bot \supseteq \mathcal{F}_2^\bot \)) and \( \mathcal{F} \subseteq \mathcal{F}^{\bot\bot} \), we also have \( \mathcal{F}^{\bot\bot\bot} \subseteq \mathcal{F}^{\bot} \).
(4)~Since \( \LQT \rightarrowtriangle \yoneda(1) \) equals \( \LQT^{\bot} \).
\end{proof}

\begin{lemma}\label{lem:based-idempotency}
  The monad \( \neg\neg \) is idempotent on \( \basedCQ \), \ie~the multiplication \( \mu_{M} \colon \neg\neg\neg\neg M \longrightarrow \neg\neg M \) of the monad \( \neg\neg \) is an isomorphism for each \( M \in \basedCQ \).
\end{lemma}
\begin{proof}
    By \cref{lem:negation-of-pseudo-representable}, \( \mu_{\yoneda(n)} \colon \neg\neg\neg\neg \yoneda(n) \longrightarrow \neg\neg \yoneda(n) \) is an isomorphism.

    Let us consider the general case.
    By the general theory of idempotent monad, it suffices to show that \( \mu \) is a monomorphism.

    Assume that \( M \in \basedCQ \).
    Let \( (n_i, e_i, \varphi_i)_i \) be a basis.
    By the naturality of \( \mu \),
    \begin{equation*}
        \xymatrix{
            \neg\neg\neg\neg M \ar[d]^\mu \ar[r]^{\neg\neg\neg\neg\varphi_i} & \neg\neg\neg\neg \yoneda(n_i) \ar[d]^{\mu}
            \\
            \neg\neg M \ar[r]^{\neg\neg \varphi_i} & \neg\neg \yoneda(n_i)
        }
    \end{equation*}
    and
    \begin{equation*}
        \xymatrix{
            \neg\neg\neg\neg M \ar[d]^\mu & \neg\neg\neg\neg \yoneda(n_i) \ar[d]^{\mu} \ar[l]^{\neg\neg\neg\neg e_i}
            \\
            \neg\neg M & \neg\neg \yoneda(n_i) \ar[l]^{\neg\neg e_i}
        }
    \end{equation*}
    commute for every \( i \).
    By the \( \SMon \)-enrichment,
    \begin{align*}
        \ident_{\neg\neg\neg\neg M}
        &=
        \neg\neg\neg\neg \ident_M
        \\
        &=
        \sum_i \neg\neg\neg\neg e_i \circ \neg\neg\neg\neg \varphi_i
    \end{align*}
    and
    \begin{align*}
        \ident_{\neg\neg M}
        &=
        \neg\neg \ident_M
        \\
        &=
        \sum_i \neg\neg e_i \circ \neg\neg \varphi_i.
    \end{align*}
    So the claim follows from \cref{lem:diagonal-monotone}.
\end{proof}

By the general theory of idempotent monad, the Eilenberg-Moore category \( \classicalCQ \) is the full subcategory of \( \basedCQ \) consisting of \( \CQ \)-modules \( M \) such that the canonical morphism \( M \longrightarrow \neg\neg M \) is an isomorphism.
In particular, \( \neg\neg M \) is isomorphic to \( M \) for every \( M \in \classicalCQ \).

\subsection{Multiplicatives and Additives}
\( \classicalCQ \) has additives and multiplicatives, some of which slightly differ from those in \( \pshCQ \).

\begin{lemma}\label{lem:additive-on-classical}
  Let \( M, N \in \classicalCQ \).
  Their product and coproduct in \( \classicalCQ \) are \( M \times N \) and \( \neg\neg(M \amalg N) \).
\end{lemma}
\begin{proof}
  An easy calculation shows that \( \neg\neg(M \times N) \) and \( \neg\neg(M \amalg N) \) are the product and coproduct in \( \classicalCQ \), respectively.
  Since \( \classicalCQ \hookrightarrow \basedCQ \) as a right adjoint preserves products, \( \neg\neg(M \times N) \cong M \times N \). 
\end{proof}

The tensor product is defined via the bilinear maps.
For \( M,N \in \classicalCQ \), consider the \( \SMon \)-enriched functor \( \Bilin(M,N; {-}) \colon \classicalCQ \longrightarrow \SMon \), the restriction of the previously considered functor to \( \classicalCQ \).
If its representing object exists for each \( M,N \in \classicalCQ \), then \( \classicalCQ \) has an induced symmetric monoidal structure \( \otimes \) such that \( M \otimes N \) is the representing object of \( \Bilin(M,N;{-}) \).
\begin{lemma}\label{lem:monoidal-structure-on-classical}
  Let \( M, N \in \classicalCQ \).
  Then \( \neg\neg(M \ptensor N) \) is the representing object in \( \classicalCQ \) of the bilinear maps.
  The right adjoint of \( \neg\neg(({-}) \ptensor M) \) is \( M \multimap ({-}) \).
\end{lemma}
\begin{proof}
  The former comes from \( \classicalCQ(\neg\neg(M \ptensor N), L) \cong \pshCQ(M \ptensor N, L) \cong \Bilin(M,N; L) \).
To prove the latter, we first check that \( M \multimap N \in \classicalCQ \) for \( M,N \in \classicalCQ \).
  Since \( M \cong \neg\neg M \), we have \( (N \multimap M) \cong \neg (N \ptensor \neg M) \).
  By \cref{lem:based-additive-multiplicative}, \( N \ptensor \neg M \in \basedCQ \).
  Then \( \neg (N \ptensor \neg M) \in \classicalCQ \) because \( L \in \basedCQ \) implies \( \neg L \in \classicalCQ \) for every \( L \) (see \cref{lem:single-negation-and-biorthogonal} below).
Then we have \( \classicalCQ(\neg\neg(L \ptensor M), N) \cong \pshCQ(L \ptensor M, N) \cong \pshCQ(L, M \multimap N) \cong \classicalCQ(L, M \multimap N) \).
\end{proof}

\begin{lemma}\label{lem:single-negation-and-biorthogonal}
  For \( M \in \basedCQ \), we have \( \neg M \in \classicalCQ \).
\end{lemma}
\begin{proof}
  Recall that \( \neg\neg \) is the monad induced from the self-adjunction \( \neg : \basedCQ \leftrightarrows \basedCQ^{\op} : \neg \).
  Since \( \classicalCQ \) is the Eilenberg-Moore category, we have a comparison \( K \colon \basedCQ^{\op} \longrightarrow \classicalCQ \) and the right adjoint \( \neg \colon \basedCQ^{\op} \longrightarrow \basedCQ \) factors as \( \basedCQ^{\op} \stackrel{K}{\longrightarrow} \classicalCQ \hookrightarrow \basedCQ \).
\end{proof}

\begin{lemma}\label{lem:double-negation-strong-monoidal}
  \( \neg\neg: \basedCQ \longrightarrow \classicalCQ \) is strong monoidal.   
\end{lemma}
\begin{proof}
  By the defining universal properties of the tensor products in these categories.
\end{proof}

\begin{theorem}
  There exists a fully faithful strong-monoidal (\(\SMon\)-enriched) functor \( \CQ \hookrightarrow \classicalCQ \).
\end{theorem}
\begin{proof}
  The Yoneda embedding \( \yoneda \colon \CQ \longrightarrow \pshCQ \) is fully faithful.
  It is also strong monoidal since the monoidal structure of \( \pshCQ \) is the Day tensor.
  As \( \yoneda(n) \in \basedCQ \), its codomain can be restricted to \( \basedCQ \).
  Then \( \CQ \stackrel{\yoneda}{\hookrightarrow} \basedCQ \stackrel{\neg\neg}{\longrightarrow} \classicalCQ \) is strong monoidal.
  It is fully faithful since \( \classicalCQ \) is a reflective subcategory of \( \basedCQ \) and \( \yoneda(n) \in \classicalCQ \).
\end{proof}

\subsection{Exponential}\label{sec:exponential}
This subsection proves that \( \neg\neg\cofreeexp ({-}) \) is a linear exponential comonad on \( \classicalCQ \), giving a monoidal adjunction
\begin{equation*}
  \xymatrix{
    \Comon(\basedCQ) \ar[r] & \basedCQ \ar[r]^-{\neg\neg} & \classicalCQ \ar@/^18pt/[ll]_-{\cofreeexp}
  },
\end{equation*}
where \( \Comon(\basedCQ) \) is the category of comonoids in \( \basedCQ \) and \( \Comon(\basedCQ) \longrightarrow \basedCQ \) is the forgetful functor.

\begin{theorem}\label{thm:classical-exponential}
  \( \neg\neg\cofreeexp \) is a linear exponential comonad in \( \classicalCQ \).
\end{theorem}
\begin{proof}
  We have the following categories and functors:
  \begin{equation*}
    \xymatrix{
      \Comon(\pshCQ) \ar@<0.8ex>[r] \ar@{}[r]|-{\bot}& \pshCQ \ar@<0.8ex>[r]^-{\neg\neg} \ar@<0.8ex>[l]^-{\cofreeexp} \ar@{}[r]|-{\bot} & \pshCQ{}^{\neg\neg} \ar@<0.8ex>[l]
      \\
      \Comon(\basedCQ) \ar@{^{(}-_>}[u] \ar[r] & \basedCQ \ar@{^{(}-_>}[u] \ar@<0.8ex>[r]^-{\neg\neg} \ar@{}[r]|-{\bot}& \classicalCQ \ar@{^{(}-_>}[u] \ar@/^18pt/@<0.8ex>[ll]_-{\cofreeexp} \ar@{^{(}-_>}@<0.8ex>[l]
    }.
  \end{equation*}
  In this diagram, the functors in the lower part is the restriction of the upper counterparts.
  \cref{thm:exponential-based-setting-formal-power-series} ensures that \( \classicalCQ \hookrightarrow \basedCQ \hookrightarrow \pshCQ \stackrel{\cofreeexp}{\longrightarrow} \Comon(\pshCQ) \) factors through \( \Comon(\basedCQ) \hookrightarrow \Comon(\pshCQ) \).
As the restriction of an adjunction, \( {!} \colon \classicalCQ \to \Comon(\basedCQ) \) is the right adjoint of \( \Comon(\basedCQ) \to \basedCQ \stackrel{\neg\neg}{\to} \classicalCQ \).
  Its left adjoint is strong monoidal as the composite of strong monoidal functors (\cf~\cref{lem:double-negation-strong-monoidal} for the second one).
\end{proof}

\subsection{Proof of \cref{thm:classical-structure}}
\begin{claim*}[of \cref{thm:classical-structure}]
  \( \classicalCQ \) is a model of classical linear logic with the tensor product \( (M \otimes N) \defe \neg \neg (M \ptensor N) \)and the linear exponential comonad \( {!}M \defe \neg\neg\cofreeexp M \).
  Its product is \( M \times N \) and coproduct is \( (M + N) \defe \neg\neg(M \amalg N) \).
  The linear function space is \( M \multimap N \). 
\end{claim*}

A consequence of \cref{lem:additive-on-classical,lem:monoidal-structure-on-classical,thm:classical-exponential}.

 \section{Supplementary Materials for Section~5}
\tk{\cref{sec:recursive}}

\subsection{Splitting Object of an Idempotent in $\classicalCQ$}
Assume \( N \in \classicalCQ \) and \( \iota' \colon N \longrightarrow N \) is an idempotent.  As \( \pshCQ \) is locally presentable, it is complete and thus has a splitting object \( M' \in \pshCQ \) of the idempotent \( \iota \).  Let \( e' \colon M' \longrightarrow N \) and \( p' \colon N \longrightarrow M' \) be the retraction such that \( e' \circ p' = \iota' \).  Then \( \neg\neg M' \stackrel{\neg\neg e'}{\longrightarrow} \neg\neg N \cong N \) and \( N \stackrel{p'}{\longrightarrow} M' \longrightarrow \neg\neg M' \) also gives a splitting object.  So we can assume without loss of generality that \( M' = \neg M'_0 \) for some \( M'_0 \).  Let \( (\LQT_b, \ket{b}, \bra{b})_b \) be a basis of \( N \).  Then \( (\LQT_i, p' \circ \ket{b}, \bra{b} \circ e')_b \) is a basis of \( M' \).  Hence \( M' \in \classicalCQ \).  This basis for \( M' \) is orthogonal if that for \( N \) is.

\subsection{Proof of \cref{lem:classical-exponential-cpo-enriched}}
\begin{claim*}[of \cref{lem:classical-exponential-cpo-enriched}]
    \( {!} \colon \classicalCQ \longrightarrow \classicalCQ \) is an \( \omega \)CPO-enriched functor.
\end{claim*}

Assume an \( \omega \)-chain \( f_0 \le f_1 \le \dots \) in \( \classicalCQ(M,N) \).
We can assume without loss of generality that \( f_0 = 0 \).
Let \( f \defe \bigvee_i f_i \).

By definition, there exists \( g_i \) such that \( f_i + g_i = f_{i+1} \) (\( i \in \Nat \)).
Then \( f_i = g_0 + \dots + g_{i-1} \).
For every finite \( I \subseteq_{\mathrm{fin}} \Nat \), the sum \( \sum_{i \in I} g_i \) converges since \( \sum_{i \in I} g_i \le f_{\max(I)+1} \) (where \( \max(\emptyset) = 0 \)).
By \( \omega \)-completeness of \( \classicalCQ(M,N) \) (\cref{lem:omega-cpo-enrichment}), the infinite sum \( \sum_{i \in \Nat} g_i \) is defined.
Furthermore, by \cref{lem:omega-cpo-enrichment}, \( \sum_{i \in \Nat} g_i = \bigvee_{I \subseteq_{\mathrm{fin}} \Nat} \sum_{i \in I} g_i \).
Since \( \sum_{i \in I} g_i \le f_{\max(I)+1} \) for every \( I \subseteq_{\mathrm{fin}} \Nat \), we have \( \bigvee_{I \subseteq_{\mathrm{fin}} \Nat} \sum_{i \in I} g_i \le \bigvee_i f_i \).
Since \( \sum_{i \in I_n} g_i = f_n \) where \( I_n = \{ 0,1,\dots,n-1 \} \), we have \( \bigvee_{I \subseteq_{\mathrm{fin}} \Nat} \sum_{i \in I} g_i \ge \bigvee_i f_i \).
Therefore \( \sum_{i \in \Nat} g_i = \bigvee_i f_i \).

Recall that, given \( h \in \classicalCQ(M,N) \), the morphism \( {!}h \in \classicalCQ(M,N) \) has a (non-canonical) matrix representation \( ({!}h_{\neg\neg\vec{a},\neg\neg\vec{b}})_{\neg\neg\vec{a} \in \Base{!M}, \neg\neg\vec{b} \in \Base{!N}} \) where
\begin{align*}
    !h_{\neg\neg\vec{a},\neg\neg\vec{b}}
    \quad\defe\quad
    \begin{cases}
        \neg\neg(\bra{b_1}h\ket{a_1} \ptensor \dots \ptensor \bra{b_k}h\ket{a_k}) & \mbox{if \( |\vec{a}| = |\vec{b}| = k \)} \\
        0 & \mbox{if \( |\vec{a}| \neq |\vec{b}| \)}
    \end{cases}
\end{align*}
where \( |\vec{a}| \) is the length of the sequence \( \vec{a} \).
In this proof, we always associate this matrix for a morphism of the form \( !h \).
Note that, given matrices \( h^{(i)} = (h^{(i)}_{\neg\neg\vec{a},\neg\neg\vec{b}})_{\neg\neg\vec{a},\neg\neg\vec{b}} \) representing morphisms in \( \classicalCQ(!M,!N) \), if \( h^{(1)}_{\neg\neg\vec{a},\neg\neg\vec{b}} \le h^{(2)}_{\neg\neg\vec{a},\neg\neg\vec{b}} \) for every \( \vec{a} \) and \( \vec{b} \), then \( h^{(1)} \) represents a morphism that is greater than the morphism represented by \( h^{(2)} \).\footnote{The converse does not hold unless the matrices \( h^{(i)} \), \( i = 1,2 \), are canonical.}

We first prove that \( h \le h' \) in \( \classicalCQ(M,N) \) implies \( {!}h \le {!}h' \).
This is trivial since
\begin{equation*}
    \neg\neg (\bra{b_1}h\ket{a_1} \ptensor \dots \ptensor \bra{b_k}h\ket{a_k}) \le \neg\neg (\bra{b_1}h'\ket{a_1} \ptensor \dots \ptensor \bra{b_k}h'\ket{a_k})
\end{equation*}
by the \( \SMon \)-enrichment of the composition, \( \ptensor \) and \( \neg \).
(Note that \( \SMon \)-enrichment implies monotonicity.)
So \( ! \) is monotone.

We then consider the sequence \( !f_0 \le !f_1 \le \dots \).
Since \( f_n = \sum_{i < n} g_i \), the above chosen matrix \( !f_n \) has entries
\begin{align*}
    (!f_n)_{\neg\neg(a_1\dots a_k),\neg\neg(b_1\dots b_k)}
    &=
    \textstyle
    \neg\neg(\bra{b_1}(\sum_{i<n} g_{i})\ket{a_1} \ptensor \dots \ptensor \bra{b_k}(\sum_{i<n}g_i)\ket{a_k})
    \\
    &\Kle
    \sum_{\varpi \colon \{1,\dots,k\} \to \{0,\dots,n-1\}} \neg\neg(\bra{b_1}g_{\varpi(1)}\ket{a_1} \ptensor \dots \ptensor \bra{b_k}g_{\varpi(k)}\ket{a_k}).
\end{align*}
Let \( h_n \) be a morphism defined by the matrix
\begin{equation*}
    (h_n)_{\neg\neg(a_1\dots a_k),\neg\neg(b_1\dots b_k)}
    \quad=\qquad
    \sum_{\mathclap{\substack{\varpi \colon \{1,\dots,k\} \to \{0,\dots,n\} \\ \exists i. \varpi(i)=n}}} \neg\neg(\bra{b_1}g_{\varpi(1)}\ket{a_1} \ptensor \dots \ptensor \bra{b_k}g_{\varpi(k)}\ket{a_k}).
\end{equation*}
We obtain \( h_n + {!}f_n = {!}f_{n+1} \) by comparing the entries of matrices.
Furthermore, since \( f = \sum_{i \in \Nat} g_i \), we have
\begin{align*}
    (!f)_{\neg\neg(a_1\dots a_k), \neg\neg(b_1\dots b_k)}
    &=
    \textstyle
    \neg\neg(\bra{b_1}(\sum_{i\in \Nat} g_{i})\ket{a_1} \ptensor \dots \ptensor \bra{b_k}(\sum_{i\in\Nat}g_i)\ket{a_k})
    \\
    &\Kle
    \sum_{\varpi \colon \{1,\dots,k\} \to \Nat} \neg\neg(\bra{b_1}g_{\varpi(1)}\ket{a_1} \ptensor \dots \ptensor \bra{b_k}g_{\varpi(k)}\ket{a_k}).
\end{align*}
Because \( (\{ 1,\dots,k \} \to \Nat) = \biguplus_{n \in \Nat} \{ \varpi \colon \{ 1,\dots,k \} \to \{0,\dots,n\} \mid \exists i. \varpi(i) = n \} \), we have
\begin{equation*}
    {!}f
    \quad=\quad
    \sum_{n \in \Nat} h_n.
\end{equation*}
Since \( {!}f_n = \sum_{i < n} h_i \), it is not difficult to see that \( {!}f = \bigvee_i {!}f_i \).

\subsection{Proof of \cref{thm:basis-of-types}}
\begin{claim*}[of \cref{thm:basis-of-types}]
    Let \( A = A(X_1,\dots,X_k) \) be a type with free variables.
    Assume \( \CQ \)-modules \( M_1,\dots,M_k \in \classicalCQ \) with orthogonal pseudo-representable bases \( \mathcal{B}_1,\dots,\mathcal{B}_k \), respectively.
    Then \( \Base{A}[\vec{\mathcal{B}}] \) is an orthogonal pseudo-representable basis for \( \sem{A}(M_1,\dots,M_k) \).
\end{claim*}

The orthogonality (\ie~\( \braket{b}{b'} = 0 \) for \( b,b' \in \Base{A}[\vec{\mathcal{B}}] \) with \( b \neq b' \)) can be easily proved by induction on the structure of \( b \in \Base{A}[\vec{\mathcal{B}}] \).
Note that induction on \( A \) does not work for the case of recursive types.

It is also easy to see that, under the assumption that \( \braket{b}{b} \neq 0 \) for \( b \in \mathcal{B}_i \), we have \( \braket{a}{a} \neq 0 \) for every \( a \in \Base{A}[\vec{\mathcal{B}}] \).
This proposition is again proved by induction on \( a \).
Hereafter we assume that an orthogonal basis in the sequel satisfies this condition.
Under this condition, \( b = b' \) if and only if \( \braket{b}{b'} \neq 0 \).

We prove the claim by induction on the structure of type \( A \).
We strengthen the claim as follows.
\begin{lemma}
    Let \( A = A(X_1,\dots,X_n) \) be a type with free variables.
    \begin{enumerate}
        \item
            Assume \( \CQ \)-modules \( M_1,\dots,M_k \in \classicalCQ \) with orthogonal pseudo-representable bases \( \mathcal{B}_1,\dots,\mathcal{B}_k \), respectively.
            Then \( \Base{A}[\vec{\mathcal{B}}] \) is an orthogonal pseudo-representable basis for \( \sem{A}(M_1,\dots,M_k) \).
        \item
            Assume \( \CQ \)-modules \( M_1,\dots,M_k,M'_1,\dots,M'_k \in \classicalCQ \) with orthogonal pseudo-representable bases \( \mathcal{B}_1,\dots,\mathcal{B}_k,\mathcal{B}'_1,\dots,\mathcal{B}'_k \), respectively.
            Let \( (e_i, p_i) \colon M_i \longrightarrow M'_i \) be an embedding-projection pair for each \( i \).
            Assume an injection \( \kappa_i \colon \mathcal{B}_i \longrightarrow \mathcal{B}'_i \) for each \( i \) such that, for each \( b \in \mathcal{B}_i \),
            \begin{align*}
                \BObj{b} &= \BObj{\kappa_i(b)} \\
                \bra{b}p_i &= \bra{\kappa_i(b)} \\
                e_i\ket{b} &= \ket{\kappa_i(b)}.
            \end{align*}
            Then there exists an injection \( \kappa \colon \Base{A}[\vec{\mathcal{B}}] \longrightarrow \Base{A}[\vec{\mathcal{B}}']\) such that
            \begin{align*}
                \BObj{b} &= \BObj{\kappa(b)} \\
                \bra{b}p &= \bra{\kappa(b)} \\
                e\ket{b} &= \ket{\kappa(b)}.
            \end{align*}
            for every \( b \in \Base{A}[\vec{\mathcal{B}}] \).
            Here \( (e,p) \defe \sem{A}((e_1,p_1),\dots,(e_k,p_k)) \).
    \end{enumerate}
\end{lemma}
\begin{proof}
    We prove the claim by induction on \( A \).
    Both claims trivially follow from the induction hypothesis, except for the case \( A = \mu X. B \) where \( B = B(X,\vec{Y}) \).
    We prove this case.

    We prove the first claim.
    Let \( \vec{M} \) be a sequence of objects in \( \classicalCQ \) and \( \vec{\mathcal{B}} \) be the corresponding sequence of orthogonal pseudo-representable bases.
    By the induction hypothesis, \( \sem{B}(0, \vec{M}) \) has an orthogonal pseudo-representable basis \( \Base{B}(\emptyset, \vec{\mathcal{B}}) \).
    We have the trivial embedding-projection pair \( (e_0,p_0) \defe (0,0) \colon 0 \longrightarrow \sem{B}(0, \vec{M}) \) with the trivial function \( \kappa_0 \colon \emptyset \longrightarrow \Base{B}(\emptyset,\vec{\mathcal{B}}) \).
    Let \( F \defe \sem{B}({-}, \vec{M}) \) and \( G \defe \Base{B}({-}, \vec{B}) \).
    So we have \( (e_0,p_0) \colon 0 \longrightarrow F(0) \) and \( \kappa_0 \colon \emptyset \longrightarrow G(\emptyset) \).
    The triple \( (e_0,p_0,\kappa_0) \) satisfies the condition on the second claim.
    Assume an embedding-projection pair \( (e_i,p_i) \colon F^i(0) \longrightarrow F^{i+1}(0) \) and \( \kappa_i \colon G^i(\emptyset) \longrightarrow G^{i+1}(\emptyset) \) that satisfy the condition on the second claim. 
    Let \( (e_{i+1},p_{i+1}) \defe F(e_i,p_i) \colon F^{i+1}(0) \longrightarrow F^{i+2}(0) \).
    By the second claim, we have \( \kappa_{i+1} \colon G^{i+1}(\emptyset) \longrightarrow G^{i+2}(\emptyset) \) that satisfies the condition on the second claim.
    Hence we have a family \( (e_i,p_i,\kappa_i)_{i \in \Nat} \) of triples that satisfy the condition on the second claim.
    By definition, \( \sem{A}(\vec{M}) \) is the colimit of the \( \omega \)-chain in \( \classicalCQep \)
    \begin{equation*}
        \xymatrix{
            0 \ar@<3pt>[r]^{e_0} \ar@<3pt>[dd]^{\mu_0^e} & F(0) \ar@<3pt>[r]^{e_1} \ar@<3pt>[l]^{p_0} \ar@<3pt>[ldd]^{\mu_1^e} & F(F(0)) \ar@<3pt>[r]^{e_2} \ar@<3pt>[lldd]^{\mu_2^e} \ar@<3pt>[l]^{p_1} & \cdots \ar@<3pt>[l]^{p_2}
            \\
            \\
            \sem{A}(\vec{M})
}.
    \end{equation*}
    We also have a diagram in \( \Set \)
    \begin{equation*}
        \xymatrix{
            \emptyset \ar[r]^{\kappa_0} & G(\emptyset) \ar[r]^{\kappa_1} & G(G(\emptyset)) \ar[r]^{\kappa_2} & \cdots 
        }.
    \end{equation*}
    Let \( (\mathcal{A}, (\eta_i \colon G^i(\emptyset) \longrightarrow \mathcal{A})_i) \) be the colimiting cocone.

    For each \( b \in G^i(\emptyset) \), we have 
    \begin{equation*}
        \BObj{b},
        \qquad
        \ket{b} \colon \BObj{b} \longrightarrow F^i(0)
        \quad\mbox{and}\quad
        \bra{b} \colon F^i(0) \longrightarrow \BObj{b}.
    \end{equation*}
    The above diagram gives
    \begin{align*}
        \ket{b}_A &\defe \mu_i^e \circ \ket{b} \quad\colon\quad \BObj{b} \longrightarrow F^i(0) \longrightarrow \sem{A}(\vec{M}) \\
        \bra{b}_A &\defe \bra{a} \circ \mu_i^p \quad\colon\quad \sem{A}(\vec{M}) \longrightarrow F^i(0) \longrightarrow \BObj{b}
    \end{align*}
    where \( \mu_i^p \) is the projection corresponding to \( \mu_i^e \).
    We have
    \begin{equation*}
        \ket{\kappa_i(b)}_A
        \quad=\quad
        \mu_{i+1}^e \circ \ket{\kappa_i(b)}
        \quad=\quad
        \mu_{i+1}^e \circ e_i \circ \ket{b}
        \quad=\quad
        \mu_i^e \circ \ket{b}
        \quad=\quad
        \ket{b}_A
    \end{equation*}
    and \( \bra{\kappa_i(b)}_A = \bra{b}_A \) and \( \BObj{\kappa_i(b)}_A = \BObj{b}_A \) by a similar argument.
    Hence \( \ket{-}_A \) and \( \bra{-}_A \) are well-defined on the colimit \( \mathcal{A} \).

    We prove that \( (\BObj{b}_A, \ket{b}_A, \bra{b}_A)_{b \in \mathcal{A}} \) is an orthogonal basis for \( \sem{A}(\vec{M}) \).
    To see that \( \sum_{b \in \mathcal{A}} \ket{b}_A \bra{b}_A \) converges, by the \( \omega \)-completeness of \( \classicalCQ(\sem{A}(\vec{M}), \sem{A}(\vec{M})) \), it suffices to show that \( \sum_{b \in I} \ket{b}_A \bra{b}_A \) converges for every finite subset \( I \subseteq_{\mathrm{fin}} \mathcal{A} \).
    Given a finite subset \( I \subseteq_{\mathrm{fin}} \mathcal{A} \), there exist \( \ell \) and a finite subset \( J \subseteq G^{\ell}(\emptyset) \) such that \( I = \{ \eta_i(b) \mid b \in J \} \) because of the finiteness of \( I \).
    Then \( \sum_{b \in J} \ket{b}\bra{b} \) converges as it is a partial sum of \( \sum_{b \in G^\ell(\emptyset)} \ket{b}\bra{b} = \ident_{F^\ell(0)} \).
    So
    \begin{equation*}
        \textstyle
        \mu^e_\ell \circ (\sum_{b \in J} \ket{b} \bra{b}) \circ \mu_i^p
        \quad\Kle\quad
        \sum_{b \in J} \mu^e_\ell \ket{b} \bra{b} \mu_i^p
        \quad=\quad
        \sum_{b \in I} \ket{b}_A \bra{b}_A
    \end{equation*}
    converges.
    Since
    \begin{equation*}
        \mu_i^e \mu_i^p
        \quad=\quad
        \mu_i^e (\sum_{b \in G^i(\emptyset)} \ket{b}\bra{b}) \mu_i^p
        \quad\Kle\quad
        \sum_{b \in G^i(\emptyset)} \mu_i^e \ket{b} \bra{b} \mu_i^p
        \quad\le\quad
        \sum_{b \in \mathcal{A}} \ket{b}_A \bra{b}_A 
    \end{equation*}
    holds for every \( i \) (here in the list inequation, we implicitly use the fact that \( \eta_i \) is injective; this fact follows from injectivity of \( \kappa_j \) for every \( j \)), we have 
    \begin{equation*}
        \ident_{\sem{A}(\vec{M})}
        \quad=\quad
        \bigvee_i (\mu_i^e \mu_i^p)
        \quad\le\quad
        \sum_{b \in \mathcal{A}} \ket{b}_A \bra{b}_A
    \end{equation*}
    where the first equation is a well-known characterisation of a colimit of an \( \omega \)-chain of embedding-projection pairs.
    To prove the converse, it suffices to shows that \( \sum_{b \in I} \ket{b}_A \bra{b}_A \le \ident_{\sem{A}(\vec{M})} \) because of the coincidence of the infinite sum and the least upper bound of finite partial sums (\cref{lem:omega-cpo-enrichment}).
    By the finiteness of \( I \), there exists \( \ell \) such that \( I \subseteq \eta_\ell(G^\ell(\emptyset)) \).
    So \( \sum_{b \in I} \ket{b}_A \bra{b}_A \le \mu^e_\ell \mu^p_\ell \le \ident_{\sem{A}(\vec{M})} \).
    
    We define \( \rho_\ell \colon G^{\ell}(\emptyset) \longrightarrow \Base{A}[\vec{\mathcal{B}}] \) by induction on \( n \) such that
    \begin{equation*}
        \BObj{b} = \BObj{(\rho_\ell(b))}
        \qquad
        \mu^e_\ell \ket{b} = \ket{\rho_\ell(b)}
        \quad\mbox{and}\quad
        \bra{b} \mu^p_\ell = \bra{\rho_\ell(b)}
    \end{equation*}
    for every \( b \in G^\ell(\emptyset) \).
    For \( \ell = 0 \), there exists a unique function \( \rho_0 \colon \emptyset \longrightarrow \Base{A}[\vec{\mathcal{B}}] \).
    For \( \ell > 0 \), by the second claim of the induction hypothesis, there exists \( \rho' \colon \Base{B}[G^\ell(\emptyset), \vec{\mathcal{B}}] \longrightarrow \Base{B}[(\Base{A}[\vec{\mathcal{B}}]), \vec{M}] \) such that
    \begin{equation*}
        \BObj{b} = \BObj{(\rho_\ell(b))}
        \qquad
        F^e(\mu^e_\ell, \mu^p_\ell) \ket{b} = \ket{\rho'(b)}
        \quad\mbox{and}\quad
        \bra{b} F^p(\mu^e_\ell, \mu^p_\ell) = \bra{\rho'(b)}
    \end{equation*}
    for every \( b \in \Base{B}[G^\ell(\emptyset), \vec{\mathcal{B}}] \).
    the size of \( a \in \Base{A}[\vec{\mathcal{B}}] \).
    By definition, \( a \in \Base{A}[\vec{\mathcal{B}}] \) implies \( a = \TFold(b) \) for some \( b \in \Base{B}[(\Base{A}[\vec{\mathcal{B}}]), \vec{\mathcal{B}}] \).
Since \( \sem{B}({-}, \vec{M}) \) preserves \( \omega \)-colimits in \( \classicalCQep \),
    \begin{equation*}
        \xymatrix{
            0 \ar@<3pt>[r]^{e_0} \ar@<3pt>[dd]^{0} & F(0) \ar@<3pt>[r]^{e_1} \ar@<3pt>[l]^{p_0} \ar@<3pt>[ldd]^{F(\mu_0)} & F(F(0)) \ar@<3pt>[r]^{e_2} \ar@<3pt>[lldd]^{F(\mu_1)} \ar@<3pt>[l]^{p_1} & \cdots \ar@<3pt>[l]^{p_2}
            \\
            \\
            F(\sem{A}(\vec{M}))
        }
    \end{equation*}
    is also a colimiting cocone and \( \mathit{fold} \) is the canonical isomorphism by definition.
    This means that \( (\mathit{fold}, \mathit{fold}^{-1}) \circ F(\mu^e_\ell, \mu^e_\ell) = (\mu^e_{\ell+1}, \mu^p_{\ell+1}) \) (where the composition is taken in \( \classicalCQep \)).
    Hence \( \rho_{\ell+1}(b) \defe \TFold(\rho'(b)) \) satisfies the requirement.
    
    For \( b \in G^{\ell}(\emptyset) \), we have \( \kappa_\ell(b) \in G^{\ell+1}(\emptyset) \) and hence \( \rho_\ell(b), \rho_{\ell+1}(\kappa_\ell(b)) \in \Base{A}[\vec{\mathcal{B}}] \).
    We show that \( \rho_\ell(b) = \rho_{\ell+1}(\kappa_\ell(b)) \).
    It suffices to show that \( \braket{\rho_\ell(b)}{\rho_{\ell+1}(\kappa_\ell(b))} \neq 0 \).
    \begin{align*}
        \braket{\rho_\ell(b)}{\rho_{\ell+1}(\kappa_\ell(b))}
        &=
        \bra{b}\mu^p_\ell \circ \mu^e_{\ell+1} e_\ell \ket{b}
        \\
        &=
        \braket{b}{b}
        \\
        &\neq
        0.
    \end{align*}
    This means that \( (\Base{A}[\vec{\mathcal{B}}], (\rho_\ell)_{\ell}) \) is a cocone of the diagram of basis.
    So there exists a canonical map \( \rho \colon \mathcal{A} \longrightarrow \Base{A}[\vec{\mathcal{B}}] \).

    This canonical map is injective.
    For \( b,b' \in \mathcal{A} \) with \( b \neq b' \), there exists \( \ell \) such that \( b = \rho_\ell(b_0) \) and \( b' = \rho_\ell(b_0') \) for some \( b_0,b_0' \in G^\ell(\emptyset) \) with \( b_0 \neq b_0' \).
    So \( \braket{b}{b'} = \braket{\rho_\ell(b_0)}{\rho_\ell(b_0')} = \bra{b_0}\mu^e_\ell \mu^p_\ell \ket{b_0'} = \braket{b_0}{b_0'} = 0 \) by the orthogonality of the basis \( G^\ell(\emptyset) \).

    We prove that the canonical map is surjective.
    Suppose that it is not the case.
    Let \( a_0 \in \Base{A}[\vec{\mathcal{B}}] \) be an element not in the image of \( \rho \).
    Note that \( \bra{\rho_\ell(b)} = \mu^e_\ell \bra{b} = \bra{\eta_\ell(b)}_A \) for every \( \ell \) and \( b \in G^\ell(\emptyset) \), and similarly \( \ket{\rho_\ell(b)} = \ket{\eta_\ell(b)}_A \).
    Hence \( \bra{\rho(a)} = \bra{a}_A \) and \( \ket{\rho(a)} = \ket{a}_A \) for every \( a \in \mathcal{A} \).
    We know \( \sum_{a \in \mathcal{A}} \ket{a}_A \bra{a}_A = \ident \) but
    \begin{equation*}
        \braket{a_0}{a_0}
        \quad=\quad
        \bra{a_0} (\sum_{a \in \mathcal{A}} \ket{a}_A \bra{a}_A) \ket{a_0}
        \quad=\quad
        \bra{a_0} (\sum_{a' \in \rho(\mathcal{A})} \ket{a'} \bra{a'}) \ket{a_0}
        \quad=\quad
        0
    \end{equation*}
    since \( a_0 \notin \rho(\mathcal{A}) \) and the orthogonality of \( \Base{A}[\vec{\mathcal{B}}] \).
    This contradicts the assumption \( \braket{a_0}{a_0} \neq 0 \).

    The second claim can be easily proved by appearing to the fact that \( \mathcal{A} \) coincides with \( \Base{A}[\vec{\mathcal{B}}] \).
    Let \( G' \defe \Base{B}[{-}, \vec{\mathcal{B}}'] \) and define \( \kappa'_i \colon (G')^i(\emptyset) \longrightarrow (G')^{i+1}(\emptyset) \) similarly to \( (\kappa_i)_i \).
    Then we have
    \begin{equation*}
        \xymatrix{
            \emptyset \ar[d] \ar[r]^{\kappa_0} & G(\emptyset) \ar[d] \ar[r]^{\kappa_1} & G(G(\emptyset)) \ar[d] \ar[r]^{\kappa_2} & \cdots 
            \\
            \emptyset \ar[r]^{\kappa'_0} & G'(\emptyset) \ar[r]^{\kappa'_1} & G'(G'(\emptyset)) \ar[r]^{\kappa'_2} & \cdots 
            }
    \end{equation*}
    where the vertical morphisms are the iterative application of \( \sem{A}({-}, (e_1,p_1),\dots,(e_k,p_k)) \) to \( \ident_0 \in \classicalCQep(0,0) \).
    Each square in the above diagram commutes.
    This fact can be proved by using \( \braket{b}{b'} \neq 0 \Leftrightarrow b = b' \).
    So the above diagram gives the canonical map \( \mathcal{A} = \mathord{\mathrm{colim}}_{\ell} G^{\ell}(\emptyset) \longrightarrow \mathord{\mathrm{colim}}_{\ell} (G')^{\ell}(\emptyset) = \mathcal{A}' \), which is a function with the desired property.
\end{proof}

 \section{Supplementary Materials for Section~8}
\tk{\cref{sec:norm}}

\subsection{$\fdCQ$ and Pseudo-Representable Module}
We prove the following proposition, which is used in \cref{sec:family-norms}.
\begin{proposition}\label{prop:finite-dimesnional-idempotent-splitting}
    Every \( M \in \fdCQ \) is a splitting object of an idempotent \( f \colon \LQT \longrightarrow \LQT \) of a pseudo-representable module \( \LQT \).
\end{proposition}
\begin{proof}
    Assume \( M \in \fdCQ \).
    Then \( M \) has a pseudo-representable basis \( (\BObj{b}, \ket{b}, \bra{b})_{b \in \Base{M}} \) indexed by a finite set \( \Base{M} \).
    Let \( b_1,\dots,b_k \) be an enumeration of elements in \( \Base{M} \), \( \ell_i \defe \#\BObj{b_i} \) and \( \ell \defe \sum_{i=1}^k \ell_i \).
    We define linear maps \( p_i \colon \Mat_{\ell} \longrightarrow \Mat_{\ell_i} \) and \( e_i \colon \Mat_{\ell_i}(\Complex) \longrightarrow \Mat_{\ell}(\Complex) \) by
    \begin{equation*}
        \left(
            \begin{matrix}
                X_{1,1} & \cdots & X_{1,i} & \cdots & X_{1,k} \\
                \vdots & \ddots & \vdots & & \vdots \\
                X_{i,1} & \cdots & X_{i,i} & \cdots & X_{i,k} \\
                \vdots &  & \vdots & \ddots & \vdots \\
                X_{k,1} & \cdots & X_{k,i} & \cdots & X_{k,k}
            \end{matrix}
        \right)
        \quad\stackrel{p_i}{\mapsto}\quad
        X_{i,i}
        \quad\stackrel{e_i}{\mapsto}\quad
        \left(
            \begin{matrix}
                0 & \cdots & 0 & \cdots & 0 \\
                \vdots & \ddots & \vdots & & \vdots \\
                0 & \cdots & X_{i,i} & \cdots & 0 \\
                \vdots &  & \vdots & \ddots & \vdots \\
                0 & \cdots & 0 & \cdots & 0
            \end{matrix}
        \right)
    \end{equation*}
    where \( X_{i,j} \) is an \( (\ell_i \times \ell_j) \)-matrix.
    It is not difficult to see that \( p_i \) and \( e_i \) are completely positive: \( p_i \in \CPM(\ell,\ell_i) \) and \( e_i \in \CPM(\ell_i,\ell) \).
    Let \( \LQT \hookrightarrow \CPM({-}, \ell) \) be the hereditary submodule given by
    \begin{equation*}
        \LQT_n
        \quad\defe\quad
        \{\, x \in \CPM(n,\ell) \mid \forall i. p_i \circ x \in \BObj{b_i},\: M \models \IsDef{(\sum_i \ket{b_i} \cdot (p_i \circ x))} \,\}.
    \end{equation*}
    We have
    \begin{align*}
        \bra{\ast} \defe \sum_i e_i \circ \bra{b_i} &\colon M \longrightarrow \LQT \\
        \ket{\ast} \defe \sum_i \ket{b_i} \circ p_i &\colon \LQT \longrightarrow M.
    \end{align*}
    Then
    \begin{align*}
        \ket{\ast}\bra{\ast}
        &\quad=\quad
        (\sum_i \ket{b_i} \circ p_i) \circ (\sum_{j} e_j \circ \bra{b_j})
        \\ &\quad\Kle\quad
        \sum_{i,j} \ket{b_i} \circ p_i \circ e_j \circ \bra{b_j}
        \\ & \quad=\quad
        \sum_i \ket{b_i} \bra{b_i}
        \\ & \quad=\quad
        \ident_M.
    \end{align*}
    So \( M \) is a splitting object of the idempotent \( \braket{\ast}{\ast} \) on \( \LQT \).

    We prove that \( \LQT \) is pseudo-representable.

    We give an upper bound of the norm in \( \LQT \).
    Assume \( B_i \) be an upper bound of the norm in \( \BObj{b_i} \).
    We prove that \( B \defe \sum_i B_i \) is an upper bound of the norm in \( \LQT \).
    Assume \( x \in \LQT_n \).
    Then \( p_i \circ x \in \BObj{b_i} \) by definition.
    So \( \opnorm{\sum_i p_i \circ x} \le \sum_i \opnorm{p_i \circ x} \le B \).
    Since \( \sum_i p_i \in \CPM(\ell,\ell) \) preserves the main diagonal and the trace, we have \( \opnorm{x} = \opnorm{\sum_i p_i \circ x} \le B \).

    We prove that \( r \ident_\ell \in \LQT_\ell \) for some \( r > 0 \).
    Let \( r_i > 0 \) be a real number such that \( r_i \ident_{\ell_i} \in \BObj{b_i} \) for some \( r_i > 0 \).
    Let \( r' = \min_i r_i \).
    Since \( p_j \circ (r' e_i) = r'\ident_{\ell_i} \) if \( i = j \) and otherwise \( 0 \), we have \( r' e_i \in \LQT_{\ell_i} \) for each \( i \).
    Hence \( r' e_i p_i \in \LQT_{\ell} \) since \( p_i \in \CQ(\ell,\ell_i) \).
    By the convexity of \( M \), which follows from \( M \in \classicalCQ \), we have \( (\sum_i (1/k) r' e_i p_i) \in \LQT_{\ell} \).
    Hence \( \sum_{j} \ket{b_j} \cdot (p_j \circ (\sum_i (1/k) r' e_i p_i)) \) is defined.
    So
    \begin{equation*}
        \sum_{j} \ket{b_j} \cdot (p_j \circ (\sum_i (1/k) r' e_i p_i))
        \quad\Kle\quad
        \sum_{i,j} \ket{b_j} \cdot (r'/k) p_j e_i p_i
        \quad=\quad
        \sum_i \ket{b_i} \cdot p_i (r'/k)
    \end{equation*} 
    is defined.
    Hence \( (r'/k) \ident_{\ell} \in \LQT_\ell \).
\end{proof}

\subsection{Proof of \cref{lem:linear:property-of-induced-norm}}
\begin{claim*}[of \cref{lem:linear:property-of-induced-norm}]
    Let \( \LQT \) be a pseudo-representable \( \CQ \)-module such that \( \LQT \cong \neg\neg\LQT \) and
    \begin{equation*}
        \qnorm{x}{\LQT}^{(n)}
        \quad\defe\quad
        \inf \left\{\, r \in \Real_{\ge 0} ~\middle|~ x \in r \cdot \LQT_n \,\right\}
    \end{equation*}
    where \( r \cdot X \defe \{ r x \mid x \in X \} \) for \( X \subseteq \CPM(m,\ell) \).
    Then the family \( (\qnorm{-}{\LQT}^{(n)})_n \) satisfies the following conditions.
    \begin{enumerate}
        \item \( \qnorm{{-}}{\LQT}^{(n)} \) is a norm on \( \CPM(n,\#\LQT) \) for each \( n \).
        \item There exists \( B \) such that \( \opnorm{x} \le B \qnorm{x}{\LQT}^{(n)} \) for every \( n \) and \( x \in \CPM(n,\#\LQT) \).
        \item \( \CQ \)-action is norm-non-increasing: \( \forall x \in \CPM(n,\#\LQT). \forall \varphi \in \CQ(m,n).\: \qnorm{x \circ \varphi}{\LQT}^{(m)} \le \qnorm{x}{\LQT}^{(n)} \).
        \item
For every \( x \in \CPM(n,\ell) \),
\begin{align*}
            \qnorm{x}{\LQT}^{(n)} &= \sup \left\{\, \varphi \circ (\ident_m \otimes x) \circ \psi ~~\middle|~~
            \begin{aligned}
                m,k \in \Nat,\:
                \varphi \in \CPM(m \otimes \ell, 1),\:
                \psi \in \CPM(1, m \otimes k) \\
                \forall y \in \LQT_k. \varphi \circ (\ident_m \otimes y) \circ \psi \le 1
            \end{aligned}
            \right\}.
            \qed
\end{align*}
\end{enumerate}
\end{claim*}

(1)
Let \( r \in \Real_{>0} \) such that \( r\,\ident_\ell \in \LQT_n \).
Then \( x \in \frac{\opnorm{x}}{r} \LQT_n \) for every \( x \in \LQT_n \), so \( \qnorm{x}{\LQT}^{(n)} \le \frac{\opnorm{x}}{r} < \infty \).
The pseudo-representable \( \CQ \)-module \( \LQT \) has an upper bound of the norm, \ie~\( B \in \Real_{>0} \) such that \( \opnorm{x} \le B \) for every \( n \) and \( x \in \LQT_n \).
So \( \opnorm{y}/B \le \qnorm{y}{\LQT}^{(n)} \) for every \( y \in \CPM(n,\ell) \) and thus \( \qnorm{x}{\LQT}^{(n)} = 0 \) implies \( \opnorm{x}=0 \) and \( x = 0 \).
Trivially \( \qnorm{t\,x}{\LQT}^{(n)} = t\,\qnorm{x}{\LQT}^{(n)} \) for every \( t \in \Real_{\ge 0} \).
Since \( \neg\neg\LQT \cong \LQT \), it is downward-closed (\ie~\( x \in \LQT_n \) and \( y \le x \) implies \( y \in \LQT_n \)) and convex (\ie~\( x,y \in \LQT_n \) implies \( px + (1-p)y \in \LQT_n \) for every \( p \in [0,1] \)).
We have that \( x \le y \) implies \( \qnorm{x}{\LQT}^{(n)} \le \qnorm{y}{\LQT}^{(n)} \) from the former and the triangular inequality from the latter.

(2)
Already show in the above argument.

(3)
Easy.

(4)
Since \( \LQT = \neg\neg\LQT \) by the assumtion, we have \( \qnorm{x}{\LQT} = \qnorm{x}{\neg\neg\LQT} \).
The claim follows from the direct computation by expanding \( \neg\neg\LQT \) using \cref{lem:linear:lqt-function-and-tensor,lem:negation-of-pseudo-representable}.

\subsection{Proof of \cref{thm:norm-category-equivalent}}
\begin{claim*}[of \cref{thm:norm-category-equivalent}]
    \( \CQ^{\P} \) is equivalent to \( \fdCQ \).
\end{claim*}

Let \( \odcCQ \) be the full subcategory of \( \classicalCQ \) consisting of pseudo-representable modules \( \LQT \) such that \( \neg\neg \LQT \cong \LQT \).
By \cref{prop:finite-dimesnional-idempotent-splitting}, the category \( \fdCQ \) is equivalent to the Karoubi envelope of \( \odcCQ \).
Let \( \CQ^{\P\P} \) be the full subcategory of \( \CQ^{\P} \) consisting of objects \( (\ell, \varphi, (\annorm{-}^{(n)})_n)) \) such that \( \varphi \in \CPM(\ell,\ell) \) is the identity.
Obviously \( \CQ^{\P} \) is the Karoubi envelope of \( \CQ^{\P\P} \).
Since the Karoubi envelope preserves equivalence of categories, it suffices to prove the equivalence of \( \odcCQ \) and \( \CQ^{\P\P} \).

For simplicity, we omit the (useless) idempotent component \( \varphi \) of an object in \( \CQ^{\P\P} \).
\begin{lemma}
    \( \odcCQ \) is equivalent to \( \CQ^{\P\P} \).
\end{lemma}
\begin{proof}
    We first give the bijective correspondence between objects.
    The norm of \( \LQT \in \odcCQ \) has already given.
    Assume a family \( (\annorm{-}^{(n)})_n \) of cone norms on \( \CPM({-}, \ell) \) that satisfies the conditions in \cref{lem:linear:property-of-induced-norm}.
    It defines a hereditary \( \CQ \)-submodule \( \LQT \hookrightarrow \CPM({-}, \ell) \) by
    \begin{equation*}
        \LQT_n
        \quad\defe\quad
        \{\, x \in \CPM(n,\ell) \mid \annorm{x}^{(n)} \le 1 \,\}.
    \end{equation*}
    This is obviously pseudo-representable: the bound \( B \) comes from Condition (2), \( \frac{1}{\annorm{\ident_\ell}^{(\ell)}}\ident_{\ell} \in \LQT_\ell \), and Condition (3) ensures that \( (\LQT_n)_n \) is closed under the \( \CQ \)-action in \( \CPM({-}, \ell) \).
    Condition (4) is equivalent to \( \neg\neg \LQT = \LQT \).

    The mappings on objects are clearly the inverses of each other.

    The action of functors on morphisms can be defined as the identity when we identify \( \pshCQ(\LQT, \LQT') \) with the subset of \( \CPM(\ell, \ell') \) by \cref{thm:cpm-representation}.
    The condition on the subset of \( \CPM(\ell,\ell') \) in \cref{thm:cpm-representation} is clearly equivalent to the norm-non-increasing condition for the corresponding norms.    
\end{proof}

\subsection{Proof of \cref{thm:selingers-q-prime}}
\begin{claim*}[of \cref{thm:selingers-q-prime}]
    \( \CQ' \) is isomorphic to a full subcategory of \( \CQ^{\P} \).
    The embedding is strong monoidal.
\end{claim*}

Since \( \CQ^{\P} \) is equivalent to \( \fdCQ \), it suffices to prove the result for \( \fdCQ \) instead of \( \CQ^{\P} \).
Recall that \( \CQ' \) is the Karoubi envelope of the full subcategory \( \CQ'' \) consisting of object whose idempotent is the identity.
The Karoubi envelope preserves a full subcategory.
Furthermore the Karoubi envelope of a monoidal category has a canonical monoidal structure, and the Karoubi envelope maps a strong monoidal functor to a strong monoidal functor.
Hence it suffices to prove the result for \( \CQ'' \) and \( \odcCQ \).
(Recall that \( \odcCQ \) is the full subcategory of \( \classicalCQ \) consisting of pseudo-representable modules \( \LQT \) such that \( \neg\neg \LQT \cong \LQT \).)

Given a norm \( \annorm{-} \) on \( \CPM(1,\ell) \), let \( \LQT_{\annorm{-}} \) be \( \{ y \in \CPM(1,\ell) \mid \annorm{y} \le 1 \}^{\bot\bot} \).
\begin{lemma}
    \( \LQT_{\annorm{-}} \) is a pseudo-representable module for every norm \( \annorm{-} \) on \( \CPM(1,\ell) \).
\end{lemma}
\begin{proof}
    Since \( \annorm{-} \) is a norm on \( \CPM(1,\ell) \) and all norms on \( \CPM(1,\ell) \) is equivalent,\footnote{This is because a norm on \( \CPM(1,\ell) \) can be extended to a norm on the finite-dimensional \( \Real \)-vector space \( \SelfAdjoint(\Mat_\ell(\Complex)) \) and all norms on a finite-dimensional vector space are equivalent.}
    we have a real number \( B_1,B_2 > 0 \) such that \( \annorm{x} \le B_1 \tracenorm{x} \) and \( \tracenorm{x} \le B_2 \annorm{x} \) for every \( x \in \CPM(1,\ell) \).

    We show that \( \LQT_{\annorm{-}} \) is bounded.
    Let \( \mathcal{F} \defe B_2 \CQ({-}, \ell) \) be the family consisting of \( x \in \CPM(n,\ell) \) with \( \opnorm{x} \le B_2 \).
    Then \( \{ y \in \CPM(1,\ell) \mid \annorm{y} \le 1 \} \subseteq \mathcal{F} \).
    So \( \mathcal{L} = \{ y \in \CPM(1,\ell) \mid \annorm{y} \le 1 \}^{\bot\bot} \subseteq \mathcal{F}^{\bot\bot} \).
    Since \( \mathcal{F} \) is pseudo-representable, so is \( \mathcal{F}^{\bot\bot} = \neg\neg\mathcal{F} \) (\cref{lem:negation-of-pseudo-representable}).
    Hence \( \mathcal{L} \) is bounded.

    Let \( \mathcal{F}_1 \defe \{ y \in \CPM(1,\ell) \mid \annorm{y} \le 1 \}^{\bot} \).
    \begin{claim*}
        There exists a real number such that \( \varphi \in \mathcal{F}_1 \) implies \( \opnorm{\varphi} \le C_1 \)
        for every \( n \) and \( \varphi \in \CPM(n,\neg\ell) \).
    \end{claim*}
    \begin{proof}
        Assume that \( \varphi \in \mathcal{F}_1 \).
        Then \( \opnorm{\mathbf{ev} \circ (\varphi \otimes x)} \le 1 \) for every \( x \in \CPM(1,\ell) \) with \( \annorm{x} \le 1 \).
        So \( \opnorm{\mathbf{ev} \circ (\varphi \otimes (1/B_1)\psi)} \le 1 \) for every \( \psi \in \CQ(1,\ell) \) since \( \annorm{\psi} \le B_1 \tracenorm{\psi} = B_1 \opnorm{\psi} \).
        Since \( \mathbf{ev} \circ (\varphi \otimes x) = \Lambda^{-1}(\varphi) \circ x \), where \( \Lambda^{-1} \colon \CPM(n, \neg \ell) \cong \CPM(n \otimes \ell, 1) \), we have \( \opnorm{\Lambda^{-1}(\varphi) \circ (1/B_1)\psi} \le 1 \) for every \( \psi \in \CQ(1,\ell) \).
        This means that \( \opnorm{\Lambda^{-1}(\varphi)} \le B_1 \).
        By the same argument as the proof of \cref{lem:appx:lqt-function1}, we have some \( C' > 0 \) such that \( C' \opnorm{\varphi} \le \opnorm{\Lambda^{-1}(\varphi)} \) for every \( \varphi \in \CPM(n,\neg\ell) \).
        Hence \( \opnorm{\varphi} \le (1/C')\opnorm{\Lambda^{-1}(\varphi)} \le B_1/C' \).
        \qedhere\textbf{[Claim]}
    \end{proof}

    Since the operator norm of \( \mathcal{F}_1 \defe \{ y \in \CPM(1,\ell) \mid \annorm{y} \le 1 \}^{\bot} \) is bounded, \( \{ y \in \CPM(1,\ell) \mid \annorm{y} \le 1 \}^{\bot\bot} \) contains all completely positive maps with sufficiently small operator norm.
    In particular, \( r \ident_\ell \in \{ y \in \CPM(1,\ell) \mid \annorm{y} \le 1 \}^{\bot\bot} \) for some \( r > 0 \).
\end{proof}

\begin{corollary}
    \( \LQT_{\annorm{-}} \in \odcCQ \).
\end{corollary}
\begin{proof}
    We have shown that \( \LQT_{\annorm{-}} \) is pseudo-representable.
    By \cref{lem:negation-of-pseudo-representable}, \( \neg\neg\LQT_{\annorm{-}} \cong \LQT_{\annorm{-}} \). 
\end{proof}

\( \annorm{-} \mapsto \LQT_{\annorm{-}} \) gives the action on objects of the functor \( \CQ'' \longrightarrow \odcCQ \).
The action on morphisms can be defined as the identity when we identify \( \pshCQ(\LQT, \LQT') \) with the subset of \( \CPM(\ell, \ell') \) by \cref{thm:cpm-representation}.\tk{Here it would be better to give some details, which can be found in my hand-written note.}
It is easy to see the fullness and faithfulness of the functor.

The next lemma shows that the functor is strong monoidal.
\begin{lemma}\label{lem:entanglement-free-tensor}
    Let \( \annorm{-}_i \) be a cone norm on \( \CPM(1,\ell_i) \) for \( i = 1,2 \) and \( \annorm{-} \) be the tensor norm on \( \CPM(1,\ell_1 \otimes \ell_2) \).
    Then \( (\LQT_{\annorm{-}_1} \otimes \LQT_{\annorm{-}_2}) = \LQT_{\annorm{-}} \).
\end{lemma}
\begin{proof}
    We first prove the following claim.
    \begin{claim*}
Let \( k \in \Nat \) and \( f \in \CPM(k \otimes \ell_1 \otimes \ell_2, 1) \) be a completely positive map such that \( \opnorm{f \circ (\ident_k \otimes x)} \le 1 \) for every \( x \in \CPM(1, \ell_1 \otimes \ell_2) \) with \( \annorm{x} \le 1 \).
        Then, for every \( m_1,m_2 \in \Nat \), \( y_1 \in \CPM(m_1, \ell_1) \) and \( y_2 \in \CPM(m_2, \ell_2) \) with \( \annorm{y_1}_1 \le 1 \) and \( \annorm{y_2}_2 \le 1 \), we have
        \( \opnorm{f \circ (\ident_k \otimes y_1 \otimes y_2)} \le 1 \).
    \end{claim*}
    \begin{proof}
        Let \( y'_1 \in \CPM(1, \ell_1) \) and \( y'_2 \in \CPM(1, \ell_2) \) be arbitrary elements such that \( \annorm{y'_1}_1 \le 1 \) and \( \annorm{y'_2}_2 \le 1 \).
        Then
        \begin{equation*}
            \annorm{y'_1 \otimes y'_2} \le \annorm{y'_1}_1 \annorm{y'_2}_2 \le 1.
        \end{equation*}
        Hence
        \begin{equation*}
            \opnorm{f \circ (\ident_k \otimes y'_1 \otimes y'_2)} \le 1
        \end{equation*}
        by the assumption on \( f \).
        Let \( g[y'_2] \in \CPM(k \otimes \ell_1, 1) \) be given by
        \begin{equation*}
            g[y'_2]
            \quad\defe\quad
            f \circ (\ident_k \otimes \ident_{\ell_1} \otimes y'_2).
        \end{equation*}
        Then \( g[y'_2] \circ (\ident_k \otimes y'_1) = f \circ (\ident_k \otimes y'_1 \otimes y'_2) \) and thus
        \begin{equation*}
            \opnorm{g[y'_2] \circ (\ident_k \otimes y'_1)} \le 1.
        \end{equation*}
        Since \( y'_1 \in \CPM(1, \ell_1) \) is an arbitrary element such that \( \annorm{y'_1}_1 \le 1 \), by definition of \( \LQT_{\annorm{-}_1} \) and \( y_1 \in (\LQT_{\annorm{-}_1})_{m_1} \), we have
        \begin{equation*}
            \opnorm{f \circ (\ident_k \otimes y_1 \otimes y'_2)}
            =
            \opnorm{g[y'_2] \circ (\ident_k \otimes y_1)} \le 1.
        \end{equation*}
        Since \( y'_2 \) is an arbitrary element satisfying \( \annorm{y'_2}_2 \le 1 \), by the same argument applied to \( g' \defe f \circ (\ident_{k} \otimes y_1 \otimes \ident_{\ell_2}) \in \CPM((k \otimes m_1) \otimes \ell_2, 1) \), we have
        \begin{equation*}
            \opnorm{f \circ (\ident_k \otimes y_1 \otimes y_2)}
            =
            \opnorm{g' \circ (\ident_{k\otimes m_1} \otimes y_1)} \le 1.
        \end{equation*}
        \qedhere\textbf{[Claim]}
    \end{proof}
    
    We first prove that they coincide on \( 1 \)-st component, \ie,
    \begin{equation*}
        \qnorm{-}{\neg\neg(\LQT_{\annorm{-}_1} \ptensor \LQT_{\annorm{-}_2})}^{(1)}
        \quad=\quad
        \annorm{-}.
    \end{equation*}
    Let \( x \in \CPM(1, \ell_1 \otimes \ell_2) \).
    
    Assume that \( \annorm{x} < 1 \).
    Then \( x \le \sum_{i=1}^N p_i y_{1,i} \otimes y_{2,i} \) and \( \sum_i p_i \annorm{y_{1,i}}_1 \annorm{y_{2,i}} < 1 \).
    We can assume without loss of generality that \( \annorm{y_{1,i}}_1 = \annorm{y_{2,i}} = 1 \).
    Then \( y_{1,i} \otimes y_{2,i} \in (\LQT_{\annorm{-}_1} \otimes \LQT_{\annorm{-}_2})_1 \), which implies \( y_{1,i} \otimes y_{2,i} \in \neg\neg(\LQT_{\annorm{-}_1} \otimes \LQT_{\annorm{-}_2}) \).
    Since \( \sum_i p_i < 1 \), by the convexity and downward-closedness of \( \neg\neg(\LQT_{\annorm{-}_1} \otimes \LQT_{\annorm{-}_2}) \), we have \( \sum_{i=1}^N p_i y_{1,i} \otimes y_{2,i} \in \neg\neg(\LQT_{\annorm{-}_1} \otimes \LQT_{\annorm{-}_2}) \) and \( x \in \neg\neg(\LQT_{\annorm{-}_1} \otimes \LQT_{\annorm{-}_2})_1 \).
    
    Since \( \neg\neg(\LQT_{\annorm{-}_1} \otimes \LQT_{\annorm{-}_2}) \) is a topologically closed set and any norm on a finite dimensional vector space is continuous, the above result can be extended to \( \annorm{x} \le 1 \).
    
    Assume that \( \annorm{x} > 1 \).
    By the Hahn-Banach theorem for \emph{continuous normed cone}~\cite[Theorem~2.14]{Selinger2004a}, there exists a \( \Real_{\ge 0} \)-linear function \( f \colon \CPM(1, \ell_1 \otimes \ell_2) \longrightarrow \Real_{\ge 0} \) such that \( f(x') \le 1 \) for \( \annorm{x'} \le 1 \) and \( f(x) > 1 \).
    Note that any \( \Real_{\ge 0} \)-linear function \( f \colon \CPM(1, \ell_1 \otimes \ell_2) \longrightarrow \Real_{\ge 0} \) can be extended to a completely positive map \( f \in \CPM(\ell_1 \otimes \ell_2, 1) \).
    For this \( f \), \( \opnorm{f \circ (y_1 \otimes y_2)} \le 1 \) for every \( y_1 \in \LQT_{\annorm{-}_1} \) and \( y_2 \in \LQT_{\annorm{-}_2} \) (by the above claim), but \( \opnorm{f \circ x} > 1 \).
    This means that \( x \notin \neg\neg(\LQT_{\annorm{-}_1} \otimes \LQT_{\annorm{-}_2})_1 \).

Let \( \mathcal{F} \subseteq \CPM({-}, \ell_1 \otimes \ell_2) \) be the family given by
    \begin{equation*}
        \mathcal{F}_1 \defe \{\, x \in \CPM(1, \ell_1 \otimes \ell_2) \mid \annorm{x} \le 1 \,\}
    \end{equation*}
    and \( \mathcal{F}_n \defe \emptyset \) for \( n \neq 0 \).
    For \( x \in \mathcal{F}_1 \), we have \( 1 \ge \annorm{x} \ge \qnorm{x}{\neg\neg(\LQT_{\annorm{-}_1} \otimes \LQT_{\annorm{-}_2})}^{(1)} \), and hence \( x \in \neg\neg(\LQT_{\annorm{-}_1} \otimes \LQT_{\annorm{-}_2})_1 \).
    Since \( \mathcal{F}_n = \emptyset \) for \( n \neq 0 \), we have \( \mathcal{F} \subseteq \neg\neg(\LQT_{\annorm{-}_1} \otimes \LQT_{\annorm{-}_2}) \).
    By \cref{lem:negation-of-pseudo-representable},
    \begin{equation*}
        \LQT_{\annorm{-}}
        \quad=\quad
        \mathcal{F}^{\bot\bot}
        \quad\subseteq\quad
        (\neg\neg(\LQT_{\annorm{-}_1} \otimes \LQT_{\annorm{-}_2}))^{\bot\bot}.
    \end{equation*}
    Again by \cref{lem:negation-of-pseudo-representable},
    \begin{align*}
        (\neg\neg(\LQT_{\annorm{-}_1} \otimes \LQT_{\annorm{-}_2}))^{\bot\bot}
        &\quad=\quad
        \neg\neg\neg\neg(\LQT_{\annorm{-}_1} \otimes \LQT_{\annorm{-}_2})
        \\
        &\quad=\quad
        \neg\neg(\LQT_{\annorm{-}_1} \otimes \LQT_{\annorm{-}_2}).
    \end{align*}
    Hence
    \begin{equation*}
        \LQT_{\annorm{-}}
        \quad\subseteq\quad
        \neg\neg(\LQT_{\annorm{-}_1} \otimes \LQT_{\annorm{-}_2}).
    \end{equation*}
    
    Assume that \( x \in (\LQT_{\annorm{-}_1} \otimes \LQT_{\annorm{-}_2})_n \).
    By \cref{lem:linear:lqt-function-and-tensor},
    \begin{equation*}
        x \quad=\quad (y_1 \otimes y_2) \cdot h
    \end{equation*}
    for some \( y_1 \in (\LQT_{\annorm{-}_1})_{m_1} \), \( y_2 \in (\LQT_{\annorm{-}_2})_{m_2} \) and \( h \in \CQ(n, m_1 \otimes m_2) \).
    Let \( f \in \CPM(k \otimes \ell_1 \otimes \ell_2, 1) \) be a completely positive map such that \( \annorm{x'} \le 1 \Longrightarrow \opnorm{f \circ (\ident_k \otimes x')} \le 1 \) for every \( x' \in \CPM(1, \ell_1 \otimes \ell_2) \).
    Then \( \opnorm{f \circ (\ident_k \otimes y_1 \otimes y_2)} \le 1 \) and
    \begin{equation*}
        \opnorm{f \circ (\ident_k \otimes x)}
        =
        \opnorm{f \circ (\ident_k \otimes y_1 \otimes y_2) \circ (\ident_k \otimes h)}
        \le 1        
    \end{equation*}
    since \( \ident_k \otimes h \in \CQ(k \otimes n, k \otimes m_1 \otimes m_2) \).
    Since \( f \in \CPM(k \otimes \ell_1 \otimes \ell_2, 1) \) is an arbitrary completely positive map such that \( \annorm{x'} \le 1 \Longrightarrow \opnorm{f \circ (\ident_k \otimes x')} \le 1 \) for every \( x' \in \CPM(1, \ell_1 \otimes \ell_2) \), by definition of \( \LQT_{\annorm{-}} \), we have \( x \in \LQT_{\annorm{-}} \).
    Therefore we obtain
    \begin{equation*}
        (\LQT_{\annorm{-}_1} \otimes \LQT_{\annorm{-}_2})
        \quad\subseteq\quad
        \LQT_{\annorm{-}},
    \end{equation*}
    which implies
    \begin{equation*}
        \neg\neg(\LQT_{\annorm{-}_1} \otimes \LQT_{\annorm{-}_2})
        \quad\subseteq\quad
        \neg\neg\LQT_{\annorm{-}}
        \quad=\quad
        \LQT_{\annorm{-}}.
    \end{equation*}    
\end{proof}
 \fi

\end{document}